%% file: main_paper.tex
\documentclass[12pt, reqno]{amsart}
\usepackage[utf8]{inputenc}
\usepackage{amsmath}
\usepackage{amssymb}
\usepackage{esint}
\usepackage[authoryear]{natbib}
\usepackage{longtable}
\usepackage{multirow}
\usepackage[hypertexnames=false]{hyperref}
\usepackage{booktabs}
\usepackage{rotating,capt-of}
\usepackage{float}
\usepackage{bm}
\usepackage{appendix}
\usepackage{soul}
\usepackage{enumitem}
\usepackage{tablefootnote}

\makeatletter

\usepackage{amsfonts}
\usepackage[english]{babel}
\usepackage{graphicx}
\usepackage{color}
\usepackage{colordvi}
\usepackage{multirow}
\usepackage{multicol}
\usepackage{babel}
\usepackage{amsthm}
\input{isesv.tex}

\newcommand{\vps}{{\varepsilon}}
\newtheorem{Assumption}{Assumption}[section]
\newtheorem{Theorem}{Theorem}[section]
\newtheorem{Corollary}{Corollary}[section]
\newtheorem{Lemma}{Lemma}[section]
\newtheorem{Proposition}{Proposition}[section]

\newtheorem{Definition}{Definition}[section]

\def\defeq{\stackrel{\mathrm{def}}{=}}  % for definitions
\newcommand{\bigO}{\mathcal{O}}
\newcommand{\smallO}{\mbox{\tiny $\mathcal{O}$}}

\def\beq{\begin{equation}}
\def\eeq{\end{equation}}

\def\exp{\operatorname{exp}}
\makeatother

\def\vps{\varepsilon}
\theoremstyle{definition}

\newtheorem{Remark}{Remark}[section]
\newtheorem{Example}{Example}[section]

\usepackage[margin=1in]{geometry}

\begin{document}

\title[AB-LASSO]{Arellano-Bond LASSO Estimator for Dynamic Linear Panel Models$^*$}
 %\\
%\vspace{0.5cm}
%\large{[Preliminary and Incomplete Version]}}
\author[Chernozhukov, Fern\'andez-Val, Huang and Wang]{Victor Chernozhukov \and Iv\'an Fern\'andez-Val \and  Chen Huang  \and  \\ Weining Wang$^\dag$}

\thanks{$^*$ We thank the editor Stephane Bonhomme, three anonymous referees, Marine Carrasco, Xu Cheng, Matt Hong, Arturas Joudis, Whitney Newey and Wendun Wang for helpful comments. Fern\'andez-Val acknowledges research support from the Spain's State Research Agency under the Mar\'ia de Maeztu Unit of Excellence Program for the project CEX2020-001104M. Huang acknowledges financial support from the Independent Research Fund Denmark through the Inge Lehmann Grant (1132-00019B). Wang gratefully acknowledges financial support from the ESRC Grant (ES/T01573X/1); the project ``IDA Institute of Digital Assets'', CF166/15.11.2022, financed under the Romania's National Recovery and Resilience Plan; and the Marie Sk\l{}odowska-Curie Actions under the European Union's Horizon Europe research and innovation program for the Industrial Doctoral Network on Digital Finance, Project No. 101119635. The codes to implement the algorithms are publicly accessible via the GitHub repository: \href{https://github.com/huangche/Arellano-Bond-LASSO}{Arellano-Bond-LASSO}. Additionally, a CRAN package, \href{https://cran.r-project.org/web/packages/ablasso/index.html}{\textsf{ablasso}}, has been developed for \textsf{R} users to facilitate implementation.}

\thanks{$^\dag$ Chernozhukov: Department of Economics and Center for Statistics and Data Science, MIT; Fern\'andez-Val: Department of Economics, Boston University; Huang: Aarhus Center for Econometrics, Aarhus University; Wang: Department of Economics, University of Bristol.}
% \author[a]{Victor Chernozhukov}
% \affil[a]{\small{Department of Economics and Center for Statistics and Data Science, MIT}}
% \author[b]{Iv\'an Fern\'andez-Val}
% \affil[b]{\small{Department of Economics, Boston University}}
% \author[c]{ Chen Huang}
% \affil[c]{\small{Department of Economics and Business Economics, Aarhus University}}
% \author[d]{\\ Weining Wang}
% \affil[d]{\small{Falculty of Economics and Business, University of Groningen}}

%\begin{document}
	\maketitle
\begin{abstract}
The Arellano-Bond estimator is a fundamental method for dynamic panel data models, widely used in practice.  It can be severely biased when the time series dimension of the data, $T$, is long. The source of the bias is the large degree of overidentification. We propose a simple two-step approach to deal with this problem. The first step applies LASSO to the cross-section data at each time period to select the most informative moment conditions, exploiting the approximately sparse structure of these conditions. The second step applies a linear instrumental variable estimator using the instruments constructed from the moment conditions selected in the first step. Using asymptotic sequences where the two dimensions of the panel grow with the sample size, we show that the new estimator is consistent and asymptotically normal under much weaker conditions on $T$ than the Arellano-Bond estimator. Our theory covers models with high-dimensional covariates including multiple lags of the dependent variable and strictly exogenous covariates, which are becoming common in modern applications.  We illustrate our approach by applying it to weekly county-level panel data from the United States to study opening K-12 schools and other mitigation policies' short and long-term effects on COVID-19's spread.

\par
\vspace{0.2cm}

\end{abstract}

\section{Introduction}

%Review dynamic panel, typical issue of endogeneity. Anderson Hsiao. Why AB is good.\\

Panel data involve observations collected for cross-sectional units ($i=1,\ldots,N$) over multiple time periods ($t=1,\ldots, T$). Models for panel data are commonly used in economics and other social sciences because they allow researchers to control for unobserved unit and time heterogeneity and account for unit-level dynamics. These models have multiple applications, including evaluating job training and minimum wage regulations in labor economics, studying household consumption and economic growth in macroeconomics, estimating demand models for products in microeconomics, and analyzing payout policies and investment decisions in corporate finance. See \citet{bond2002dynamic} for a review of methods and applications of dynamic panel data models.

%provides a practical guide to analyzing micro data using firm-level panels (which typically have a small $T$ and a large $N$).

%Panel data are broadly encountered in economics and other social science studies, especially with the emerging big data era. The observations are collected  over both time $t=1,\ldots, T$ and cross-sectional $i=1,\ldots,N$ dimensions. In particular, dynamic model including the feedback effect of lagged response variables is of interest. %The model and calibration lead to unique issues.
%Applications of dynamic panel data models are widespread, including empirical models of labor economics (job training, minimum wage effects), macroeconomics (household consumption, economic growth), and corporate finance (payout policy, investment decisions), among others. For instance, a practical guide to micro data using the firm-level panels (small $T$, large $N$) can be found in \citet{bond2002dynamic}.

\cite{holtz1988estimating} introduced instrumental variable (IV) estimation in dynamic panel data models without strict exogeneity. Following this line of research, the Arellano-Bond estimator (AB) has become the most widely used fundamental method for panel models \citep{arellano1991some,arellano1995another}. It applies to dynamic linear models that include lagged dependent variables and predetermined covariates as explanatory variables and unobserved unit and time fixed effects. After taking first  differences or forward orthogonal deviations to remove the unit fixed effects, AB constructs moment conditions using sufficiently lagged dependent variables and covariates as instruments and applies the Generalized Method of Moments (GMM) to estimate the model parameters. However, AB might be severely biased in long panels. The problem arises because the number of moment conditions grows with the square of $T$, $T^2$, leading to many instrument bias caused by the large degree of overidentification in the GMM problem \citep[e.g.,][]{newey2004higher}. More precisely, AB has an asymptotic bias of order $T/N$, which might not be negligible compared to $1/\sqrt{NT}$, the size of the stochastic error, when the time dimension $T$ is sufficiently large relative to $N$ \citep{alvarez2003time}. This problem causes bias in estimators and undercoverage of confidence intervals.\footnote{The original motivation of AB was specification testing. We abstract from this issue because our method does not naturally lead to a specification test. Readers interested in this topic are referred to the existing literature on specification testing in IV settings with many instruments, including \cite{anatolyev2011specification}, \cite{hansen2008estimation}, \cite{chao2014testing}, and \cite{shi2025testing}.}

We address the bias issue of AB within long panels through a two-step method. After removing the unobserved unit fixed effects by forward orthogonal deviations, we first select the most informative moment conditions, followed by applying a linear instrumental variable estimation using instruments derived from these conditions. Specifically, as the number of AB's moment conditions varies across time periods, we perform a moment selection procedure on cross-section data for each time period separately. We utilize the least absolute shrinkage and selection operator (LASSO) of \cite{tibshirani1996regression} as our selector, given the naturally sparse structure of the moment conditions under appropriate weak temporal dependence conditions.

Several moment selection methods have been previously established in other contexts. For instance, \citet{donald2009choosing} introduced an alternative method for selecting instruments based on asymptotic mean squared error calculations. \cite{belloni2012sparse} described a similar approach to select optimal instruments using LASSO in cross-section instrumental variable models. Other methods for constructing optimal instruments through model averaging have been proposed by \cite{kuersteiner2010constructing} and \cite{okui2011instrumental}. \citet{luo2016selecting} expanded the LASSO selector to nonlinear GMM settings with many potential moments, noting its computational advantages over \citet{donald2009choosing}.
 Although some studies employ the AB estimator to underpin their analyses or conduct simulations \citep[e.g.,][]{newey2009generalized}, none addresses the AB estimator directly, mainly because they assume independent and identically distributed data.
 %-- an assumption invalidated in the AB context by the introduction of temporal dependencies when first differences are used remove unit fixed effects.

LASSO utilizes the $\ell_1$-norm to select the moment conditions. To show the validity of the selector, we build on theoretical results achieving near-oracle rates for LASSO and related  estimators in \citet{bickel2009simultaneous},  \citet{BCH2011} and \citet{belloni2013least}. An additional complication comes from the high dimensionality of the problem. As we mentioned above, the number of moment conditions grows with $T^2$. Moreover, we allow for high-dimensional covariates  including multiple lags of the dependent variable and  other strictly exogenous covariates, which are becoming common in modern applications, and unit and time fixed effects, whose number grows with the two dimensions of the panel. While our method does not suffer from shrinkage and model selection biases because the moment conditions of the second step are Neyman-orthogonal with respect to the parameters estimated in the first step,  it can still be subject to over-fitting bias, specially in the presence of high-dimensional covariates. We deal with this problem by combining the two steps of our procedure using cross-fitting \citep{DML}. Thus, we partition the panel in two parts. We select the moment conditions in the first part and estimate the parameters in the second part. Then, we repeat the procedure reversing the roles of the two parts and aggregate the results by averaging the estimates of the model parameters from the two orderings. Cross-fitting does not change the large sample properties of the estimator in this case because the use of forward orthogonal deviations to remove the unit fixed effects attenuates the dependence between the two stages by spreading the transformed error over many future periods.  We nevertheless find small sample improvements in various simulation settings.\footnote{In a previous version of the paper, we used first differences to remove the unit fixed effects. In that case, cross-fitting improved the rate condition for deriving the asymptotic distribution of the estimator by a factor of $\sqrt{T}$. The improvement arises because cross-fitting removes the dependence between the generated errors of the selected instruments and the  transformed errors in the main regression %, which become serially dependent after the first-difference transformation, 
by conditioning on the sub-sample used for moment selection.} %{\red In the current version of the paper, we adopt forward orthogonal deviations to transform the model, which attenuate the two stages dependence by spreading the transformed error over many future periods and render the asymptotic properties of the estimator identical with or without sample splitting.}} 

There is an extensive recent literature on panel data with large $T$, including dynamic linear models. \citet{alvarez2003time} studied the properties of AB and other estimators in long panels. They showed that AB exhibits asymptotic bias when $T/N$ tends to a constant. \citet{moral2013likelihood}, \citet{moral2019dynamic} and \citet{alvarez2022robust} developed likelihood-based alternatives to the AB estimator  and \citet{chen2019mastering} proposed a debiasing method based on applying the split-panel idea of \citet{dhaene2015split} to the cross-section dimension of the panel, whereas \citet{nickell1980correcting}, \citet{kiviet1995bias}, \citet{hahn2002asymptotically} and \citet{chudik2018half} developed alternatives based on bias corrections of the fixed effects estimator.  We compare our method to these alternatives in numerical simulations. 
\citet{okui2009optimal} proposed a method to select instruments by characterizing the mean squared error of the one-step AB estimator that uses the matrix of second moments of the instruments as weighting matrix in models with homoskedastic errors and strictly exogenous regressors. In the same setting, \citet{carrasco2024regularized} developed a version of the one-step AB estimator that regularizes the weighting matrix building on \citet{carrasco2012regularization}. Our instrument selection method is different and applies to both one-step and two-step AB estimators. It also does not rely on homoskedastic errors and allows for weakly exogenous covariates. \cite{cheng2023weight} have recently proposed a minimum distance estimator that combines graphical LASSO to regularize the weighting matrix with cross-fitting to reduce bias. While the steps of the procedure are similar to ours, the setting and regularization method are different.

We make four main theoretical contributions. First, we show that the moment conditions of AB exhibit an approximately sparse structure under suitable temporal dependence conditions (see, e.g., Proposition \ref{prop:decayK0}) and propose a LASSO version of AB that fruitfully exploits this structure. In particular,  the effective dimension of the non-zero coefficients in the first step estimation of the selected instruments at each time period $t$ is the minimum of $\log N$ and $t$,  which is very ``low" relative to the cross-sectional size $N$.  Second, we propose a 
%LASSO-type procedure that exploits this sparsity} by selecting the most informative moment conditions, together with a related 
cross-fitting procedure based on sample splitting that reduces finite-sample overfitting bias and opens the door to the use of machine learning methods other than LASSO.
%This enables us to remove the over-fitting bias and achieve a more favorable convergence rate of the estimator under simple regularity conditions. 
Unlike the existing high-dimensional IV literature, which typically focuses on cross-sectional settings, our framework is developed for dynamic panel models, where the relevant moment conditions may vary across time periods. We deal with this issue by performing moment selection separately for each time period and then aggregating the information over both $i$ and $t$ in the final instrumental variable estimator. The corresponding theory accommodates temporal dependence in the outcome and covariates, as well as the presence of time effects.
 Third, we show that first differences and forward orthogonal deviations have different properties in our setting. In particular, under suitable conditions, forward orthogonal differences yields more efficient estimators than first differences, and its large-sample properties are not affected by the use of cross-fitting. %[Add new contribution about the panel data structure including time effects and temporal dependency]
%Fourth,  we extend existing methodologies in the high-dimensional regression literature by accommodating complex temporal dependence in the outcome and covariates and allowing for time fixed effects.}
%we adopt forward orthogonal deviations to transform the model, thereby improving the efficiency of the estimator relative to first differencing.} %We further employ a cross-fitting procedure based on sample splitting to reduce the finite-sample overfitting bias and to open the door to the use of machine learning methods beyond LASSO in moment selection.} 
Fourth, we consider models with high-dimensional covariates that are becoming common in modern applications; e.g., \cite{klosin2022estimating,semenova2017estimation,gao2024robust,takeshima2025detecting}. Here, we achieve moment selection %selection of moment conditions and covariates simultaneously} 
by constructing orthogonal conditions and employing regularized GMM with a sparse weighting matrix via the Dantzig selector, in order to partial out the effect of high-dimensional nuisance parameters.

%Our theoretical contribution lies in the following aspects. First, we show that due to the specific decaying structure of the dynamic panel data, the effective dimension of the nonzero parameters in the preliminary prediction step is actually very ``low'' for each $t$ (i.e. $\min(\log N, t)$). Secondly, a cross-fitting procedure based on cross-sectional sample-splitting is introduced so that {the dependency of the generated errors of the fitted IV and the errors in the main regression} is removed by conditioning on the sub-sample and we can achieve a more favorable convergence rate of the estimator. {Thirdly, many exogenous variables of a diverging dimension are allowed as a model extension. We partial out the effect of high dimensional nuisance parameters by constructing orthogonal conditions and employ regularized GMM with a sparse weighting matrix via Dantzig selector to accomplish variable and moment selections simultaneously. }

The finite sample properties of our method are illustrated through comprehensive Monte Carlo simulations, where we compare it with alternative approaches such as likelihood-based,  debiased fixed effects, and debiased AB estimators. Lastly, we report results from an empirical application to the effect of the opening of K-12 schools and other policies on the spread of COVID-19 using a panel of $2,510$ US counties over $32$ weeks, extracted from the dataset used in \citet{chernozhukov2021association}. We estimate a panel regression model with rich dynamics, incorporating four lags of the dependent variable and several predetermined covariates. Due to the large number of instruments ($m=3,375$), the small bias condition for AB of \cite{chen2019mastering}, $m^2/(NT) \to 0$, fails for the AB estimator ($m^2/(NT) \approx 168$).\footnote{In this calculation $T=27$ because the first 5 observations are used as initial conditions of the dynamic model.} Compared to AB, our method %finds that policies such as the opening of K-12 schools, stay-at-home orders, and banning gatherings have more muted absolute long-run effects, and  college visits have substantially smaller effects in both the short and long run. 
estimates substantially smaller absolute effects of college visits and banning gatherings, but more negative effects of mask mandates, in both the short and long run. 

\subsection*{Notation.}
%Throughout this paper, we adopt the following notations.
For a column vector $v=(v_1,\ldots,v_d)^\top\in \R^d$ and a constant $r\geq 1$, we denote $|v|_r=(\sum_{i=1}^d |v_i|^r)^{1/r}$ and $|v|_{\infty}=\max\limits_{1\leq i\leq d}|v_i|$. Define $|v|_0$ as the zero norm, i.e. the number of non-zero coordinates. For a matrix $A=(a_{ij})_{1\le i\le m, 1\le j\le n}$, we define $|A|_{\max}=\max\limits_{1\leq i\leq m,1\leq j\leq n}|a_{ij}|$, $|A|_{1}=\max\limits_{1\leq j\leq n}\sum_{i=1}^m|a_{ij}|$, %$|A|_2$ as the spectral norm of $A$,
$|A|_{\infty}=\max\limits_{1\leq i\leq m}\sum_{j=1}^n|a_{ij}|$, and $|A|_{1,1} = \sum_{i=1}^m\sum_{j=1}^n|a_{ij}|$. %For any function on a measurable space $f:\mathcal{W}\to\R$, we write $\E_n(f)=n^{-1}\sum_{t=1}^n{f(w_t)}$ and $|f|_{\infty}=\sup_{w\in\mathcal W}|f(w)|$.
For a random variable $X_{it}$, we say $X_{it}\in\mathcal L^r$ if $\lVert X_{it}\rVert_r\defeq(\E|X_{it}|^r)^{1/r}<\infty$ for some $r>0$, and define the sub-Gaussian norm as $\|X_{it}\|_{\psi_{1/2}} = \inf\{s >0: \E\exp(X_{it}^2/s^2)\leq 2\}$, where $\E$ denotes the expectation conditional on the unit and time effects. We denote the limit cross-sectional average by $\bar\E(\cdot)$, that is $\bar\E(X_{it}) = \lim\limits_{N\to\infty}N^{-1}\sum_{i=1}^N\E(X_{it})$, provided that the limit exists.\footnote{Throughout the paper, we assume that the relevant limits exist, whenever we take expectations.} %The centered $X$ is denoted by $\E_0X$, i.e., $\E_0X=X-\E X$.
Given two sequences of positive numbers ${a_n}$ and ${b_n}$, we write $a_n\lesssim b_n$ (resp. $a_n\asymp b_n$) if there exists $C>0$, which does not depend on $n$, such that $a_n/b_n\le C$ (resp. $1/C\le a_n/b_n\le C$) for all large $n$. For a sequence of random variables ${X_n}$, we use the notation $X_n\lesssim_\P b_n$ to denote $X_n=\bigO_{\P}(b_n)$. For two real numbers, set $x\vee y=\max(x,y)$ and $x\wedge y=\min(x,y)$. %$a_n\ll b_n$ means that $a_n =o(b_n)$.

\subsection*{Outline.} The rest of the paper is organized as follows. Section \ref{sec:set} introduces the model and estimators. Section \ref{mainmm} presents the main theoretical results. Sections \ref{sim} and \ref{app} report the results of the simulation study and empirical application, respectively. Section \ref{sec:conclusion} contains some concluding remarks. An appendix collects the deferred proofs of the theoretical results, additional results for the analysis that uses first differences, instead of forward orthogonal deviations, to remove the unobserved unit effects, and supplementary results for the simulation study.

%The following notations are adopted throughout this paper. For a vector $v=(v_1,...,v_d)^\top\in \R^d$ and a constant $q\geq 1$, we denote $|v|_q=(\sum_{i=1}^dv_i^q)^{1/q}$ and $|v|_{\infty}=\max_{i}|d_i|$. For a matrix $A=(a_{i,j})_{1\le i\le m, 1\le j\le n}$, we define $|A|_{\text{max}}=\max_{i,j}|a_{i,j}|$. For any function on a measurable space $f:\mathcal{W}\to\R$, write $\E_n(f)=n^{-1}\sum_{t=1}^n\{f(w_t)\}$ and $|f|_{\infty}=\sup_{w\in\mathcal W}|f(w)|$. For a random variable $X$, we say $X\in\mathcal L^s$ if $\lVert X\rVert_s\defeq(\E|X|^s)^{1/s}<\infty$ for some $s>0$. Given two sequences of positive numbers $\{a_n\}$ and $\{b_n\}$, write %$a_n=\bigO(b_n)$ or
%$a_n\lesssim b_n$ (resp. $a_n\asymp b_n$) if there exists $C>0$ which does not depend on $n$ such that $a_n/b_n\le C$ (resp. $1/C\le a_n/b_n\le C$) for all large $n$. %, and say $a_n=\smallO(b_n)$ if $a_n/b_n\rightarrow0$ as $n\rightarrow\infty$.
%For a sequence of random variables $\{x_n\}$, we use the notation $x_n\lesssim_\P b_n$ to denote $x_n=\bigO_{\P}(b_n)$. %, namely $\forall \epsilon>0$, $\exists C>0$ such that $\P(x_n/b_n\geq C)<\epsilon$ for all large $n$, and say $x_n=\smallO_{\P}(b_n)$ if $x_n/b_n\stackrel{\P}{\to}0$ as $n\rightarrow\infty$.
%Finally, we denote the centered $X$ by $\E_0X$, i.e., $\E_0X=X-\E X$.

%The rest of the paper is organized as follows. Section \ref{sec:set} introduces the model setup and estimation steps. Section \ref{mainmm} presents the main theorems. Section \ref{sim} shows some simulation studies and Section \ref{app} is concerning the empirical application.

\section{Model and Estimators}\label{sec:set}
\subsection{Basic Model}
Let $\{(Y_{it},D_{it},C_{it}): 1 \leq i \leq N, 1 \leq t \leq T\}$ be a panel dataset, where $i$ and $t$ index unit and time period, respectively. $Y_{it}$ is a scalar outcome or response variable, $D_{it}$ is the policy variable or treatment of interest, and $C_{it}$ is a  vector of covariates of fixed dimension including, for example, $Y_{i,t-1}$ and other treatments. To measure the effect of $D_{it}$ on $Y_{it}$, we consider a dynamic linear panel model:
\begin{equation}\label{main:model}
Y_{it} = \alpha_i + \gamma_t + X_{it}^{\top}\theta^0 + \varepsilon_{it}, \quad X_{it} := (D_{it}, C_{it}^\top)^{\top}, \quad 1 \leq t \leq T,
\end{equation}
where $\theta^0$ is the parameter of interest, $\alpha_i$ is an unobserved unit effect, $\gamma_t$ is an unobserved time effect, and $\varepsilon_{it}$ is an idiosyncratic error with zero mean and constant variance. We might also be interested in functions of $\theta^0$ such as long-run effects in dynamic models that include lags of the dependent variable as covariates. We refer to the empirical application in Section \ref{app} for an example.

In the theoretical analysis, we shall treat the unobserved unit and time effects as fixed parameters. This is equivalent to conditioning on the realization of all these effects.\footnote{Due to this conditioning, all probability statements should be qualified with almost surely. We shall omit this qualifier for notational convenience.}
We assume that $\{(D_{it},C_{it},\varepsilon_{it}): 1 \leq t \leq T\}$ are independent over $i$, and $\varepsilon_{it}$ is an uncorrelated sequence over $t$. %with respect to a proper filtration ({in terms of $t-1$}) to be defined later.
In addition, we assume that the treatment and covariates in $X_{it}$ are predetermined with respect to $\varepsilon_{it}$ in the sense that
$$\E(X_{is}\varepsilon_{it}) = 0, \text{ for all }  1 \leq s \leq t \leq T.$$

% In the basic model \eqref{main:model}, %$C_{it}$ is a fixed dimension vector of covariates.
% satisfying $$\E(\Delta C_{it} \Delta \varepsilon_{jt}) = 0, \text{ for all } i,j=1,\ldots,N.$$
% Examples of such covariates include strictly exogenous variables and lagged predetermined covariates such as second and higher lags of the dependent variable.
% we denote the dimension of $X_{it}$ by $d = 1+\dim(C_{it})$, where $\dim(C_{it})$ is fixed.
%{\red Moreover, we assume that $\E (\alpha_i \varepsilon_{it}) =0$ and $\E ( \gamma_t \varepsilon_{it}) =0$, such that $\theta^*$ can be interpreted as the coefficient of $X_{it}$ in the linear projection of $Y_{it}$ on $X_{it}$, $\alpha_i$ and $\gamma_t$. (REMOVE THIS SENTENCE?)

We remove the unobserved effects by taking forward orthogonal deviations (FOD) over time and demeaning all the variables at the unit level, namely
\begin{equation}\label{main:dmodel}
\Delta \widetilde Y_{it} =    \Delta \widetilde X_{it}^\top\theta^0 + \Delta \widetilde \varepsilon_{it}, \quad 1\leq t\leq T-1,
\end{equation}
where $\Delta \widetilde  Z_{it} = \Delta Z_{it} - \sum_{j=1}^N \Delta Z_{jt}/N$, $\Delta Z_{it} = c_t\{Z_{it} - \sum_{s=1}^{T-t} Z_{i,t+s}/(T-t)\}$, and $c_t = \sqrt{(T-t)/(T-t+1)}$, for $Z_{it} \in \{Y_{it},  X_{it},  \varepsilon_{it}\}$.\footnote{In the previous version of the paper \cite{chernozhukov2024arellano}, we used first differences instead of FOD to remove the unobserved effects. Appendix \ref{app:FD} contains some of the results for the resulting estimators.}  The transformed error, $\Delta \widetilde \varepsilon_{it}$, is an uncorrelated sequence over $t$ and satisfies the moment conditions
$$
\E (X_{is}\Delta \widetilde \varepsilon_{it}) = 0, \text{ for all } 1\leq s \leq t \leq T-1.
$$

%{\red To make the above moment conditions hold, we need $\E (X_{is}\varepsilon_{jt} )= 0$, for all $s \leq t$ and $i,j=1,\ldots,N$.}
AB uses these moment conditions to construct a GMM estimator of $\theta^0$.\footnote{More precisely, the version of AB that uses moment conditions in FOD is from \cite{arellano1995another}.} It should be noted that AB is biased when $T$ is large due to the large number of moment conditions, i.e. $m = %(T-2)(T-1)/2 =
\bigO(T^2)$; see, e.g., \cite{newey2004higher} for more discussion. We propose an alternative estimator that is computationally simple and has lower bias when $T$ is large. It is based on the application of LASSO to select the most informative moment conditions to estimate the parameters. Thus, the estimator has two stages. It first selects moment conditions using LASSO, and then estimates the parameters of interest by instrumental variables, with the predicted values of the endogenous regressors obtained from the selected moment conditions serving as instruments. We name the new estimator as AB-LASSO as a shorthand for Arellano-Bond LASSO estimator.

\begin{Definition}[AB-LASSO] The AB-LASSO estimator consists of two steps:
\begin{enumerate}
\item[1] For $t= 1, \ldots, T-1$ and $W_{it}$ denoting any element of $\Delta \widetilde X_{it} = (\Delta \widetilde D_{it}, \Delta \widetilde C_{it}^\top)^\top$, run the LASSO regressions:
\begin{align} \label{LASSO}
\widehat{\Pi}_t \defeq (\widehat \pi_{t0},\widehat\pi_{t1}^\top, \ldots, \widehat\pi_{tt}^\top )^{\top} \in \arg \min_{\pi_{t0},\ldots,  \pi_{tt}} \bigg\{&\sum_{i=1}^N \bigg(W_{it} - \pi_{t0} - \sum_{s=1}^{t}  X_{is}^\top\pi_{ts}\bigg)^2 \notag\\
&+ \lambda_t \sum_{s=1}^{t} \omega_{ts} | \pi_{ts} |_1\bigg\},
\end{align}
 where $\lambda_t$ is a penalty tuning parameter, and $\omega_{ts}$ is a non-negative penalty weight that can incorporate any priors on the importance of the lags or account for conditional heteroskedasticity. For example, $\omega_{ts}$ can be specified as a non-decreasing function of $t-s$ if closer lags are believed to be more informative.\footnote{As a specific example, we can incorporate a factor of $t/s$ into the penalty weights.}
%is a non-negative weight that is a non-increasing function of $t-s$, e.g., $\omega_{ts} =1$.
Obtain the predicted values of the previous regression. For $\widehat W_{it}$ denoting any element of $\widehat{\Delta \widetilde X_{it}} = (\widehat{\Delta \widetilde D_{it}}, \widehat{\Delta \widetilde C_{it}^\top})^\top$,
%$\widehat W_{it} \in \widehat{\Delta \widetilde X_{it}}$,
\begin{equation*}
\widehat W_{it} = \widehat \pi_{t0} + \sum_{s=1}^{t} X_{is}^{\top}\widehat \pi_{ts}.
\end{equation*}
\item[2] Estimate \eqref{main:dmodel} by instrumental variable regression using $\widehat {\Delta \widetilde X_{it}}$ as the instrument for $\Delta \widetilde X_{it}$, that is
\begin{equation}\label{est}
\widehat \theta = \bigg(\sum_{i=1}^N \sum_{t=1}^{T-1} \widehat{\Delta \widetilde X_{it}} \Delta \widetilde X_{it}^{\top} \bigg)^{-1} \sum_{i=1}^N \sum_{t=1}^{T-1} \widehat{\Delta \widetilde X_{it}} \Delta \widetilde Y_{it}.
\end{equation}
 \end{enumerate}
\end{Definition}
%\item[b'] For model   \ref{main:model3} with fixed dimension $d$, replace the definition of $\widehat {\Delta X_{it}}$ with\\ $\widehat {\Delta X_{it}} = (\widehat{\Delta \tilde Y_{i,t-1}}, \cdots,   \Delta \tilde{Y}_{i,t-k}, \widehat{\Delta \tilde D_{it}},\Delta \tilde{C}_{it})^{\top}$.

%  {\color{red}
%  \begin{remark}[Remark \ref{extend}, continued]
% For model \eqref{main:model2} with fixed dimension $d$, we need to include the covariates $C_{i,t-1},\ldots, C_{i,2}$ in the LASSO step and replace the definition of $\widehat {\Delta\widetilde X_{it}}$ with $\widehat {\Delta \widetilde X_{it}} = (\widehat{\Delta \widetilde Y_{i,t-1}}, \Delta \widetilde{Y}_{i,t-2}, \ldots,  \Delta \widetilde{Y}_{i,t-k}, \widehat{\Delta \widetilde D_{it}}^{\top},\Delta \widetilde{C}_{it}^{\top})^{\top}$.
%  \end{remark}}

\begin{Remark}[Initial Conditions]\label{Remark:ic}
    We have implicitly assumed so far that the initial conditions of $Y_{it}$ are observed in models that include lags of the dependent variable as covariates. For example, we have assumed that $Y_{i0}$ is observed when $C_{it}$ includes $Y_{i,t-1}$. If $Y_{it}$ is first observed at $t=1$, then the vector $X_{is}$ in \eqref{LASSO} needs to be modified to include only the observed values of $C_{is}$. In models where $C_{it} = Y_{i,t-1}$, for example, $X_{i1} = D_{i1}$ instead of $X_{i1} = (D_{i1},Y_{i0})^\top$. \qed
\end{Remark}

\begin{Remark}[Post-LASSO]\label{Remark:post}
    A post-selection step can be applied to reduce the bias of LASSO in the estimated coefficients of the selected variables. For example, an ordinary least squares (OLS) regression can be performed after the first step  using only the variables in $X_{i1},\ldots,X_{i,t-1}$ selected by LASSO to estimate the instrument. This modification usually improves the finite sample properties of the estimator, but does not improve its properties in large samples in general. In Section \ref{mainmm}, we derive the properties of the estimator defined in \eqref{LASSO} that uses only LASSO. The main results in Theorem \ref{main} also apply to the post-LASSO estimator. For more details on post-LASSO and its comparison with LASSO, we refer to \citet{belloni2013least}. \qed
\end{Remark}

\begin{Remark}[Neyman-Orthogonality] Let $V_{it} = (1,X_{i1}^\top,\ldots,X_{it}^\top)^\top$ and \\$\Pi_t = (\pi_{t0},\pi_{t1}^\top,\ldots,\pi_{tt}^\top)^\top$. The estimator given in \eqref{est} is a moment estimator with moment function:
\begin{equation*}
    g_{i}(\theta,\Pi_1,\ldots,\Pi_{T-1}) = \sum_{t=1}^{T-1} \Pi_t^\top V_{it} ( \Delta \widetilde Y_{it} -  \Delta \widetilde X_{it}^\top \theta).
\end{equation*}
This moment function is Neyman-orthogonal with respect to each of the first stage parameters $\Pi_t$, $t=1,\dots,T-1$, because
\begin{equation*}
    \frac{\partial \E[g_{i}(\theta,\Pi_1,\ldots,\Pi_{T-1})]}{\partial \Pi_t}\Big|_{\theta=\theta^0} = \E[V_{it} ( \Delta \widetilde Y_{it} - \Delta \widetilde X_{it}^{\top}\theta^0 )] = 0, \quad t = 1,\ldots,T-1.
\end{equation*}
Note that the 2SLS (Two-Stage Least Squares) version of the second stage that replaces $\Delta \widetilde X_{it}$ by $\widehat{\Delta \widetilde X_{it}}$ in \eqref{est} does not satisfy this condition.\footnote{This difference between the IV and 2SLS versions of the estimator's moment conditions had not been noted previously, to the best of our knowledge. Some work such as \cite{belloni2012sparse} used the IV version and other work such as \cite{zhu2018sparse} used the 2SLS version. However, we are not aware of any work comparing these two versions.} \qed
\end{Remark}

% {\color{red} W2: 2SLS does not have fitted instrument variable, and therefore, is not directly coma parable to our estimator. AB  does not have generated error in the instrument. }

%\begin{remark}[Overfitting Bias]
AB is an instrumental variable estimator. Its bias comes from overfitting because the same observations are used to project the endogenous regressors on the instruments and to estimate the parameters \citep{phillips1977bias,angrist1995split,angrist1999jackknife}. The order of the bias is $m/n$, where $m$ is the number of moment conditions and $n$ is the sample size. In the case of AB, $m=\bigO(T^2)$ and $n=NT$, so that the order of the bias is $T/N$. The order of the sampling noise is $n^{-1/2} = (NT)^{-1/2}$, so that the small bias condition of \cite{chen2019mastering} is $m/n^{1/2} \to 0$ or equivalently $m^2/n = T^3/N \to 0$. Our proposed AB-LASSO estimator reduces the overfitting bias by selecting moment conditions and by using the FOD transformation.
%In particular, compared with first differencing, FOD attenuates the two stages dependence by spreading the transformed error over many future periods (with harmonic weights of order $1/(T-t)$). 
Up to logarithmic terms,
%$m = \bigO(T)$ for AB-LASSO. AB-LASSO therefore reduces the order of the bias to $1/N$ and
the small bias condition for AB-LASSO becomes $\max\limits_{1\leq t\leq T-1}\sqrt{s_t^*/N} \to 0$ ($s_t^*$ is the dimension of effective instruments for each $t$).  When $s_t^*$ is moderately large relative to $N$, AB-LASSO might still exhibit small sample or higher order bias. % [CHECK CHANGE OF $T$ BY $s_t^*$]
%The temporal dependence of the data can further exacerbate the problem  in finite samples.
%Indeed, we observe a more severe bias, compared to the i.i.d. case studied in \citet{belloni2012sparse}, in numerical simulations.
To reduce this bias, we develop a sample-splitting procedure over the cross-section dimension following the idea of the split-sample IV estimator of \citet{angrist1995split}.
% AB-LASSO may exhibit bias due of the correlation between the instruments generated in step 1 and the ordinary model error $\varepsilon_{it}$ in \eqref{main:model}. {\red It is well-known in the literature that the bias of such two-stage estimator increases with the dimension of the instruments used in the first stage. In our model setting, for each $t=2,\ldots,T$, we have the number of instruments depends on $t$ accordingly. Moreover, the temporal dependence of the generated errors in step 1 further contributes to the bias of our final estimator, which is aggregated over $t$ in step 2. Indeed, we observe a more severe bias compared to the i.i.d. case studied in \citet{belloni2012sparse}.} To reduce the bias, we develop a sample-splitting procedure over the cross-section dimension.
%\end{remark}
We name the version of AB-LASSO with sample splitting and cross-fitting as AB-LASSO-SS.

\begin{Definition}[AB-LASSO-SS] The AB-LASSO-SS estimator consists of the following steps:
\begin{enumerate}
    \item[1] Partition the sample $\{(Y_{it},D_{it},C_{it}): 1 \leq i \leq N, 1 \leq t \leq T\}$ along the cross-section dimension into two parts or sub-samples A and B, corresponding to the indexes $i \in \{1, \ldots, \lfloor N/2 \rfloor\} =: \mathbb{I}_A$ and $i \in \{\lfloor N/2 \rfloor +1, \dots,  N\} =: \mathbb{I}_B$, where $\lfloor \cdot \rfloor$ denotes the integer part.
    \item[2] In each sub-sample, take FOD over time and demean all the variables at the unit level, namely $\Delta \widetilde  Z_{it,s} = \Delta Z_{it} - \sum_{j \in \mathbb{I}_s} \Delta Z_{jt}/|\mathbb{I}_s|$, $i\in\mathbb I_s$, $s \in \{A,B\}$, and $\Delta Z_{it} = c_t[Z_{it} - \sum_{s=1}^{T-t} Z_{i,t+s}/(T-t)]$, for $Z_{it} \in \{Y_{it},  X_{it}\}$.
    \item[3] For $t= 1, \ldots, T-1$ and $W_{it}$ denoting any element  of $\Delta \widetilde X_{it,A} =  (\Delta \widetilde D_{it,A}, \Delta \widetilde C_{it,A}^\top)^\top$,
    %$W_{it} \in \Delta \widetilde X_{it,A}$,
    run step 1 of AB-LASSO in sub-sample A by estimating the LASSO regressions:
\begin{align} \label{LASSO.SS}
\widehat{\Pi}_{t,A} \defeq (\widehat \pi_{t0,A},\widehat\pi_{t1,A}^\top, \ldots, \widehat\pi_{tt,A}^\top )^{\top} \in \arg \min_{\pi_{t0},\ldots,  \pi_{tt}} \bigg\{&\sum_{i \in \mathbb{I}_A} \bigg(W_{it} - \pi_{t0} - \sum_{s=1}^{t}  X_{is}^\top\pi_{ts}\bigg)^2 \notag\\
&+ \lambda_t \sum_{s=1}^{t} \omega_{ts} | \pi_{ts} |_1\bigg\},
\end{align}
where $\lambda_t$ is a penalty tuning parameter, and $\omega_{ts}$ is a non-negative penalty weight.
 %that can incorporate any priors on the importance of the lags. For example, $\omega_{ts}$ can be specified as a non-decreasing function of $t-s$ if closer lags are believed to be more informative.
%is a non-negative weight that is a non-increasing function of $t-s$, e.g., $\omega_{ts} =1$.
Obtain the predicted values in sub-sample B using the previous estimates from sub-sample A. For  $\widehat W_{it,BA}$ denoting any element of $\widehat{\Delta \widetilde X_{it,BA}} =  (\widehat{\Delta \widetilde D_{it,BA}}, \widehat{\Delta \widetilde C_{it,BA}^\top})^\top$,
%$\widehat W_{it,BA} \in \widehat{\Delta \widetilde X_{it,BA}}$,
$$
\widehat W_{it,BA} = \widehat \pi_{t0,A} + \sum_{s=1}^{t} X_{is}^\top\widehat \pi_{ts,A},\quad i \in\mathbb{I}_B.
$$
Run the second step of AB-LASSO in sub-sample B using the instruments $\widehat{\Delta \widetilde X_{it,BA}}$,
\begin{equation}\label{est.SS}
\widehat \theta_{B,A} = \bigg( \sum_{i \in \mathbb{I}_B} \sum_{t=1}^{T-1} \widehat{\Delta \widetilde X_{it,BA}} \Delta \widetilde X_{it,B}^\top \bigg)^{-1} \sum_{i\in \mathbb{I}_B}\sum_{t= 1}^{T-1} \widehat{\Delta \widetilde X_{it,BA}} \Delta \widetilde Y_{it,B}.
\end{equation}
\item[4] Run step 3 reversing the roles of sub-samples A and B to obtain $\widehat \theta_{A,B}$.
\item[5] Compute the cross-fitting estimator of $\theta^0$ as the average of the estimators in the two orderings
\begin{equation}\label{SS}
\widehat{\theta}_{SS}= (\widehat{\theta}_{A,B}+ \widehat{\theta}_{B,A})/2.
\end{equation}

\end{enumerate}
%     More specifically, we  Note that the cross-section mean difference to remove the time effects should be taken within the two splits.
% %$$
% %\widehat W_{it,A} = \hat \pi_{t1,A} + \sum_{s=2}^{t-1} X_{is}^\top\hat \pi_{ts,A}.
% %$$
% Denote by $(\widehat \pi_{t0,A},\ldots, \widehat \pi_{t,t-1,A} ) $ the LASSO estimators of step 1 using sub-sample A. Accordingly, we have the predicted values
% $$
% \widehat W_{it,A} = \widehat \pi_{t0,A} + \sum_{s={\red 1}}^{t-1} X_{is}^\top\widehat \pi_{ts,A},\quad i=\lfloor N/2 \rfloor +1,\ldots,N.
% $$
% In the second step, we use the fitted instruments $\widehat W_{it,A}$ from step 1 and the sub-sample B to compute
% \begin{equation*}
% \widehat \theta_{B,A} = \bigg( \sum_{i=\lfloor N/2\rfloor +1}^N \sum_{t={\red 2}}^T \widehat{\Delta \widetilde X_{it,A}} \Delta \widetilde X_{it,A}^\top \bigg)^{-1} \sum_{i=\lfloor N/2\rfloor +1}^N\sum_{t={\red 2}}^T \widehat{\Delta \widetilde X_{it,A}} \Delta \widetilde Y_{it}.
% \end{equation*}
% In a similar fashion, we attain $\hat \theta_{A,B}$ by reversing the roles of A and B. At last, averaging the two estimators yields the final estimate
% \begin{equation*}
% \widehat{\theta}_{SS}= (\widehat{\theta}_{A,B}+ \widehat{\theta}_{B,A})/2.
% \end{equation*}

\end{Definition}

\begin{Remark}[$K$-Fold and Multiple Splitting] The above cross-fitting procedure can be further generalized with $K$-fold sample splitting (e.g. $K=5$). Each of the $K$ sub-samples is used as the main sample for estimating \eqref{est.SS} while the rest form the auxiliary sample to fit the LASSO estimate in \eqref{LASSO.SS}. The resulting $K$ estimates corresponding to the different partitions are averaged. The FOD transformation and cross-section demeaning are taken within the main and auxiliary samples. Moreover, since the ordering of the cross-section units is arbitrary by the independence assumption, we recommend repeating the procedure for multiple splits by randomly permuting the index $i$ across units and aggregate the estimates by averaging or taking the median across permutations. The use of multiple sample splits makes the estimator invariant to the ordering of the cross-section units. \qed
    \end{Remark}

\begin{Remark}[Comparison with SSIV] AB-LASSO-SS has two main differences with respect to the split-sample IV (SSIV) estimator of \cite{angrist1995split} applied to a dynamic panel model. First, we use LASSO instead of %ordinary least squares (OLS)
OLS in the first step to project the endogenous regressors on the instruments. %{\color{red}In numerical simulations, we find that using OLS in the first stage produces biased estimators of the parameters of interest, even when we combine it with sample splitting, see Tables \ref{table:olst1} and \ref{table:olst2} in the Appendix.}
Second, we use cross-fitting to improve efficiency and employ multiple sample splits to ensure robustness to the choice of sample split.  \qed
\end{Remark}

\subsection{General Model with Many Exogenous Covariates}\label{debias}
In many empirical panel applications, researchers augment dynamic specifications with a rich set of additional covariates to strengthen identification and improve robustness. Examples include dynamic treatment effect models with many policy indicators and firm-level panels with extensive macroeconomic controls.  
%These covariates may include macroeconomic controls, policy indicators, or industry--time interactions, as well as  higher-order lags of the dependent variable. 
In modern applications, the number of such controls can be large relative to the sample size, especially when flexible specifications or many interaction terms are considered. To reflect this practice, we extend the basic model in \eqref{main:model} by including additional covariates:
\begin{equation}\label{main:model2}
Y_{it} = \alpha_i + \gamma_t + X_{it}^{\top}\theta^0 + \varepsilon_{it}, \quad X_{it} := (D_{it}, C_{it}^\top, X_{2,it}^{\top})^{\top}, %\quad k+1 \leq t \leq T,
\end{equation}
where $X_{2,it}$ is a possibly high-dimensional vector of covariates that is independent over $i$ and satisfies
$\E(\Delta\widetilde X_{2,it} \Delta\widetilde\varepsilon_{it}) = 0$. 
The leading cases of such covariates are strictly exogenous variables with respect to $\varepsilon_{it}$.
%{\red and lagged predetermined covariates such as second and higher lags of the dependent variable}.
Denote the dimension of $X_{2,it}$ by $d_2$.
%We consider two cases dependending on $d_2$, the dimension of $X_{2,it}$. The low-dimensional case arises when $X_{2,it}$ has few components such that we can treat $d_2$ as fixed in the asymptotic analysis. In this case, $X_{2,it}$ can be partialled-out in a similar fashion to the time effects without introducing any bias. We do not discuss this case further.
The high-dimensional case arises when $d_2$ is large relative to the sample size $n=NT$ such that it is more appropriate to treat $d_2$ as increasing in the asymptotic analysis. 

%For this case, we propose a debiasing procedure to partial out the effect of $X_{2,it}$.

% {\color{red}
% This framework is particularly useful in applications where researchers wish to include many plausibly exogenous controls while maintaining valid inference for a small set of structural parameters. Examples include dynamic treatment effect models with rich policy controls, firm-level panels with extensive macro covariates, and macro panels with many lags or interaction terms. Our approach differs from standard IV-LASSO procedures in that we do not select instruments for a high-dimensional structural parameter; instead, we construct orthogonalized instruments that partial out high-dimensional exogenous controls while targeting inference on a fixed-dimensional parameter of interest within a dynamic panel setting with unit and time effects.
% }

%Examples of such covariates include strictly exogenous variables and lagged predetermined covariates including second and higher lags of the dependent variable when $d$ is fixed (low-dimensional covariate case) or when $d$ diverges with $N$ and $T$ (high-dimensional covariates case). In what follows, we shall include $Y_{i,t-2}, \cdots, Y_{i,t-k}$ in the vector $C_{it}$.

\begin{Remark}[Strictly Exogenous Covariates]\label{SEC} If $X_{2,it}$ includes strictly exogenous covariates, then there are additional moment conditions that can be used to estimate $\theta^0$. In particular,
$$
\E (X_{2,is}^{se} \Delta \widetilde \varepsilon_{it}) = 0, \text{ for all } 1 \leq s \leq T, \quad 1\leq t\leq T-1,
$$
where $X_{2,is}^{se}$ is the subset of strictly exogenous covariates of $X_{2,it}$. These additional moment conditions can be incorporated to step 1 of AB-LASSO.  \qed
\end{Remark}

When the dimension of the additional covariates is small, their effects can be removed using standard projection arguments together with the unit and time fixed effects. However, when the number of covariates is large relative to the sample size, classical partialling-out is no longer feasible. In particular, direct estimation of the full parameter vector by least squares becomes ill-posed, and naive regularization (e.g., plugging in LASSO estimates) introduces shrinkage bias that contaminates inference on the low-dimensional parameters of interest. To address this issue, we treat the coefficients on the high-dimensional covariates as nuisance parameters and focus inference on a small set of structural parameters. Our strategy is to construct orthogonal moment functions that are insensitive, to first order, to regularization error in the estimation of the nuisance parameters. The key idea is to build time-specific orthogonalized instruments for the endogenous regressors that are strongly correlated with the components of interest and orthogonal to the high-dimensional covariates. This orthogonality ensures that small estimation errors in the nuisance component do not affect the asymptotic distribution of the estimator of the parameters of interest.

To formally explain this partialling-out procedure, it is convenient to rewrite the extended model \eqref{main:model2} as:
$$Y_{it} = \alpha_i + \gamma_t + X_{1,it}^{\top}\theta^0_1 + X_{2,it}^{\top}\theta^0_2 + \varepsilon_{it}, \quad X_{1,it} := (D_{it}, C_{it}^\top)^{\top},$$
where $\theta_1^0\in\R^{d_1}$ and $\theta_2^0\in\R^{d_2}$. Denote $d=d_1+d_2$, where $d_1$ is fixed and $d_2$ is growing with $n$. Assume the sparsity assumption $|\theta_2^0|_0=\smallO(n)$. %Suppose that $X_{2,it}$, which includes more lags (beyond 1) of $Y_{it}$ and other controls, satisfies $\E(\Delta\widetilde X_{2,it}\Delta\widetilde\varepsilon_{it})=0$.
The moment functions are given by
$$g_{it}(\theta_1,\theta_2) = \E\big\{(\Delta\widetilde Y_{it} - \Delta\widetilde X_{1,it}^\top\theta_1 - \Delta \widetilde X_{2,it}^\top\theta_2)U_{it}\big\},$$
where $U_{it}=(U_{it}^{0\top}, \Delta \widetilde X_{2,it}^\top)^\top$, and $U_{it}^0$ ($d_1\times1$) contains the most informative IVs for $\Delta\widetilde X_{1,it}$.

To construct orthogonalized instruments in the high-dimensional setting, we seek a time-specific weighting matrix $\mathcal W_t$ ($d\times d_1$) such that the transformed instruments $\mathcal W_t^\top U_{it}$ remain strongly correlated with $\Delta \widetilde X_{1,it}$, to preserve identification of $\theta_1^0$, and are orthogonal to $\Delta \widetilde X_{2,it}$, so that the influence of the high-dimensional nuisance component is removed. Formally, we aim to choose $\mathcal W_t$ so that
$\E \big\{\Delta \widetilde X_{2,it} (\mathcal W_t^\top U_{it})^\top \big\} = 0$, 
while ensuring that $\E \big\{ \Delta \widetilde X_{1,it} (\mathcal W_t^\top U_{it})^\top \big\}$ has full rank $d_1$.
% For each $t$, we wish to construct instruments for $\Delta\widetilde X_{1,it}$ from $U_{it}$, $\mathcal W_t^\top U_{it}$, that are orthogonal to $\Delta \widetilde X_{2,it}$. We can achieve this goal by finding a weighting matrix $\mathcal W_t$ ($d\times d_1$) such that
% $${\E} \Delta\widetilde X_{2,it} (\mathcal W_t^\top U_{it})^\top=0,$$
% and $\E \Delta\widetilde X_{1,it} (\mathcal W_t^\top U_{it})^\top$ is of rank $d_1$. 
This problem can be solved by Dantzig selector:
\begin{align}\label{penalize}
&\widehat{\mathcal W}_{t}=\arg\min_{\mathcal W_{t}}|\mathcal W_{t}|_{1,1} \quad\text{ subject to }\notag\\
&\left|N^{-1}\sum_{i=1}^N\left\{\begin{pmatrix}
\Delta \widetilde X_{1,it}\\
\Delta\widetilde  X_{2,it}
\end{pmatrix}
\widehat{U}_{it}^{\top}\right\}\mathcal W_{t}-\mathbf I_{d\times d_1} \right|_{\max}\leq\ell_{t},
\end{align}
where we have replaced $U_{it}^0$ by the LASSO predictions, i.e. $\widehat U_{it}=(\widehat {\Delta \widetilde X_{1,it}^\top}, \Delta \widetilde X_{2,it}^\top)^\top$, and $\mathbf I_{d\times d_1}$ represents the $d\times d_1$ sub-matrix of the $d\times d$ identity matrix. %$\widehat U_{it}=(\widehat{\Delta \widetilde D_{it}},\widehat{\Delta \widetilde C_{it}}^\top,\Delta \widetilde X_{2,it}^\top)^\top$.
Then, the instrument for $\Delta\widetilde X_{1,it}$ is $\widehat{\mathcal W}_{t}^{\top}\widehat{U}_{it}$ and the estimator of the parameters of interest becomes:\footnote{When $X_{2it}$ includes second or higher lags of the dependent variable the summation over $t$ in \eqref{diverged1} needs to be modified to include only the observed values of $\widehat U_{it}$. See Remark \ref{Remark:ic} for a related discussion.}
\begin{equation}\label{diverged1}
\widehat{\theta}_{1}=\bigg(\sum_{i=1}^N\sum_{t=1}^{T-1}\widehat{\mathcal W}_{t}^{\top}\widehat{U}_{it}{\Delta\widetilde X}_{1,it}^{\top}\bigg)^{-1}\bigg(\sum_{i=1}^N\sum_{t=1}^{T-1}\widehat{\mathcal W}_{t}^{\top}\widehat{U}_{it}\Delta\widetilde Y_{it}\bigg).
\end{equation}

Conceptually, in the low-dimensional case this orthogonalization reduces to a standard projection onto the orthogonal complement of the nuisance score. In the high-dimensional case, however, the projection cannot be computed exactly. We therefore approximate it by solving a constrained $\ell_1$-minimization problem that delivers a sparse weighting matrix. This step can be viewed as constructing an approximate Neyman-orthogonal score tailored to the panel structure of the model. The resulting estimator achieves valid inference for $\theta_1^0$ even when the dimension of the additional covariates grows with the sample size, provided the nuisance parameters are sufficiently sparse.

\section{Main Theorems}\label{mainmm}
In this section, we present the theoretical foundation of the proposed estimator.
We begin with the basic model \eqref{main:model}, which is a special case of the general model \eqref{main:model2}, %to build intuition before presenting results for the extended setting. We will start by demonstrating some results related to the basic model.
and establish the main results for this setting before turning to the extended model. 
%In particular, we will first address the model without the time effect $\gamma_t$ and then we discuss how the theory adapts in presence of $\gamma_t$.
Throughout this section, we impose the following conditions on the data generating processes.

\begin{Assumption}[Data Generating Processes]\label{a1}
The process $X_{it}\in\R^d$ is trend-stationary over $t$ and i.i.d.\ over $i$, conditional on any variables that do not change over $i$ or over $t$, that is, any individual and time effects, while the disturbance $\vps_{it}$ is stationary over $t$ and i.i.d.\ over $i$ conditional on the same variables. Both processes admit the representations: $X_{it}=F_{it}(\ldots,\xi_{i,t-1},\xi_{it})$ and $\vps_{it}=g_{it}(\ldots,\zeta_{i,t-1},\zeta_{it})$, where $F_{it}(\cdot)=(f_{it,1}(\cdot),\ldots,f_{it,d}(\cdot))^\top$, $F_{it}$ and $g_{it}$ are measurable functions, and $\xi_{it},\zeta_{it}$ for $t\in\mathbb Z, i\in\mathbb N$, are i.i.d.\ random elements.
\end{Assumption}

We allow for overlap in the innovations $\xi_{it}$ and $\zeta_{it}$, as long as the exogeneity conditions specified in Section \ref{sec:set} are satisfied, i.e., $\E(X_{is}\varepsilon_{it}) = 0$, for all $1 \leq s \leq t$.  Note that Assumption \ref{a1} allows for certain forms of non-stationarity in the process for  $X_{it}$ such as unit-specific deterministic time trends and unrestricted time effects. For example, the policy variables of the empirical application in Section \ref{app} can be modeled as unit-specific deterministic time trends. It does rule out, however, other forms of stochastic non-stationarity such as unit roots. The following definition, along with Assumptions \ref{a1} and \ref{a2}(i) below, adapts the functional dependence measure proposed by \citet{wu2005nonlinear} for stationary time series processes to heterogeneous panel data processes.

\begin{Definition}[Dependence Adjusted Norm]\label{dep}
For each $k=1,\ldots,d$, let $$X_{it,k}^{\ast}(\ell)=f_{it,k}(\ldots,\xi^\ast_{i,t-\ell},\ldots,\xi_{it}),$$ where $\xi_{i,t-\ell}$ is replaced by an i.i.d.\ copy $\xi^\ast_{i,t-\ell}$. For $r\geq1$, define the functional dependence measure $\delta_{it,k,r}(\ell) \defeq\|X_{it,k}^{\ast}(\ell) - X_{it,k}\|_r$, which measures the dependency of $\xi_{i,t-\ell}$ on $X_{it,k}$. Additionally, define $\Delta_{k,r,m}\defeq \sum\limits_{\ell=m}^\infty\max\limits_{1\leq i\leq N,1\leq t\leq T}\delta_{it,k,r}(\ell)$, which measures the cumulative effects for all $\ell\geq m$ and is uniform over $i$ and $t$. Moreover, the dependence adjusted norm of $X_{it,k}$ is introduced by $\|X_{\cdot,k}\|_{r,\varsigma}\defeq\sup_{m\geq0}(m+1)^{\varsigma}\Delta_{k,r,m}$, where $\varsigma>0$.\footnote{Assumption \ref{a1} presumes that conditioning on the individual and time effects $\{\alpha_1,\ldots,\alpha_N,\gamma_1,\ldots,\gamma_T\}$ is equivalent to conditioning on $\{\alpha_i,\gamma_t\}$. %, provided that $\{\alpha_i\}_{i=1}^N$ is independent across $i$ and $\{\gamma_t\}_{t=1}^T$ is independent over $t$.
As a result, the functional dependence measure $\delta_{it,k,r}(\ell)$ is a random function of $\alpha_i$ and $\gamma_t$, and we therefore define $\Delta_{k,r,m}$ (as well as the dependence adjusted norm) as a uniform measure over $i$ and $t$.}
\end{Definition}

% {\red General definition for possibly non-stationary time series}: For $k=1,\ldots,d$, let $X_{it,k}^{\ast}(\ell)=f_{ik}(\ldots,\xi^\ast_{i,t-\ell},\ldots,\xi_{it})$ where $\xi_{i,t-\ell}$ is replaced by its i.i.d. copy $\xi^\ast_{i,t-\ell}$. The functional dependence measure is denoted by $\delta_{it,k,r}(\ell) \defeq\|X_{it,k}^{\ast}(\ell) - X_{it,k}\|_r$ and define $\Delta_{k,r,m}\defeq \max_{1\leq i\leq N}\max_{1\leq t\leq T} \sum_{\ell=m}^\infty\delta_{it,k,r}(\ell)$ which measure the cumulative effects. The dependence adjusted norm of $X_{it,k}$ is introduced by  $\|X_{\cdot,k}\|_{r,\varsigma}\defeq\sup_{m\geq0}(m+1)^{\varsigma}\Delta_{k,r,m}$.

% {\color{red}
% % It is worth noting that we assume that the dependence adjusted norm
% % conditioning on $\sigma(\alpha_1, \cdots, \alpha_N)$ is equivalent to conditioning on $\alpha_i$.
% % The dependence adjusted norm is thus a random function of $\alpha_i$ and shall be  assumed to be uniformly bounded over $i$.
% It is worth noting that in Assumption \ref{a1}, conditional on the individual and time effects $\{\alpha_1,\ldots,\alpha_N,\gamma_1,\ldots,\gamma_T\}$ is equivalent to conditional on $\{\alpha_i,\gamma_t\}$. %, provided that $\{\alpha_i\}_{i=1}^N$ is independent across $i$ and $\{\gamma_t\}_{t=1}^T$ is independent over $t$.
% Therefore, in Definition \ref{dep}, the functional dependence measure $\delta_{it,k,r}(\ell)$ is a random function of $\alpha_i$ and $\gamma_t$, and we define $\Delta_{k,r,m}$ (as well as the dependence adjusted norm) as a uniform measure over $i$ and $t$.
% }

\begin{Assumption}[Data Generating Processes, Continued]\phantomsection
\label{a2}
\begin{enumerate}
    \item[(i)] For each $k=1,\ldots,d$, assume that $\|X_{\cdot,k}\|_{r,\varsigma}<\infty$ for some $r\geq4$, $\varsigma>0$, and
    $$\|X_{\cdot,k}\|_{\psi_\nu,\varsigma}\defeq \sup_{r\geq2}r^{-\nu}\|X_{\cdot,k}\|_{r,\varsigma}<\infty, \text{ for some } \nu\geq0,\varsigma>0.$$
    Specifically, $\|X_{\cdot,k}\|_{\psi_\nu,\varsigma}$ is the dependence adjusted sub-Gaussian or sub-exponential norm, with $\nu$ taking values of 1/2 or 1, respectively.
    % \item[(ii)] Let $X_{it}^m\defeq\E(X_{it} \mid \xi_{i,t-m},\zeta_{i,t-m},\ldots,\xi_{it},\zeta_{it})$ for $m\in\mathbb N$, and assume that for some $r\geq4$,
    % $$\max_{1\leq i\leq N,1\leq t\leq T} \|X_{it}-X_{it}^m\|_{r}\leq m^{-\varphi}\vartheta, \text{ for some }\vartheta>0,\varphi\geq0.$$
    \item[(ii)] $\vps_{it}$ is a martingale difference sequence (m.d.s.) over $t$ with respect to the filtration $\mathcal{F}_{it}=\{(X_{is})_{s=1}^{t},(Y_{is})_{s=1}^{t-1}\}$, i.e. $\E(\vps_{it}\mid \mathcal F_{it})=0$, and has constant unconditional variance. There exists a constant $\bar\sigma>0$ such that $\E(\vps_{it}^2\mid \mathcal F_{it})\leq\bar\sigma^2$ for all $i=1,\ldots,N,t=1,\ldots,T$. %\footnote{{\red We note that the m.d.s. condition can be further generalized by considering the reduced filtration $\mathcal{F}_{it}=\{(X_{is})_{s=1}^{t-1},(Y_{is})_{s=1}^{t-1}\}$.}}
    An analogous assumption to part (i) for $X_{it,k}$ holds for $\vps_{it}$.
    \item[(iii)] The sub-Gaussian norms $\max\limits_{1\leq k\leq d}\|X_{it,k}\|_{\psi_{1/2}}<\infty$, and $\|\varepsilon_{it}\|_{\psi_{1/2}}<\infty$, for all $i=1,\ldots,N,t=1,\ldots,T$.
    \item[(iv)] For each component $k=1,\ldots,d$, the degree of predeterminedness is summable with respect to the lag, that is, 
    $\sum\limits_{j=1}^{T-1}\sup\limits_{1\leq i\leq N,1\leq t\leq T-j}|\E(X_{i,t+j,k}\vps_{it})|<\infty$.
\end{enumerate}
\end{Assumption}

%Consider an AR(1) process: $X_t=\alpha X_{t-1} + \xi_t=\sum_{\ell=0}^\infty\alpha^\ell\xi_{t-\ell}$, where the assumptions in part (i) and (ii) are satisfied, with $|\alpha|<1$ and sub-Gaussian error $\xi_t$.
Example \ref{ex:ar1} provides an example of a univariate heterogeneous linear process that satisfies Assumption \ref{a2}(i). % and (ii).
In Example \ref{dyna}, we verify Assumptions \ref{a1}--\ref{a2} for a basic panel AR(1) model.

\begin{Example}[Heterogeneous Linear Process]\label{ex:ar1} Assume that $X_{it}$ is univariate. For each $i = 1,\ldots, N$, consider the linear process: %\textcolor{blue}{(If $X_{it}$ is iid $i$, should we consider different coefficients of $i$?) $X_{it}$ is iid $i$ is i.i.d. conditional on individual effects}
\begin{equation*}
X_{it} = \sum_{\ell \geq 0} a_{i\ell} \xi_{i,t-\ell}, \quad 1 \leq t \leq T,
\end{equation*}
where the coefficients $a_{i\ell}$ can be heterogeneous over $i$ and $\ell$, and satisfy $|a_{i\ell}|\leq |c|^\ell$, for some $|c|<1$ and %$\sum_{\ell \geq 0} |\alpha_{i\ell}| < \infty $
for all $i$ and $\ell$. The unobservable $\xi_{it}$'s are sub-Gaussian random variables that are i.i.d.\ over $i$ and $t$, and have finite $r$th-moment for some $r\geq4$.
It follows that
\begin{align*}
% &\delta_{it, r} = \|X_{it}^\ast - X_{it}\|_r = \|a_{it}\xi_{i0}^\ast-a_{it}\xi_{i0}\|_r=|a_{it}|\|\xi_{i0}^\ast-\xi_{i0}\|_r,\\
% %\leq \max_{1 \leq i \leq N} |\alpha_{it}|
% &\Delta_{r,m} =\max_{1\leq i\leq N}\sum_{t\geq m}\delta_{it,r} %= \|\xi_{i0}^\ast-\xi_{i0}\|_r\max_{1\leq i\leq N}\sum_{t\geq m}|\alpha_{it}|
% \leq\|\xi_{i0}^\ast-\xi_{i0}\|_r\sum_{t\geq m}|c|^t\propto|c|^m,\\
&\delta_{it, r}(\ell) = \|X_{it}^\ast(\ell) - X_{it} \|_r = \|a_{i\ell}\xi_{i,t-\ell}^\ast-a_{i\ell}\xi_{i,t-\ell}\|_r=|a_{i\ell}|\|\xi_{i,t-\ell}^\ast-\xi_{i,t-\ell}\|_r,\\
%\leq \max_{1 \leq i \leq N} |\alpha_{it}|
&\Delta_{r,m} =\sum_{\ell\geq m}\max_{1\leq i\leq N,1\leq t\leq T}\delta_{it,r}(\ell) %= \|\xi_{i0}^\ast-\xi_{i0}\|_r\max_{1\leq i\leq N}\sum_{t\geq m}|\alpha_{it}|
\leq\sum_{\ell\geq m}|c|^\ell\max_{1\leq i\leq N,1\leq t\leq T}\|\xi_{i,t-\ell}^\ast-\xi_{i,t-\ell}\|_r\propto|c|^m,\\
&\|X_{\cdot}\|_{r,\varsigma} = \sup_{m\geq0}(m+1)^\varsigma\Delta_{r,m}<\infty, \text{ for some }r\geq4,\varsigma>0,\\
&\|X_{\cdot}\|_{\psi_{1/2},\varsigma}=\sup_{r\geq2}r^{-\nu}\|X_{\cdot,k}\|_{r,\varsigma}<\infty, \text{ for some }\varsigma>0.
%\leq  \max_{1 \leq i \leq N} \sum_{t\geq m} |\alpha_{it}|,
\end{align*}
%and $\Delta_{r,m} \leq  \max_{1 \leq i \leq N} \sum_{t\geq m} |\alpha_{it}| $.
% \textcolor{blue}{Finite sum $\sum_{t\geq m}|\alpha_{it}|$ is not sufficient, it should decay with m, e.g. $|a_{it}|\leq |c|^t$, $|c|<1$? To make the second part in (i) holds we also need e.g. sub-Gaussian shocks.}

% As for part (ii), consider the case of $m=1$. It can be seen that
% $$\max_{1\leq i\leq N,1\leq t\leq T}\|X_{it}-X_{it}^m\|_{r}=\max_{1\leq i\leq N,1\leq t\leq T}\Big\|\sum_{\ell \geq 2}a_{i\ell} \xi_{i,t-\ell}\Big\|_r\lesssim \max_{1\leq i\leq N}\sum_{\ell \geq 2}|a_{i\ell}|<\infty.$$

In the special case of a heterogeneous over $i$ and stationary over $t$ AR(1) process: $X_{it} = \beta_i X_{i,t-1} + \xi_{it}$, with $|\beta_i|<1$, we have $a_{it} = \beta_{i}^{t}$, and $\Delta_{r,m} \propto \max\limits_{1\leq i\leq N}|\beta_i|^m$. \qed
%$\Delta_{1,r,m} \lesssim \max_{1 \leq i \leq N} (1-a_i)^{-1}$.
\end{Example}

\begin{Example}[Panel AR(1) model]\label{dyna}
Consider a panel AR(1) model:
\begin{equation}\label{maindyna}
Y_{it}=\alpha_i+\gamma_t + \theta_1^0 Y_{i,t-1}+\theta_2^0 D_{it}+\varepsilon_{it}, \quad |\theta_1^0|<1.
\end{equation}
Assume that $\{(D_{it},\varepsilon_{it}):1\leq t\leq T\}$ are independent across $i$, $\varepsilon_{it}$ has zero mean and is independent over $t$, and $D_{it}$ is predetermined with respect to $\varepsilon_{it}$ such that  $\E(D_{it}\varepsilon_{is})=0$ for $t\leq s$. If $D_{it}$ follows the heterogeneous linear process introduced in Example \ref{ex:ar1}, recursive substitution yields
\begin{align*}
Y_{it}&=\frac{\alpha_i}{1-\theta_1^0}+\sum_{s\geq0}(\theta_1^0)^s(\gamma_{t-s} + \theta_2^0 D_{i,t-s}+\varepsilon_{i,t-s})\\
&=\frac{\alpha_i}{1-\theta_1^0}+\sum_{s\geq0}(\theta_1^0)^s\bigg\{\gamma_{t-s} +\theta_2^0 \bigg(\sum_{l\geq 0} a_{il} \xi_{i,t-s-l}\bigg)+\varepsilon_{i,t-s}\bigg\}.
\end{align*}
Conditional on $\alpha_i$ and $\gamma_t,\gamma_{t-1},\ldots$, $Y_{it}$ is stationary over $t$ and can be expressed as a measurable function of the i.i.d.\ innovations $(\ldots,\xi_{i,t-1},\varepsilon_{i,t-1},\xi_{it},\varepsilon_{it})$, thereby satisfying Assumption \ref{a1}.%\textcolor{red}{[IFV: should we also condition on $\gamma_t,\gamma_{t-1},\ldots$?]}

To verify the assumptions on the dependence adjusted norm, consider a change in the innovations at period $t-\ell$. Consequently,
\begin{align*}
\|Y_{it}^\ast(\ell) - Y_{it} \|_r & \leq |\theta_2^0|\sum_{s=0}^\ell|\theta_1^0|^s |a_{i,\ell-s}|\|\xi_{i,t-\ell}^\ast-\xi_{i,t-\ell}\|_r + |\theta_1^0|^\ell\|\varepsilon_{i,t-\ell}^\ast-\varepsilon_{i,t-\ell}\|_r\\
&\leq |\theta_2^0|\sum_{s=0}^\ell|\theta_1^0|^s\|\xi_{i,t-\ell}^\ast-\xi_{i,t-\ell}\|_r + |\theta_1^0|^\ell\|\varepsilon_{i,t-\ell}^\ast-\varepsilon_{i,t-\ell}\|_r,
\end{align*}
given that $|a_{il}|\leq |c|^l$ for some $|c|<1$. Suppose the innovations have finite $r$th-moment for some $r\geq4$, by arguments similar to Example \ref{ex:ar1}, we obtain $\|Y_{\cdot}\|_{r,\varsigma}<\infty$ and $\|Y_{\cdot}\|_{\psi_{1/2},\varsigma}<\infty$ for some $r\geq4,\varsigma>0$. Assumption \ref{a2}(i) is verified. Assumptions \ref{a2}(ii)-(iii) hold in this simple example provided that $\xi_{it}$'s and $\varepsilon_{it}$'s are i.i.d.\ sub-Gaussian random variables. Lastly, regarding Assumption \ref{a2}(iv), when the covariates are lags of $Y_{it}$, the condition holds naturally given $|\theta_1^0|<1$. For other exogenous covariates, it follows from 
$$\sum_{j=1}^{T-1}\sup_{1\leq i\leq N,1\leq t\leq T-j}\bigg|\sum_{\ell\geq 0}a_{i\ell}\E(\xi_{i,t+j-\ell}\vps_{it})\bigg|<\infty,$$
which requires that the overlap in the innovations is summable with respect to the lag.\qed
\end{Example}

%{\red \st{We note that the m.d.s. assumption mentioned above is stronger than the uncorrelated errors as assumed in \mbox{\citet{arellano1991some}}.} }
The m.d.s.\ condition in Assumption \ref{a2}(ii) aligns with the standard large $T$ panel literature; see, for example,  \citet{alvarez2003time} and \citet[page 145]{arellano2003panel}.
%{\color{red}However, it's noteworthy that our theoretical analysis does not heavily rely on this assumption. 
Using standard techniques such as $m$-dependent approximation and blockwise time series analysis, as in \cite{chen2022inference}, the setting can be generalized. %and obtain the same convergence rate without further theoretical advancements.} 
Additionally, for practitioners, we note that the sub-Gaussian conditions in Assumption \ref{a2}(iii) rule out heavy-tail distributions for $X_{it}$ and $\varepsilon_{it}$. However, this assumption is not critical for our analysis. They can be relaxed to a polynomial tail conditions with more demanding rate assumptions.

\subsection{Basic Model: Consistency of Step 1}
\label{firststep}
We will first demonstrate the consistency property of the LASSO estimator $\widehat\Pi_t$, which is obtained in step 1 of AB-LASSO by \eqref{LASSO}. For this purpose, a few definitions and assumptions are introduced as follows.

For each component $W_{it}$ of $\Delta\widetilde X_{it}$, denote the $N\times1$ vector containing $(W_{it})_{i=1}^N$ by $\boldsymbol{W}_t$.
%For simplicity let $D_{it}$ to be one variable.
%Let $V_{it}= [Y_{i,t-k-1}, D_{it-k}, \cdots, Y_{i1}, D_{i2},D_{i1}]$ as a $[(t-k-1)*2 +1]\times 1$, ({define $\ell_t = (t-k-1)*2+1$ dimensional  vector}), and $V_{t}$ as $[(t-k-1)*2+1] \times N$.
Recall that $V_{it} = (1,X_{i1}^\top,\ldots,X_{it}^\top)^\top$. Denote the dimension of $V_{it}$ as $m_t$, which is the number of instruments for each time period $t$. We further stack $V_{it}^\top$ by rows for all $i=1,\ldots,N$, to create the $N\times m_t$ matrix $\boldsymbol{V}_t$. For each $t=1,\ldots,T-1$, define the linear projection for step 1:
$$\bm W_t = \bm V_t\Pi_t^{0}+ \bm\eta_t,$$
with $\Pi_t^{0} \defeq \bar\E(V_{it}V_{it}^{\top})^{-1}\bar\E(V_{it}W_{it})$, and $\boldsymbol{\eta}_t$ is an $N\times1$ vector of errors $(\eta_{it})_{i=1}^N$. %\textcolor{red}{[IFV: i think the independence assumption is problematic when time effects have been partialed-out. is ot needed?]} 
Assumptions \ref{a1}--\ref{a2} guarantee that each component in $V_{it}$ satisfies the finite moment conditions $\E|V_{it,k}|^{2r}<\infty$ and $\E|V_{it,k}\eta_{it}|^r<\infty$, for some $r\geq2$, where $k=1,\ldots,m_t$, $i=1,\ldots,N$, $t=1,\ldots,T-1$.

In addition to the population OLS coefficients $\Pi_t^0$, %given $0<s_t^*\leq m_t$,
we consider a sparse approximation $\Pi_t^*= \arg\min\limits_{|\Pi_t|_0\leq s_t^*}\bar\E|V_{it}^\top(\Pi_t- \Pi_t^0)|^2$, where $s_t^*$ (such that $0<s_t^*\leq m_t$) characterizes the sparsity of $\Pi_t^*$. %The oracle order of $s_t^*$ 
This sparsity level is determined by the degree of temporal dependency in the data. %which we will discuss in specific cases later.
 For instance, using a special case of the panel AR(1) model from Example \ref{dyna}, we demonstrate that the coefficients in $\Pi_t^0$ exhibit geometric decay governed by the autoregressive coefficient $\theta_1^0$, which leads to the approximately sparse structure of the moment conditions.

\begin{Proposition}[Approximate sparsity of step 1 coefficients]\label{prop:decayK0}
Recall the panel AR(1) model in Example \ref{dyna}, excluding the time effects $\gamma_t$. Assume in addition that %$D_{it}$ and $\varepsilon_{it}$ are stationary processes independent of $\alpha_i$, and that
$D_{it}$ is i.i.d.\ over $t$ and predetermined in the sense that $\E(D_{it}\varepsilon_{i,t-1})\neq 0$ and $\E(D_{it}\varepsilon_{is})=0$ for all $s\neq t-1$.
% Then, there exists a constant $C<\infty$ such that the coefficients in $\Pi_t^0$ decay geometrically:
% $$|\Pi^0_{t,k}|\leq C|\theta_1^0|^{|t-1-k|}, \quad k=1,\ldots,m_t,\,t=1,\ldots, T-1.$$
Then $\Pi_t^0=(\pi_{t0}^0,\pi_{t1}^{0\top},\ldots,\pi_{tt}^{0\top})^\top$ satisfies $|\pi_{ts,k}^0|\lesssim |\theta_1^0|^{t-s}$ for $s=1,\ldots,t$ and $k=1,2$. Hence, $\Pi_t^0$ is approximately sparse.
\end{Proposition}
%IFV: can we replace the assumption $D_{it}$ is i.i.d.\ over $t$  by Assumption 3.2(iv)?}
We illustrate the sparsity structure in Proposition \ref{prop:decayK0} with a special case for analytical tractability. %In Example \ref{dyna}, $W_{it}$ has two components, $\Delta Y_{i,t-1}$ and $\Delta D_{it}$.
In this example, the projection coefficients %of $\Delta Y_{i,t-1}$ decay geometrically and therefore are approximately sparse, whereas the projection coefficients of $\Delta D_{it}$ 
are zero for regressors beyond %the first lag 
$X_{it}=(D_{it},Y_{i,t-1})^\top$ and are therefore exactly sparse. This exact sparsity arises from the serial independence of $D_{it}$. If $D_{it}$ is serially correlated, for example follows an ARMA(1,1) process with innovations $v_{it}$ correlated with $\vps_{is}$ for $s<t$, then the projection coefficients %of $\Delta D_{it}$ also 
tail off as the time distance increases and become approximately sparse.
%{\red IFV: i think the old version of the previous paragraph was more clear. Was it something wrong with the old version?}

% in the regression of $\Delta Y_{i,t-1}$ on lagged instruments.
% In this example, the population coefficients in the linear projection of $\Delta D_{i,t}$ onto lagged instruments trivially truncate beyond $X_{it}=(D_{it},Y_{i,t-1})^\top$. \textcolor{blue}{(no time effects in this example, no need for demeaning)}
% More generally, allowing for serial correlation in $D_{it}$, for example, an AR(1) process incorporating feedback from $Y_{i,t-1}$ would also lead to approximately sparse coefficients for this component.
%In this example, the coefficients on $(D_{it},\ldots,D_{i1})^\top$ truncate beyond $D_{i,t-1}$. Allowing for serial correlation in $D_{it}$ would instead lead these coefficients to decay geometrically. Moreover, incorporating feedback from $Y_{i,t-1}$ to $D_{it}$ would imply that the projection of $\Delta\widetilde D_{it}$ is also approximately sparse. Finally, a nonzero mean of $D_{it}$ does not affect the decay structure.

To quantify the approximation error of $\Pi_t^*$ with respect to $\Pi_t^0$ empirically, we consider the prediction norm defined by
\beq\label{cs}
|\Pi_t^* - \Pi_t^0|_{2,N} \defeq N^{-1/2}|\bm V_{t}(\Pi_t^* - \Pi_t^0)|_2=\Big[\frac{1}{N}\sum_{i=1}^N\{V_{it}^\top(\Pi_t^* - \Pi_t^0)\}^2\Big]^{1/2}=:C_{s_t^*}.
\eeq
We shall express all the general rate of convergence results in terms of $C_{s_t^*}$.
%Moreover, the following lemma illustrates the approximate sparsity of the coefficients.
 Under the approximate sparse condition, Theorem \ref{bound} in Appendix \ref{oracle} shows the oracle order $s_t^*\asymp \log N \wedge t$, which represents the optimal solution to the risk minimization problem. Consequently, the corresponding oracle bound for the approximation error follows as $C_{s_t^*}\lesssim_\P \sqrt{(\log N \wedge t)/N}$.
%\begin{lemma}[Oracle Order of $s_t^*$]\label{bound2}
%Under Assumptions \ref{a1}-\ref{a2}, and assuming that $\bar\E(V_{it}^\top\delta_{-J_t}^0)\lesssim c^{-s_t}$ for some constant $c>0$, we can conclude that the optimal $s_t^*$ is bounded as $s_t^*\asymp\log N\wedge t$, and $C_{s_t^\ast} \asymp \sqrt{(\log N \wedge t)/N}$.
%\end{lemma}
%It is worth noting that Assumption $\bar{\mathbb{E}}(V_{it}^\top \delta_{-J_t}^0) \lesssim c^{-s_t}$ is satisfied in many cases of practical interest. We verify this condition for the model specified in Equation~(\ref{maindyna}) (see Appendix~\ref{oracle}).

% {\red In Appendix \ref{oracle}, we show that for a specific model, the oracle order is bounded as $s_t^* \lesssim \log (N \wedge t)$, and since $C_{s_t^*}\lesssim_\P c^{-2s_t^*}+ c^{-s_t^*}N^{-1/2} \sqrt{\log s_t^*} + N^{-1}s_t^* $, we also have $C_{s_t^*} \lesssim_\P \sqrt{\log (N \wedge t) /N}$.}

Let $\Pi_{t,k}^*$ be the $k$-th element of $\Pi_t^*$, $k=1,\ldots,m_t$. Define the indices sets $J_t\defeq\{k\in\{1,\ldots,m_t\}:\Pi_{t,k}^*\neq0\}$ and $J_t^{c}\defeq\{k\in\{1,\ldots,m_t\}:\Pi_{t,k}^*=0\}$. For any $\delta_t\in\R^{m_t}$, let $J_{t,0}\subseteq\{1,\ldots,m_t\}$ be a set of indices with cardinality $|J_{t,0}|\leq s_t^*$, and let $J_{t,1}\subseteq\{1,\ldots,m_t\}$ be the set of indices corresponding to the $s_t^*$ largest in absolute value coordinates of $\delta_t$ outside of $J_{t,0}$. In the case of $s_t^*>m_t-|J_{t,0}|$, it corresponds to the $m_t-|J_{t,0}|$ largest absolute values. Let $J_{t,01}\defeq J_{t,0}\cup J_{t,1}$. Define $\delta_{t,J_t}$ as the sub-vector of $\delta_t$ corresponding to $J_t$, and define $\delta_{t,J_t^{c}}$ and $\delta_{t,J_{t,01}}$ similarly.

To show identification of $\Pi_t^*$, we consider two events (for each $t$) associated with the restricted eigenvalue (RE) conditions, as outlined in Section 3 of \citet{bickel2009simultaneous}. For $c_0>0$, define
\begin{eqnarray*}
{\mathcal{A}_{1t}} &\defeq& \Big\{\min_{\delta_{t}\neq 0, |\delta_{t}|_0\leq s_t^*,|\delta_{t,J_t^c}|_1\leq c_0 |\delta_{t,J_t}|_1 } \frac{|\bm V_t{\delta}_{t}|_{2}}{\sqrt{N}|\delta_{t, J_t}|_2}\geq \kappa_t(c_0,s_t^*)\Big\},\\
{\mathcal{A}'_{1t}} &\defeq& \Big\{\min_{\delta_{t}\neq 0, |\delta_{t}|_0\leq s_t^*,|\delta_{t,J_t^c}|_1\leq c_0 |\delta_{t,J_t}|_1 } \frac{|\bm V_t{\delta}_{t}|_{2}}{\sqrt{N}|\delta_{t, J_{t,01}}|_2}\geq \kappa_t(c_0,s_t^*)\Big\},
\end{eqnarray*}
where $\kappa_t(c_0,s_t^*)$ represents positive constant that depends on $c_0$ and $s_t^*$. In Lemma \ref{idlemma}, we will prove that these events occur with probabilities approaching 1 as $N\to\infty$, for some $\kappa_t(\cdot)>0$ related to the RE condition in population, as per Assumption \ref{33}.

\begin{Assumption}[RE Condition]\label{33}
For any constant $c_0>0$, define the subspace
$$\Omega_t(c_0,s_t^*)\defeq \{\delta_t/|\delta_t|_2:\delta_t\in\R^{m_t},\delta_t\neq0,|\delta_t|_0\leq s_t^*,|\delta_{t,J_t^c}|_1\leq c_0|\delta_{t,J_t}|_1\}.$$ Assume that there exist some positive constants $C_{\min}$ and $C_{\max}$ such that
$$ C_{\min} \leq \min_{1\leq t\leq T-1} \min_{\delta_{t}\in \Omega_t(c_0,s_t^*)} {\delta_{t}^{\top}\bar\E(V_{it}V_{it}^{\top})}\delta_{t} \leq  \max_{1\leq t\leq T-1} \max_{\delta_{t}\in \Omega_t(c_0,s_t^*)} {\delta_{t}^{\top}\bar\E(V_{it}V_{it}^{\top})}\delta_{t} \leq C_{\max}.$$
\end{Assumption}
%To examine the practical validity of the RE condition, we analyze a simple example provided in Example \ref{reexample} in the appendix.

\begin{Lemma}[Identification]\label{idlemma}
Under Assumptions \ref{a1}--\ref{a2}, if Assumption \ref{33} holds with $C_{\min} = \min\limits_{1\leq t\leq T-1}\kappa_t^2(c_0,s_t^*)-\Delta_{N,T}$ and $C_{\max} = \max\limits_{1\leq t\leq T-1}\kappa_t^2(c_0,s_t^*)+\Delta_{N,T}$, where $\min\limits_{1\leq t\leq T-1}\kappa_t(c_0,s_t^*)>0$ and $\Delta_{N,T}\defeq\max\limits_{1\leq t\leq T-1} \sqrt{s_t^*}\log m_t/\sqrt{N}\to0$ as $N,T\to\infty$, then, for each $t$,
\begin{equation*}
\min_{\delta_{t}\neq 0, |\delta_{t}|_0\leq s_t^*,|\delta_{t,J_t^c}|_1\leq c_0 |\delta_{t,J_t}|_1 } \frac{|\bm V_t{\delta}_{t}|_{2}}{\sqrt{N}|\delta_{t,J_t}|_2}\geq \kappa_t(c_0,s_t^*),
\end{equation*}
holds with probability $1-\smallO(1)$, as $N\to\infty$.
\end{Lemma}
Lemma \ref{idlemma} shows that $\P(\mathcal A_{1t})\to1$ as $N\to\infty$ for each $t$, which is in line with Lemma 1 of \citet{belloni2013least}. We can similarly verify that $\P(\mathcal A'_{1t})\to1$ as $N\to\infty$.  These results ensure that the Gram matrix $V_{it}V_{it}^\top$ is well conditioned along the sparse directions and establish the identification of the sparse solution $\Pi_t^*$ within the subspace, provided such a solution exists.
% To qualify the difference between $\Pi_t^*$ and the true $\Pi_t^0$, we consider the prediction norm defined by
% \beq\label{cs}
% |\Pi_t^* - \Pi_t^0|_{2,N} \defeq N^{-1/2}|V_{t}^\top(\Pi_t^* - \Pi_t^0)|_2=\Big[\frac{1}{N}\sum_{i=1}^N\{V_{it}^\top(\Pi_t^* - \Pi_t^0)\}^2\Big]^{1/2}=:C_{s_t^*}.
% \eeq
% In Appendix \ref{oracle}, we show that for a specific example, the oracle order of sparsity is bounded as $s_t^* \asymp \log N \wedge t$.%, and since $C_{s_t^*}\lesssim_\P c^{-2s_t^*}+ c^{-s_t^*}N^{-1/2} \sqrt{\log m_t} + N^{-1}s_t^* $, we also have $C_{s_t^*} \lesssim_\P \sqrt{(\log N \wedge t) /N}$.

Recall the LASSO estimator $\widehat\Pi_t$ obtained by \eqref{LASSO}. To achieve good prediction performance of the estimator, properly chosen penalty tuning parameters and weights are necessary. For each $t=1,\ldots,T-1$, let $\bm\omega_t$ be an $m_t\times1$ vector, with the first element being 1 and the remaining elements collecting the non-negative penalty weights $(\omega_{ts}\mathbf{1}_s)_{s=1}^{t}$, where $\mathbf{1}_s$ represents a vector of ones with the same dimension as $X_{is}$.
\begin{Assumption}[Penalty Parameters]\label{weights}
The penalty tuning parameter $\lambda_t>0$ is selected such that the event
$$\mathcal A_{2t}\defeq c|\bm V_t^\top\bm\eta_t\oslash\bm\omega_t|_\infty\leq\lambda_t$$
holds with probability at least $1-\alpha$, for a constant $c>1$ and $0<\alpha<1$. Here, $\oslash$ represents the Hadamard division, i.e. element-wise division. Moreover, assume that $|\bm\omega_{t}|_{\infty}$ is bounded by a constant. %, $|\bm\omega_{t,J_t}|_2\leq \sqrt{s_t^*}$, where $\bm\omega_{t,J_t}$ is the sub-vector of $\bm\omega_t$ corresponding to $J_t$.
\end{Assumption}

% According to the Karush-Kuhn-Tucker conditions of LASSO, the solution $\widehat\Pi_t$ satisfies
% $$|\{\bm V_t^\top(\bm W_t-\bm V_t\widehat\Pi_t)\}\oslash\bm\omega_t|_\infty\leq\lambda_t,$$
% where $\oslash$ represents the Hadamard division, i.e. element-wise division. On the other hand, to ensure the true $\Pi_t^0$ is feasible for the LASSO problem with high probability,
% %the constraint $|\bm\omega_t\circ\Pi_t^0|_1\leq\lambda_t$, where $\circ$ denotes the Hadamard product,
% we need the event
% $$\mathcal A_{2t}\defeq\{c|\bm V_t^\top\bm\eta_t\oslash\bm\omega_t|_\infty\leq\lambda_t\}$$
% to occur with high probability, where $c>2$ is a constant.
%[IFV: this sentence needs to be rewritten. How about: Assumption \ref{weights} ensures that $\Pi_t^0$ is a feasible solution to the constrained optimization representation of the LASSO problem with probability at least $1-\alpha$? We might want to add a footnote with the constrained representation.]
Assumption \ref{weights} is consistent with Eq. (2.3) in \citet{bickel2009simultaneous} and serves as a necessary condition for $\Pi_t^0$ to be the minimizer of the LASSO problem expressed in \eqref{LASSO} with probability at least $1-\alpha$. This assumption is crucial for establishing the consistency of our estimator and implies that an ideal choice of the tuning parameter $\lambda_t$ is given by the ($1-\alpha$) quantile of the random variable $c|\bm V_t^\top\bm\eta_t\oslash\bm\omega_t|_\infty$.\footnote{Empirically, since $\bm\eta_t$ is unobserved, the tuning parameter $\lambda_t$ can be selected either based on quantiles of the standard normal distribution or through a more data-dependent approach using the multiplier bootstrap, as discussed in \citet{lasso2018}. In our empirical analysis, to avoid over-fitting, we adopt a data-independent choice of $\lambda_t$, which is more conservative, and account for heteroskedasticity through the penalty weights $\bm\omega_t$. Further details regarding the practical selection of $\lambda_t$ and $\bm\omega_t$ are provided in Section \ref{sim}.}
Under Assumptions \ref{a1}--\ref{a2} and \ref{weights}, we can apply Lemma \ref{emp}, which provides the maximal tail probability for the partial sum of the $m_t$-dimensional process $\varpi_{it}\defeq V_{it}\eta_{it}\oslash\bm\omega_t$, %we can easily verify that the condition in Assumption \ref{a2}(i) is met for the $m_t$-dimensional process $\varpi_{it}\defeq V_{it}\eta_{it}\oslash\bm\omega_t$. This enables us to apply the maximal tail probability for high dimensional partial sums, such as %Lemma B.4 in \citet{lasso2018_sup},
to derive an upper bound for the ideal choice of $\lambda_t$,
%$$\lambda_t\lesssim\sqrt{N}(\log m_t)^{1/\gamma}\max_{1\leq k\leq m_t}\|\varpi_{\cdot,k}\|_{\psi_\nu,\varsigma}, \quad \gamma=2/(2\nu+1).$$
which is of order $\sqrt{N\log m_t}$.

Lastly, to conclude the consistency of the LASSO estimators, we present the prediction performance bounds for $\delta_{\Pi,t}\defeq\widehat\Pi_t-\Pi_t^*$ in the following theorem. This will be combined with the prediction norm of the approximation error in \eqref{cs}             to derive a performance bound for $\widehat\Pi_t-\Pi_{t}^0$  using the triangle inequality.

% {\color{red}
% Thus as the high quantile of the term $\sqrt{N}\max_{t, 1\leq k \leq K}|S_{tk}{/\Psi_{tk}}|$, we need\\
% $ \max_{k,t}(\|V_{it,k}\eta_{it}\|_q/\Psi_{tk}) (N^{1/q} t^{1/q} \vee  (\log t)^{1/2}N^{1/2} )\lesssim \lambda_t $ to ensure the event $\P(\mathcal{A}_{2t}) = 1-\Co_p(1).$
% Let $B_{N,T}\defeq \max_{2\leq t\leq T}\{C_{s_t^*}+  2\sqrt{s_t^*}\lambda_t/(N\kappa_t(3,s_t^*))\}$.}
%{\color{red}The following lemma provide prediction bound of LASSO for each time $t$. }
\begin{Theorem}[Prediction Performance Bounds of LASSO]\label{lassobound}
Under the same assumptions as in Lemma \ref{idlemma} and Assumption \ref{weights}, %on the event $\mathcal A_{2t}$ for each $t=2,\ldots,T$,
we can conclude, with probability at least $1-\alpha-\smallO(1)$,
%On the event  ${\mathcal{A}_{1t}}$  , ${\mathcal{A}_{2t}}$ and conditions  \ref{a1}- \ref{33}, we have,
\beq %\label{bound1}
|\delta_{\Pi,t}|_{2,N}\lesssim 2C_{s_t^*}+  N^{-1}\sqrt{s_t^*}\lambda_t/\kappa_t(3,s_t^*),\notag
%|X_t^{\top}{\delta}_{\Pi,t}|_{2,N} \leq C_{s_t^*}+  2\sqrt{s_t^*}\lambda_t/(N\kappa_t(3,s_t^*)),
\eeq
\begin{equation*}
|\delta_{\Pi,t}|_1 \lesssim 7\sqrt{s_t^*} \{2C_{s_t^*}+  N^{-1}\sqrt{s_t^*}\lambda_t/\kappa_t(3,s_t^*) \}/\kappa_t(3,s_t^*) +NC_{s_t^*}^2/\lambda_t,
%|\delta_{\Pi,t}|_1 \leq (1+6)\sqrt{s_t^*} \{C_{s_t^*}+  2\sqrt{s_t^*}\lambda_t/(N\kappa_t(3,s_t^*)) \}/\kappa_t(3,s_t^*) +7/6 (4N/\lambda_t) C_{s_t^*}^2.
\end{equation*}
where the implicit constant in ``$\lesssim$'' depends only on the constant $c>1$ satisfying Assumption \ref{weights}.
% If $\P(\cap_t(\mathcal{A}_{1t} \cap {\mathcal{A}_{2t}})) \to 1$, then
% \begin{eqnarray*}
% &&\max_t|\delta_{\Pi,t}|_1 \leq B_{N,T} +7/6 (4N/\min_t \lambda_t) C_{s_t^*}^2 \\
% &&\lesssim  B_{N,T} +7/6 (4N/[\min_{t\geq 3} \max_{k}(\|V_{it,k}\eta_{it}\|_q/\Psi_{tk}) (N^{1/q} t^{1/q} \vee  {\log t}^{1/2}N^{1/2} )]) C_{s_t^*}^2.
% \end{eqnarray*}
% Besides, on  $\mathcal{A}_{1t}'$,
% \begin{eqnarray*}
% \max_t|\delta_{\Pi,t}|_2 \leq (1+ c_0)  \{C_{s_t^*}+  2\sqrt{s_t^*}\lambda_t/(N\kappa_t(3,s_t^*))\}/(\min_t\kappa_t(3,s_t^*,s_t^*)).
% \end{eqnarray*}
\end{Theorem}

Based on Theorem \ref{lassobound}, Corollary \ref{joint} provides the joint prediction performance bounds, where the $\ell_2$-norm bound is derived by following Theorem 7.2 of \citet{bickel2009simultaneous}.

%\begin{Assumption}[Sparsity Rate]\label{a4}
%$\eta_{it}, V_{it,k}$ are iid over $i$.
%${[\log (N\vee T) s_t^*]}^{1/2} /{N^{1/2}}\to 0$.

%Denote $\max_k\E(\eta_{it}^2V_{it,k}^2)= \sigma_{\eta, x,t}^2$.
%For two positive constants $C_{\max}$ and $c_{\min}$,
%Let $\Omega$ be a compact set where $\delta$ take values on.
%Assume that $c_{\min}\leq \min_{\delta\in \Omega } N^{-1}\sum_t{\delta^{\top}\E(X_tX_t^{\top})}\delta \leq C_{\max}$.
%$ c_{\min}\leq \lambda_{\max}(\E(X_tX_t^{\top}))\leq \lambda_{\max}(\E(X_tX_t^{\top})) \leq C_{\max}$
%{\color{red}Assume $c_{\delta}([\sqrt{N}^{-1} \sqrt{\log( N \vee T) s^*}] \to 0$ for a positive constant $c_{\delta}$.}
%Assume that $\max\limits_{2\leq t\leq T} s_t^*(\log m_t)^{3/2}/\sqrt{N}
% \to 0$, as $N,T \to \infty$.
%\end{Assumption}

%{\color{red} CHANGE To change $s_T^*$}

% \begin{remark}
% For the joint event over $t$ let $\max_t \lambda_t \asymp_p \max_{t,k}(\|X_{i,t,k}\eta_{it}\|_q/\Psi_{tk}) (N^{1/q} T^{1/q} \vee  {(\log T)}^{1/2}N^{1/2} )$.
% We shall see that for sufficient large $q$, we have\\ $\max_t|\delta_{\Pi,t}|_{\infty} \leq \max_t|\delta_{\Pi,t}|_2\lesssim_p\sqrt{s_t^*}\sqrt{\log T} /N^{1/2}$, the proof follows from Page 1729 in the Proof of Theorem 7.1 \cite{bickel2009simultaneous}.
% $\max_t|\delta_{\Pi,t}|_{1} \lesssim_p s_t^*\sqrt{\log T} /N^{1/2}$. We shall use the uniform rate in the next steps.
% {\color{red}
% We shall also note that under Assumptions \ref{a4}, the rate on the right hand side is tending to $0$.}
% \end{remark}

\begin{Corollary}[Joint Error Bounds of LASSO]\label{joint}
%Let $B_{N,T}\defeq \max\limits_{2\leq t\leq T}\{2C_{s_t^*}+ 2N^{-1}\sqrt{s_t^*}\lambda_t/\kappa_t(3,s_t^*)\}$.
Under the same assumptions as in Theorem \ref{lassobound}, if $\P\big(\bigcap_{t=1}^{T-1} (\mathcal{A}_{1t} \cap{\mathcal{A}_{2t}})\big) \to 1$ as $N,T\to\infty$, then with probability $1-\smallO(1)$,
\begin{align*}
\max_{1\leq t\leq T-1}|\delta_{\Pi,t}|_1 &\lesssim 7\max_{1\leq t\leq T-1} \sqrt{s_t^*}\{2C_{s_t^*}+ N^{-1}\sqrt{s_t^*}\lambda_t/\kappa_t(3,s_t^*)\}\Big/\Big(\min_{1\leq t\leq T-1} \kappa_t(3,s_t^*)\Big) \\
&\quad + N\max_{1\leq t\leq T-1} C_{s_t^*}^2\Big/\min_{1\leq t\leq T-1}\lambda_t.
\end{align*}
In addition, if $\P\big(\bigcap_{t=1}^{T-1}(\mathcal{A}'_{1t} \cap {\mathcal{A}_{2t}})\big) \to 1$ as $N,T\to\infty$, then with probability $1-\smallO(1)$,
\begin{eqnarray*}
\max_{1\leq t\leq T-1}|\delta_{\Pi,t}|_2 \lesssim   %\max_{2\leq t\leq T}\{2C_{s_t^*}+  2N^{-1}\sqrt{s_t^*}\lambda_t/\kappa_t(3,s_t^*)\}
\max_{1\leq t\leq T-1}\{2C_{s_t^*}+  N^{-1}\sqrt{s_t^*}\lambda_t/\kappa_t(3,s_t^*)\}\Big/\Big(\min_{1\leq t\leq T-1}\kappa_t(3,s_t^*)\Big).
\end{eqnarray*}
\end{Corollary}

According to the oracle order of $s_t^\ast$ and the corresponding oracle bound for $C_{s_t^*}$, as shown in Appendix \ref{oracle}, we obtain explicit rates for the general results in Theorem \ref{lassobound} and Corollary \ref{joint}. Specifically, when $s_t^\ast \asymp \log N \wedge t$, $C_{s_t^\ast} \lesssim_\P \sqrt{(\log N \wedge t)/N}$, $\lambda_t \lesssim_\P \sqrt{N\log m_t}$, and there exists a positive constant $\underline\kappa$ such that $\kappa_t(3,s_t^\ast)\ge \underline\kappa>0$, we have
$$|\delta_{\Pi,t}|_{2,N}
\lesssim_\P\sqrt{(\log N \wedge t)\log m_t/N},\quad |\delta_{\Pi,t}|_{1} \lesssim_\P(\log N \wedge t)\sqrt{\log m_t/N}.$$
% \begin{equation*}
% |\delta_{\Pi,t}|_{2,N}
% \lesssim_\P 2\sqrt{(\log N \wedge t)/N} + \sqrt{(\log N \wedge t)\log m_t/N}/\underline \kappa
% \lesssim
% \sqrt{(\log N \wedge t)\log m_t/N},
% \end{equation*}
% \begin{equation*}
% |\delta_{\Pi,t}|_{1} \lesssim_\P
% N^{-1/2}(\log N \wedge t)\big(14/\underline\kappa + 7\sqrt{\log m_t}/\underline \kappa^2 + 1/\sqrt{\log m_t}\big)\lesssim(\log N \wedge t)\sqrt{\log m_t/N}.
% \end{equation*}
Moreover, the joint error bounds satisfy the following rates:
$$\max_{1\leq t\leq T-1}|\delta_{\Pi,t}|_{1} \lesssim_\P(\log N \wedge T)\sqrt{\log T/N},\quad \max_{1\leq t\leq T-1}|\delta_{\Pi,t}|_{2}
\lesssim_\P\sqrt{(\log N \wedge T)\log T/N}.$$

\subsection{Basic Model: Inference Theory for Step 2}\label{secondstep}

In this subsection, we establish the asymptotic normality of the final estimator for both AB-LASSO and AB-LASSO-SS, which will allow us to perform large sample inference on the parameters of interest and functions of them.

Define $\bm\Theta_t^0$ (resp. $\widehat{\bm\Theta}_t$) by stacking $\Pi_t^0$ (resp. $\widehat\Pi_t$) by rows for each component $W_{it}$ of $\Delta\widetilde X_{it}$. Specifically, when the number of components in $X_{it}$ is $d$, we have ${\bm\Theta}_t^0$ and $\widehat{{\bm\Theta}}_t$ with dimensions $d\times m_t$ for each $t=1,\ldots T-1$. Recall the definition $V_{it} = (1,X_{i1}^\top,\ldots,X_{it}^\top)^\top$. It follows that $\widehat{\Delta\widetilde X_{it}}=\widehat{\bm\Theta}_tV_{it}$, and the AB-LASSO estimator obtained in \eqref{est} %(irrespective of the demeaning transformation to remove the time effects) 
can be expressed by %[IFV: SHOULD THE VARIABLES HAVE A TILDE IN THE CASE WHERE THERE ARE TIME EFFECTS?]
$$\widehat\theta - \theta^0 = \bigg(\sum_{i=1}^N \sum_{t=1}^{T-1} \widehat{\bm\Theta}_tV_{it}\Delta\widetilde X_{it}^{\top} \bigg)^{-1} \bigg(\sum_{i=1}^N \sum_{t=1}^{T-1} \widehat{\bm\Theta}_tV_{it} \Delta\widetilde\varepsilon_{it}\bigg).$$
The asymptotic variance of $\widehat\theta$ has the sandwich form. We impose the following assumption to derive the specific formula for it.

% %$V_{it}=[1\cdots\tilde{X}_{i(t-1)},0\cdots]^{\text{\ensuremath{\top}}}((T-2)*2\times1).$
% We let $L$ denote the number of instruments to be fitted endogenously at each step, and ${\Pi}_{t,l}^*$ be the corresponding fitted coefficient vector of dimension $\tilde{L}$ with respect to each endogenous regressor $\{\Delta Y_{it},\Delta X_{it}\}$.  And let $\tilde{L}=\max_t \ell_t$, which is the maximum number of IV over time $t$ for each endogenous regressor.
% Moreover we define the $L\times\tilde{L}$ matrix of the estimated LASSO coefficients in corresponding to
% the first step and with $0$ filled in after $2L$th element in each
% row. Namely we let
% \[
% \hat{\Theta}_{t}=\left[\begin{array}{c}
% \text{\ensuremath{\hat{\Pi}_{t,1}^{\top}}}\\
% \vdots\\
% \hat{\Pi}_{t,L}^{\top}
% \end{array}\right].
% \]
% And  $\Theta_{t}^{*\top}$  is with $\hat{\Theta}_{t}$  replaced by ${\Pi}_{t,l}^*$ therein.

\begin{Assumption}[Nonsingularity]\label{identification}
Assume that as $N,T\to\infty$, the limit matrix $Q=\lim\limits_{N,T\to\infty}(NT)^{-1}\sum_{i=1}^N\sum_{t=1}^{T-1}{\bm\Theta}_t^0\E(V_{it}\Delta\widetilde X_{it}^{\top})$ is nonsingular and $|Q^{-1}|_{1}\vee |Q^{-1}|_\infty<\infty$.
\end{Assumption}

%\begin{Assumption}[Conditional Variance]\label{a3}
%Assume that $\E(\varepsilon_{i,t-1}^{2}|V_{it})=\sigma_{\varepsilon,t}^{2} < \infty$, and $\E(\varepsilon_{i,t-1}\varepsilon_{is}|V_{it})=0$, for $1\leq s<t-1$.
%\end{Assumption}

% {\color{red} Discuss the conditional homoscedaticity.
% \begin{remark}
% Assumption \ref{a2} implies that $\E(\varepsilon_{i,t-1}\varepsilon_{is}|V_{it})=0$, for $1\leq s<t-1$. It assumes that the error term is uncorrelated conditioning on $V_{it}$ and it implies that the error term is unconditionally uncorrelated.  In addition, we also assume that the conditional variance is finite.
% \end{remark}}

%\red [IFV: should all the $\Delta \varepsilon_{it}$ be $\Delta \widetilde \varepsilon_{it}$?] 
Note that Assumption \ref{a2}(ii) implies 
%$\E((\Delta\varepsilon_{it})^{2}|V_{it})=\sigma_{\varepsilon,t}^{2}+\sigma_{\varepsilon,t-1}^{2}$, $\E(\Delta\varepsilon_{it}\Delta\varepsilon_{i,t-1}|V_{it})=-{\sigma}_{\varepsilon,t-1}^{2}$, and
$\E(\varepsilon_{it}\varepsilon_{is} \mid V_{it})=0$ for $1\leq s<t\leq T$ and $\E(\Delta\widetilde\varepsilon_{it}\Delta\widetilde\varepsilon_{i,t-\ell}\mid V_{it})\lesssim\frac{1}{T-t+\ell}\big(1-\frac{1}{N}\big)$ for any $\ell\geq 1$. It follows that  $\lim\limits_{T\to\infty}T^{-1}\sum_{t=\ell+1}^{T-1}\E(\Delta\widetilde\varepsilon_{it}\Delta\widetilde\varepsilon_{i,t-\ell} \mid V_{it})=0$ for any $\ell\geq 1$ and $$\lim\limits_{N,T\to\infty}(NT)^{-1}\sum\limits_{i=1}^N\sum\limits_{1\leq s<t\leq T-1}\E(\Delta\widetilde\varepsilon_{it}\Delta\widetilde\varepsilon_{is} \mid V_{it})=0.$$
Therefore, $\Delta\widetilde\varepsilon_{it}$ can be treated as approximately a martingale difference sequence over $t$. 
By defining $\Sigma_{0,t}\defeq\lim\limits_{N\to\infty}N^{-1}\sum_{i=1}^N\E(V_{it}V_{it}^\top(\Delta\widetilde\varepsilon_{it})^2)$, %and $\Sigma_{1,t}\defeq\lim\limits_{N\to\infty}N^{-1}\sum_{i=1}^N\E(V_{it}V_{i,t-1}^\top\Delta\varepsilon_{it}\Delta\varepsilon_{i,t-1})$,
we can express the asymptotic variance of $\widehat\theta$ in the form of
\begin{align}\label{var}
\Omega = Q^{-1}\bm\Sigma(Q^{-1})^\top, \quad \bm\Sigma = \lim_{T\to\infty}T^{-1}\sum_{t=1}^{T-1}{\bm\Theta}_t^0\Sigma_{0,t}{\bm\Theta}_t^{0\top}.
\end{align}
% where
% $$\bm\Sigma = \lim_{T\to\infty}\frac{1}{T}\sum_{t=1}^{T-1}{\bm\Theta}_t^0\Sigma_{0,t}{\bm\Theta}_t^{0\top}.$$ %+ \frac{1}{T} \sum_{t=3}^T{\bm\Theta}_t^0\Sigma_{1,t}{\bm\Theta}_{t-1}^{0\top} + \frac{1}{T}\sum_{t=3}^T{\bm\Theta}_{t-1}^0\Sigma_{1,t}^\top{\bm\Theta}_{t}^{0\top}\Big)
Accordingly, the empirical analog of $\Omega$ is
\begin{align}\label{var.est}
\widehat\Omega = \widehat Q^{-1}\widehat{\bm\Sigma}(\widehat Q^{-1})^\top,  \quad  \widehat{\bm\Sigma}= (T-1)^{-1} \sum_{t=1}^{T-1}\widehat {\bm\Theta}_t \widehat \Sigma_{0,t} \widehat {\bm\Theta}_t^{\top},
\end{align}
where $\widehat Q = \{N(T-1)\}^{-1}\sum_{i=1}^N\sum_{t=1}^{T-1}\widehat{\bm\Theta}_tV_{it}\Delta\widetilde X_{it}^{\top}$ and 
% and
% $$\widehat{\bm\Sigma}= \frac{1}{T} \sum_{t=1}^{T-1}\widehat {\bm\Theta}_t \widehat \Sigma_{0,t} \widehat {\bm\Theta}_t^{\top},$$ %+ \frac{1}{T} \sum_{t=3}^T \widehat {\bm\Theta}_t \widehat \Sigma_{1,t} \widehat {\bm\Theta}_{t-1}^{\top} + \frac{1}{T} \sum_{t=3}^T \widehat {\bm\Theta}_{t-1} \widehat \Sigma_{1,t}^\top \widehat {\bm\Theta}_{t}^{\top},$$
$\widehat \Sigma_{0,t} = N^{-1}\sum_{i=1}^N V_{it}V_{it}^\top(\widehat{\Delta\widetilde\varepsilon_{it}})^2$, with $\widehat{\Delta\widetilde{\varepsilon}_{it}}=\Delta\widetilde Y_{it}-\Delta\widetilde X_{it}^\top\widehat\theta$.
%The consistency of the feasible variance estimator of $\widehat\theta$ is established in the following lemma, which ensures that $\widehat\Omega$ can be used to form the confidence intervals for $\theta^0$. 

We now establish the formal asymptotic properties of the AB-LASSO and AB-LASSO-SS 
estimators. The key results are consistency and asymptotic normality, which together 
justify the use of these estimators for inference on $\theta^0$ in large panels. 
\begin{Theorem}[Asymptotic Normality of AB-LASSO and AB-LASSO-SS]\label{main}
Under Assumptions \ref{a1}--\ref{identification}, suppose that the asymptotic variance $\Omega$ is positive definite and that \\
$\max\limits_{1\leq t\leq T-1}\sqrt{s_t^*}\log m_t/\sqrt{N}\to0$ as $N,T\to\infty$. The AB-LASSO estimator $\widehat{\theta}$ obtained by \eqref{est} is consistent for $\theta^0$, and
$$ %\label{asym}
\sqrt{NT}(\widehat{\theta}- \theta^0 )\stackrel{\mathcal{L}}{\to} \operatorname{N}(0, \Omega).$$
Moreover, the AB-LASSO-SS estimator $\widehat\theta_{SS}$ obtained by \eqref{SS} is consistent for $\theta^0$, and
$$ %\label{asym2}
\sqrt{NT}(\widehat{\theta}_{SS}- \theta^0)\stackrel{\mathcal{L}}{\to} \operatorname{N}(0, \Omega).$$
\end{Theorem}

% {\color{red}
% \begin{remark}
% It shall be noted that the above asymptotic theory requires that $T^2/N\to 0$, the bias condition arises due to the generated errors of fitting LASSO in the first step.
% \end{remark}}

% \begin{Theorem}[Asymptotic Normality of AB-LASSO-SS] \label{mainss}
% Under Assumptions \ref{a1}--\ref{identification}, assuming the asymptotic variance $\Omega$ is a positive definite matrix, and $\max_{2\leq t\leq T}\frac{s_t^*\log m_t}{\sqrt{N}}\to0$, the AB-LASSO-SS $\widehat\theta_{SS}$ obtained by \eqref{SS} is a consistent estimator of $\theta^0$, and
% \beq %\label{asym2}
% \sqrt{NT}(\widehat{\theta}_{SS}- \theta^0)\stackrel{\mathcal{L}}{\to} \operatorname{N}(0, \Omega).
% \eeq
% \end{Theorem}

A notable feature of Theorem \ref{main} is that both estimators share the same 
limiting distribution, despite AB-LASSO-SS employing sample splitting. This reflects the fact that, in our setting, %sample splitting serves primarily to control some higher-order bias terms arising from the correlation between the first-stage generated errors and the idiosyncratic errors, rather than to alter the asymptotic variance itself. Consequently, this 
the bias reduction from sample splitting does not affect the central limit theorem, and both estimators attain the same $\sqrt{NT}$-rate and limiting Gaussian distribution.  This theoretical finding is consistent with our simulation evidence, where AB-LASSO and AB-LASSO-SS exhibit similar inferential performance in large samples. %The following remark elaborates on this point and explains why the FOD transformation plays a central role.

\begin{Remark}[Discussion of the Rate Condition] \label{remark:rates}
It is important to note that %$\Delta\widetilde\varepsilon_{it}$ exhibits weak serial dependence, since $\varepsilon_{it}$ is an uncorrelated sequence over $t$ and the FOD transformation involves harmonic weights of order $1/(T-t)$. 
% there is no serial correlation in $\Delta\widetilde\varepsilon_{it}$, when $\varepsilon_{it}$ is uncorrelated over $t$ and has constant variance. As a result, 
we obtain the same rate condition regardless of whether or not sample splitting is employed. This invariance would not hold for the first difference (FD) estimator. 
More fundamentally, to prove the main theorem above, we require
$$(NT)^{-1/2}\sum_{i=1}^{N}\sum_{t=1}^{T-1}Q^{-1}(\widehat{{\bm\Theta}}_{t}-{\bm\Theta}_{t}^{0})V_{it}\Delta\widetilde\varepsilon_{it}\to0,$$
as $N,T\to\infty$. Without sample splitting, the generated errors $(\widehat{{\bm\Theta}}_{t}-{\bm\Theta}_{t}^{0})$ could be correlated with the transformed errors in the main regression $\Delta\widetilde\vps_{it}$, and the order of %$(NT)^{-1/2}\sum_{i=1}^{N}\sum_{t=1}^{T-1}Q^{-1}(\widehat{{\bm\Theta}}_{t}-{\bm\Theta}_{t}^{0})V_{it}(\varepsilon_{it}-\varepsilon_{i,t-1})$ 
this term under FD is given by $\max\limits_{1\leq t\leq T-1} s_t^*\log m_t\sqrt{T/N}$. See Theorem \ref{main.fd} in Appendix \ref{app:FD} for the formal results on FD estimators. In contrast, either employing sample splitting or using the FOD transformation yields a smaller order for this component: $\max\limits_{1\leq t\leq T-1} \sqrt{s_t^*}\log m_t/\sqrt{N}$. \qed

% It is important to note that the small bias condition $\max\limits_{1\leq t\leq T-1} s_t^*\log m_t\sqrt{T/N}\to0$, as stated in Theorem \ref{main}, is relaxed in Theorem \ref{mainss}. This relaxation results in a more favorable convergence rate for the estimator when using AB-LASSO-SS. This observation certifies that sample-splitting effectively mitigates the overfitting bias by employing different sub-samples for instruments selection and parameter estimation.

% More fundamentally, to prove the two theorems above, we require the term $$(NT)^{-1/2}\sum_{i=1}^{N}\sum_{t=1}^{T-1}Q^{-1}(\widehat{{\bm\Theta}}_{t}-{\bm\Theta}_{t}^{0})V_{it}\Delta\widetilde\varepsilon_{it}\to0,$$
% as $N,T\to\infty$. Without sample-splitting, the generated errors $(\widehat{{\bm\Theta}}_{t}-{\bm\Theta}_{t}^{0})$ might be correlated with the ordinary errors $\vps_{it}$, and the order of this term is given by $\max\limits_{1\leq t\leq T-1} s_t^*\log m_t\sqrt{T/N}$. While using sample-splitting, we achieve a smaller order of this term: $\max\limits_{1\leq t\leq T-1} \sqrt{s_t^*}\log m_t/\sqrt{N}$. \qed
% %and Assumption \ref{a4} is sufficient.
\end{Remark}
%[IFV: can we simplify the rates using that the max is achieved at $t = T-1$?]

\begin{Remark}[Efficiency]
In Appendix \ref{app_ebound}, we examine the link between our results and efficiency in dynamic panel models with fixed effects, focusing on a specific univariate panel AR(1) specification. In Proposition \ref{ebound}, we formally show that the asymptotic variance of $\widehat\theta$, as given in \eqref{var}, simplifies to $\Omega=1-(\theta^0)^2$ and attains the efficiency bound derived in \citet{hahn2002asymptotically} under suitable assumptions, most notably i.i.d.\ Gaussian innovations. \qed
% Our estimator is efficient under additional assumptions. For example, consider the simple model
% \begin{equation}%\label{main:model2}
% Y_{it} \;=\; \alpha_i \;+\; \theta^0\,Y_{it-1} \;+\; \varepsilon_{it}.
% \end{equation}
% Under the assumptions of Theorem 3 in \cite{hahn2002asymptotically}, the asymptotic variance of our estimator becomes $\Omega = 1 - (\theta^0)^2$,
% which coincides exactly with the semiparametric efficiency bound derived therein. \qed
\end{Remark}
We conclude this subsection with a supporting lemma showing that the variance estimator $\widehat{\Omega}$ 
is consistent, which is essential for constructing feasible confidence intervals in 
practice.

%Finally, we shall  a lemma showing the consistency of the estimator of the variance covariance $\Omega$, which can be used to obtain analytical standard errors for $\hat \theta$ and $\hat \theta_{SS}$.

\begin{Lemma}[Consistency of the Variance Estimator]\label{lem:Omega-consistency}
Under Assumptions \ref{a1}--\ref{identification}, suppose that $|\Sigma|_1\vee |\Sigma|_\infty<\infty$ and $\max\limits_{1\leq t\leq T-1}\sqrt{s_t^*}\log m_t/\sqrt{N}\to0$ as $N,T\to\infty$. Then 
$$|\widehat{\Omega}-\Omega|_{\max}=\smallO_\P(1).$$
\end{Lemma}
This lemma requires that the number of relevant instruments $s_t^*$ grows slowly 
relative to $N$, a mild condition that is standard in high-dimensional IV settings and aligns with the rate condition in Theorem \ref{main}.

\subsection{General Model}
We now extend the asymptotic theory to the general model \eqref{main:model2}, which 
allows for many exogenous covariates. 
%This is the practically relevant case in applications where researchers wish to control for a large number of observed characteristics while still estimating dynamic effects consistently. 
We present the limiting distribution of the estimator $\widehat\theta_1$ obtained via \eqref{diverged1}.

The key difference relative to the basic model is that the instrument vector now 
takes the form $U_{it}=(U_{it}^{0\top}, \Delta \widetilde X_{2,it}^\top)^\top$, 
combining the ideal IVs for the endogenous component $\Delta \widetilde X_{1,it}$\footnote{The ideal IVs for $\Delta\widetilde X_{1,it}$ are structured similarly to those for $\Delta\widetilde X_{it}$ in the basic model, expressed as ${\bm\Theta}_t^0V_{it}$. The covariates $V_{it}$ are expanded by the additional moments arising from the strictly exogenous covariates in $X_{2,it}$, as commented in Remark \ref{SEC}.}  
with the exogenous component $\Delta \widetilde X_{2,it}$, which serve as their own 
instruments. The asymptotic variance of $\widehat\theta_1$ inherits this structure and takes the sandwich form
$$
\Omega_1 = Q_1^{-1}\bm\Sigma_1(Q_1^{-1})^\top, \quad \bm\Sigma_1 = \lim_{N,T\to\infty}(NT)^{-1}\sum_{i=1}^N\sum_{t=1}^{T-1}\mathcal W_t^\top\E(U_{it}U_{it}^\top(\Delta\widetilde\varepsilon_{it})^2)\mathcal W_t,
$$
where 
$Q_1=\lim\limits_{N,T\to\infty}(NT)^{-1}\sum_{i=1}^N\sum_{t=1}^{T-1}
\mathcal W_t^\top\E(U_{it}\Delta\widetilde X_{1,it}^{\top})$ captures the relevance of the instruments for the endogenous component and is assumed to be nonsingular, and $\bm\Sigma_1$ 
% \begin{equation*}
% \bm\Sigma_1 = \lim_{N,T\to\infty}(NT)^{-1}\sum_{i=1}^N\sum_{t=1}^{T-1}
% \mathcal W_t^\top\E\!\left(U_{it}U_{it}^\top(\Delta\widetilde\varepsilon_{it})^2
% \right)\mathcal W_t
% \end{equation*}
is the outer-product variance of the moment conditions, weighted by $\mathcal W_t$. 
The weighting matrix $\mathcal W_t$ %plays the role of the optimal IV weight and 
is estimated via a Dantzig selector in practice.

To make inference feasible, we define the moment matrix
$M_t\defeq\bar\E\left\{\begin{pmatrix}
\Delta\widetilde X_{1,it}\\
\Delta\widetilde X_{2,it}
\end{pmatrix}
U_{it}^{\top}\right\}$ and its empirical counterpart $\widehat M_t\defeq N^{-1}\sum_{i=1}^N\left\{\begin{pmatrix}
\Delta\widetilde X_{1,it}\\
\Delta\widetilde X_{2,it}
\end{pmatrix}
\widehat U_{it}^{\top}\right\}$, where the ideal instruments $U_{it}^0$ in $U_{it}$ are replaced by their LASSO 
predictions, i.e.\ $\widehat U_{it}=(\widehat{\Delta\widetilde X_{1,it}^\top}, 
\Delta\widetilde X_{2,it}^\top)^\top$. 

The following assumption collects the regularity conditions needed on the weighting matrix $\mathcal W_t$ and the tuning parameters $\ell_t$ under the general model. 
%Some additional assumptions on %the moments $M_t$ and 
%the weighting matrix $\mathcal W_t$ and the tuning parameters $\ell_t$ under the general model are made as follows.
\begin{Assumption}[Additional Assumptions on the General Model]\phantomsection\label{assum}
\begin{enumerate}
    \item[(i)] There exist sequences of constants $c_{n},w_n\geq0$ (where $n=NT$), such that $\max\limits_{1\leq t\leq T-1}(|\mathcal{W}_{t}|_{1,1}\vee|M_{t}^{-1}|_{\infty})\leq c_{n}$, and $$\max\limits_{1\leq t\leq T-1}\sum_{i=1}^d\sum_{j=1}^{d_1}|\mathcal{W}_{t,ij}|^{r}\lesssim w_n,\, \text{ for some } 0\leq r<1,$$
    where $\mathcal W_{t,ij}$ is the element in the $i$-th row and $j$-th column of $\mathcal W_t$.
  %  \item[(ii)] For each $t=2,\ldots,T$, assume that $|M_{t}-\widehat{M}_{t}|_{\max}\lesssim_\P\rho_{N,t}$, with $\rho_{N,t}\to0$ as $N\to\infty$.
    \item[(ii)] Let $v_{n} \defeq\log (N\vee T\vee d)$, and assume that
    $$c_n^2w_n(v_n/N)^{\frac{1-r}{2}} + c_n\max\limits_{1\leq t\leq T-1}\sqrt{s_t^*}\log m_t/\sqrt{N}=\smallO(1).$$
    %$$c_n^2w_n(v_n/N)^{\frac{1-r}{2}}\sqrt{T} + c_n\max_{1\leq t\leq T-1}s_t^*\log m_t\sqrt{T}/\sqrt{N}\to 0,\, \text{ as }N,T\to \infty,$$
    %{\color{red} $\{ [c_{N,T}^2s^*(\sqrt{v_{n}}/\sqrt{N})^{1-r}]\vee \sqrt{N}^{-1}\sqrt{s^*\max_t\log{m_t}}\ll \sqrt{T}^{-1}$\}}
    with the same $r$ that makes part (i) hold.
    \item[(iii)] The tuning parameter $\ell_t\geq0$ used in \eqref{penalize} satisfies $\max\limits_{1\leq t\leq T-1}\ell_{t}\vartheta_n=\smallO(1/\sqrt{NT})$, and $\max\limits_{1\leq t\leq T-1}(\ell_{t}+\rho_{N,t}c_{n})\lesssim c_{n}\sqrt{v_{n}/N}$, where $\vartheta_n\defeq |\theta_2^0|_1$, and $\rho_{N,t}$ is a sequence such that $|M_{t}-\widehat{M}_{t}|_{\max}\lesssim_\P\rho_{N,t}$.
\end{enumerate}
\end{Assumption}
Assumption \ref{assum}(i) bounds the magnitude of the weighting matrix and the 
moment matrix, while Assumption \ref{assum}(ii) translates these bounds into a rate condition on the 
sparsity of the first-stage LASSO. This condition is analogous in spirit to the rate condition in 
Theorem \ref{main}, but adjusted for the additional complexity introduced by the high-dimensional exogenous covariates and the estimated weights. Assumption \ref{assum}(iii) %ensures that the penalization is not too strong, so that the shrinkage bias in $\widehat\theta_1$ is asymptotically negligible.
restricts the order of the tuning parameter $\ell_t$ used in the Dantzig selector, which controls the penalty level and the deviation $|\widehat{\mathcal W}_t - \mathcal W_t|_{\max}$.

\begin{Theorem}[Asymptotic Normality for the General Model Estimator] \label{maindiverged}
%Assume that $\|\sigma_1^2\|<\infty$ is a positive definite matrix.
Under Assumptions \ref{a1}--\ref{assum}, assuming the asymptotic variance $\Omega_1$ is a positive definite matrix, we have the general model estimator $\widehat\theta_{1}$ obtained by \eqref{diverged1} is consistent for $\theta_1^0$, and
\beq \label{asym22}
\sqrt{NT}(\widehat{\theta}_{1}- \theta_1^0)\stackrel{\mathcal{L}}{\to} \operatorname{N}(0, \Omega_1).
\eeq
\end{Theorem}

The theorem confirms that the $\sqrt{NT}$-rate and Gaussian limiting distribution 
established for the basic model carry over intact to the general setting with many 
exogenous covariates, provided the additional regularity conditions in 
Assumption \ref{assum} are satisfied. The estimator thus remains suitable for 
inference in the empirically relevant case. %where researchers include a large number of controls alongside the endogenous dynamic regressors. {\red IFV: check if the rest of the paragraph belongs here} 
Notably, the crucial rate condition in Assumption \ref{assum}(ii) remains unchanged 
% can be further improved to
% $c_n^2w_n(v_n/N)^{\frac{1-r}{2}} + c_n\max\limits_{1\leq t\leq T-1}\sqrt{s_t^*}\log m_t/\sqrt{N}=\smallO(1)$,
if a sample-splitting procedure is employed. Specifically, $\widehat{\mathcal W}_t$ obtained via the Dantzig selector and the instrument selection by LASSO are constructed from a sub-sample that is cross-sectionally independent of the sub-sample used to compute the final estimator in \eqref{diverged1}. Further details are provided at the end of the proof of Theorem \ref{maindiverged}.

\input{newsimulation}

\section{School Opening and COVID-19 Spread}\label{app}
We apply AB-LASSO to study the effect of K-12 schools opening and other policies on the spread of COVID-19 in the U.S. We use a balanced panel of 2,510 counties over 32 weeks between April 1st and December 2nd, 2020. This panel was extracted from \citet{chernozhukov2021association}, which constructed an unbalanced  panel of  U.S. counties including 7-day moving averages of daily observations for the same period. We aggregate the observations at the week level to avoid spurious serial correlation coming from the moving averages.

In this application, $Y_{it}$ is the logarithm of the number of reported COVID-19 cases in county $i$ at week $t$, $D_{it}$ is a measure of visits to K-12 schools from SafeGraph foot traffic data, $C_{it}$ contains other treatments and control variables,  $\alpha_i$ is a county fixed effect and $\gamma_t$ is a week fixed effect.  We estimate the model:
\begin{equation*}\label{covid}
 Y_{it} = \alpha_i + \gamma_t + \theta_0 D_{i,t-1} + \beta_1Y_{i,t-1} + \beta_2Y_{i,t-2} + \beta_3Y_{i,t-3} + \beta_4Y_{i,t-4} + \theta_1^\top C_{1i,t-1} + \theta_2 C_{2it}  + \varepsilon_{it},
\end{equation*}
where  $C_{1it}$ includes a measure of visits to colleges and policy indicators on mask mandates, stay-at-home orders and the ban on gatherings of more than 50 persons,  and $C_{2it}$ includes a measure of the weekly growth rate in the number of tests. We assume that there is no serial correlation in $\varepsilon_{it}$ over $t$. The variables in $D_{it}$ and $C_{1it}$ enter the model lagged one week to account for the time lag between infection and case confirmation. Additionally, we assume that $D_{it}$, $C_{1it}$ and $C_{2it}$  are predetermined with respect to $\varepsilon_{it}$. Accordingly, we use %$Y_{i,t-2},\ldots, Y_{i1}$, $D_{i,t-1},\ldots, D_{i1}$, $C_{1i,t-1},\ldots, C_{1i1}$ and $C_{2i,t-1},\ldots, C_{2i1}$
$Y_{i,t-1},\ldots, Y_{i1}$, $D_{it},\ldots, D_{i1}$, $C_{1it},\ldots, C_{1i1}$, and $C_{2it},\ldots, C_{2i1}$ to construct moment conditions at each $t=1,\ldots,T-1$, for AB-LASSO with FOD.
%For AB with first differencing, the valid instruments at each $t$ are given by $Y_{i,t-2},\ldots, Y_{i1}$, $D_{i,t-1},\ldots, D_{i1}$, $C_{1i,t-1},\ldots, C_{1i1}$ and $C_{2i,t-1},\ldots, C_{2i1}$.
This yields $m=3,375$ moment conditions and $n=NT = 2,510 \times 27 = 67,770$ observations.\footnote{The first 5 weeks are used as initial conditions.} %[IFV: We need to check the number of moment conditions for , I FIND $m=2,901$]
AB is likely to be biased in this case because
$m^2/n \approx  168$.

Table \ref{app_table1} presents estimates and standard errors for AB-LASSO, AB-LASSO-SS with $K=2$, AB, DAB-SS, and DFE-A.\footnote{For conciseness, we do not report the results for the autoregressive coefficients, as they are not the policy parameters of interest, but they are available from authors upon request.} The likelihood-based estimator is excluded because applying the \texttt{dpm} function in \textsf{R} that we use in the numerical simulations  results in a singular design error.
%it fails to handle binary covariates with very low variation.
%[IFV: CAN WE JUST SAY THAT ML DOES NOT CONVERGE OR DOES NOT RUN]
For the AB-LASSO and AB-LASSO-SS estimators, the penalty parameters are chosen based on the calibrated simulation in Appendix \ref{app.cab}. The coefficients of the model measure the short run effects of the covariates. In addition to these coefficients, we report results for the long-run effects obtained as $\theta_k/(1-\sum_{j=1}^4\beta_j)$ where $\theta_k$ is the coefficient of the covariate of interest and $\beta_1,\ldots,\beta_4$ are the coefficients of $Y_{i,t-1},\ldots,Y_{i,t-4},$ respectively.

All the methods reveal  positive and significant effects of K-12 school openings and college visits on the spread of COVID-19. In particular, allowing
 K-12 school openings and college visits is associated with a higher number of cases in both the short and long run.  The estimated effects of mask mandates and stay-at-home orders are negative and significant both in the short and long run. The positive coefficient of weekly test growth is found to be marginally significant. %The effect of banning gatherings is less statistically significant under AB-LASSO and AB-LASSO-SS. 

Compared with AB, AB-LASSO and AB-LASSO-SS produce considerably smaller absolute estimates for college visits and banning gatherings, but more negative effects of mask mandates, in both the short and long run. Our methods also provide more precise estimates of both effects than AB, particularly for school opening and college visits. %The difference might be attributed to the bias in the autoregressive coefficients in AB. In particular, AB-LASSO-SS gives significantly smaller estimates of the coefficient of $Y_{i,t-1}$.
Compared with DFE-A, AB-LASSO and AB-LASSO-SS predict smaller short-run effects but larger long-run effects for school opening and college visits, substantially larger absolute effects for both short- and long-run impacts of mask mandates and stay-at-home orders, and less negative effects of banning gatherings.\footnote{In the calibrated simulation of Online Appendix \ref{app.cab}, we find that AB-LASSO and DFE-A have biases with reverse signs for the short run effects of K12 school openings, college visits, mask mandates and stay-at-home orders.}
In general, the results of AB-LASSO-SS are not sensitive to the number of folds $K$, so we report only $K=2$. Debiasing the standard AB through half-splitting the panel does not substantially change most short-run effect estimates. %The difference in the long-run effects might be attributed to the bias in the autoregressive coefficients in AB.
%DFE-A produces a larger absolute estimate for the coefficient of the lagged dependent variable than conventional FE.

According to AB-LASSO, we conclude that the opening of K-12 schools one week is associated with an increase in the number of COVID-19 cases of about $50\%$ the week after and has a compounded long run increase of nearly $300\%$. Allowing college visits one week is associated with approximately an $80\%$ increase in cases in the following week and about a $500\%$ increase in the long run. Mask mandates, stay-at-home orders, and banning gatherings have more modest effects. The reduction in the number of cases is $10\%$, $7\%$, and $7\%$ after one week and roughly $65\%$, $40\%$, and $45\%$ in the long run, respectively. These effects are all statistically and economically relevant.

% To remove the unobserved individual effects, we transform the model by taking the first difference over time. The time effect is controlled by including time dummies. We use all lags $Y_{i,t-2},\ldots$, $D_{i,t-1},\ldots$, $\bm Z_{i,t-1},\ldots$, and $\mathcal T_{i,t-1},\ldots$, to construct moment conditions for both AB and AB-LASSO. Note that in this case, projecting $\Delta\mathcal T_{it}$  and $\Delta Y_{i,t-1}$ onto the IV space is required, but not necessary for $\Delta D_{i,t-1}$ and $\Delta\bm Z_{i,t-1}$. We implement a $K$-fold sample-splitting and cross-fitting procedure, as described in Section \ref{sim}. The debiased AB estimator (DAB-SS) has also been detailed in Section \ref{sim}.
% %For sample-splitting, the whole sample is randomly divided by two parts alone the cross-sectional dimension as the main/auxiliary sample respectively. The estimates and analytical standard deviations (in parentheses) are averaged over 100 random splittings by medians. %The empirical standard deviations of the estimates over random splittings are reported in brackets.

% Additionally, apart from the short-run effect $\alpha$, a long-run effect is also measured by $\alpha/(1-\sum_{j=1}^4\beta_j)$.

\section{Concluding Remarks}\label{sec:conclusion}
We propose  a LASSO and cross-fitting based estimator of dynamic linear panel models. This estimator shows better large sample properties and finite-sample performance in simulations than the classical AB estimator in long panels. In an empirical application, our estimator finds that policies such as the closure of K-12 schools and colleges, mask mandates, and stay-at-home orders reduced the spread of COVID-19 in both the short and long run. Our estimated effects of college visits and banning gatherings, however, are less optimistic than the estimates obtained with AB.

A potential avenue for future research is to analyze the performance of our method under weak identification. This situation arises, for example, in dynamic models when the process of the outcome $Y_{it}$ is very persistent such that lagged values of $Y_{it}$ might not be strongly correlated with the transformed outcome $\Delta Y_{it}$. Here, we can follow  \cite{blundell1998initial} and employ a system GMM estimator that exploits additional moment conditions for the equation in levels, if we make stationarity assumptions about the initial conditions of the process. The number of these moment conditions grows fast with $T$ and may introduce additional many moment conditions bias. In preliminary numerical simulations, we find that LASSO does not do a good job selecting among these moment conditions as they do not have a natural sparse structure. 
We leave a more detailed treatment of the system GMM estimator with many moment conditions to future research.
Alternatively, \cite{newey2009generalized} developed a jackknife GMM estimator to deal with many weak instruments for cross-sectional data, which could be extended to incorporate cross-fitting in our panel setting.
Here, we also conjecture that standard methods for dealing with weak instruments for cross-section data, such as the use of the Anderson-Rubin statistics, can be fruitfully applied to our setting \citep{anderson1949estimation}.  We  also leave this analysis to future research.

\begin{table}[htbp!]
	\begin{center} 	\caption{Results for $\theta_2=0.25$ with $N=100$, heteroskedastic case}\label{table:n100t2.hetero}
		\begin{tabular}{p{1.5cm} cccccccc}
			\hline\hline
			& \footnotesize{AB} &  \footnotesize{AB-LASSO} & \footnotesize{AB-LASSO-SS} &  \footnotesize{AB-LASSO-SS} & \footnotesize{DAB-SS} & \footnotesize{DFE-A} & \footnotesize{ML}\\
			& &  & \footnotesize{($K=2$)} &  \footnotesize{($K=5$)} &  \\
			\cline{2-8}
			& \multicolumn{7}{c}{\small{$T=20$}}\\
			\hline
            \footnotesize{RMSE} & 0.32 & 0.14 & 0.17 & 0.17 & 0.57 & 0.15 & \bf{0.13} \\
            \footnotesize{SD} & 0.32 & 0.12 & 0.16 & 0.16 & 0.57 & \bf{0.11} & 0.13 \\
            \footnotesize{bias} & -0.05 & -0.06 & -0.06 & -0.07 & -0.01 & 0.10 & \bf{0.00} \\
            \footnotesize{CI length} & 2.22 & 0.49 & 0.80 & 0.64 & 2.22 & 0.43 & 0.49 \\
            \footnotesize{coverage} & 1.00 & 0.91 & 0.99 & \bf{0.96} & \bf{0.96} & 0.85 & 0.93 \\
			\hline
			& \multicolumn{7}{c}{\small{$T=30$}}\\
			\hline
			\footnotesize{RMSE} & 0.55 & \bf{0.10} & 0.11 & 0.11 & 0.91 & 0.11 & \bf{0.10} \\
            \footnotesize{SD} & 0.53 & \bf{0.09} & 0.11 & 0.10 & 0.90 & \bf{0.09} & 0.10 \\
            \footnotesize{bias} & -0.12 & -0.03 & -0.03 & -0.03 & -0.14 & 0.07 & \bf{0.00} \\
            \footnotesize{CI length} & 3.31 & 0.37 & 0.49 & 0.42 & 3.31 & 0.34 & 0.37 \\
            \footnotesize{coverage} & 1.00 & 0.93 & 0.98 & \bf{0.96} & 0.93 & 0.87 & \bf{0.94} \\
			\hline
			& \multicolumn{7}{c}{\small{$T=40$}}\\
			\hline
			\footnotesize{RMSE} & 0.60 & \bf{0.08} & 0.09 & \bf{0.08} & 1.01 & 0.09 & 0.10 \\
            \footnotesize{SD} & 0.57 & 0.08 & 0.09 & 0.08 & 0.99 & \bf{0.07} & 0.10 \\
            \footnotesize{bias} & -0.18 & -0.01 & -0.01 & -0.01 & -0.20 & 0.06 & \bf{0.00} \\
            \footnotesize{CI length} & 3.55 & 0.31 & 0.38 & 0.34 & 3.55 & 0.29 & 0.31 \\
            \footnotesize{coverage} & 0.99 & \bf{0.95} & 0.97 & 0.96 & 0.91 & 0.88 & 0.89 \\
            \hline
			& \multicolumn{7}{c}{\small{$T=50$}}\\
			\hline
            \footnotesize{RMSE} & 0.68 & \bf{0.07} & \bf{0.07} & \bf{0.07} & 1.13 & 0.08 &  \\
            \footnotesize{SD} & 0.67 & 0.07 & 0.07 & 0.07 & 1.12 & \bf{0.06} &  \\
            \footnotesize{bias} & -0.12 & \bf{-0.01} & \bf{-0.01} & \bf{-0.01} & -0.10 & 0.04 &  \\
            \footnotesize{CI length} & 3.70 & 0.27 & 0.32 & 0.29 & 3.70 & 0.26 &  \\
            \footnotesize{coverage} & 0.99 & \bf{0.94} & 0.97 & \bf{0.96} & 0.90 & 0.90 &  \\
            \hline
			& \multicolumn{7}{c}{\small{$T=60$}}\\
			\hline
            \footnotesize{RMSE} & 0.66 & \bf{0.07} & \bf{0.07} & \bf{0.07} & 1.09 & \bf{0.07} &  \\
            \footnotesize{SD} & 0.65 & \bf{0.06} & 0.07 & 0.07 & 1.09 & \bf{0.06} &  \\
            \footnotesize{bias} & -0.14 & -0.02 & \bf{-0.01} & \bf{-0.01} & -0.11 & 0.03 &  \\
            \footnotesize{CI length} & 3.74 & 0.24 & 0.28 & 0.26 & 3.74 & 0.23 &  \\
            \footnotesize{coverage} & 1.00 & \bf{0.94} & \bf{0.96} & \bf{0.94} & 0.90 & 0.90 &  \\
            \hline\hline
    \multicolumn{8}{l}{\footnotesize{Notes: The numbers in the table are divided by 0.25 for RMSE, SD, bias, and CI length.}}\\
    \multicolumn{8}{l}{\footnotesize{Superior results are indicated in bold.}}
		\end{tabular}
	\end{center}
\end{table}

\begin{table}[htbp!]
	\begin{center} 	\caption{Results for $\theta_2=0.25$ with $N=200$, heteroskedastic case}\label{table:n200t2.hetero}
		\begin{tabular}{p{1.5cm} cccccccc}
			\hline\hline
			& \footnotesize{AB} &  \footnotesize{AB-LASSO} & \footnotesize{AB-LASSO-SS} &  \footnotesize{AB-LASSO-SS} & \footnotesize{DAB-SS} & \footnotesize{DFE-A} & \footnotesize{ML}\\
			& &  & \footnotesize{($K=2$)} &  \footnotesize{($K=5$)} &  \\
			\cline{2-8}
			& \multicolumn{7}{c}{\small{$T=20$}}\\
			\hline
            \footnotesize{RMSE} & 0.11 & 0.11 & 0.11 & 0.11 & 0.22 & 0.13 & \bf{0.09} \\
            \footnotesize{SD} & 0.11 & 0.09 & 0.09 & 0.09 & 0.21 & \bf{0.08} & 0.09 \\
            \footnotesize{bias} & -0.02 & -0.06 & -0.05 & -0.06 & 0.03 & 0.10 & \bf{0.00} \\
            \footnotesize{CI length} & 0.70 & 0.34 & 0.40 & 0.37 & 0.70 & 0.31 & 0.35 \\
            \footnotesize{coverage} & 1.00 & 0.87 & \bf{0.94} & 0.91 & 0.89 & 0.75 & \bf{0.94} \\
			\hline
			& \multicolumn{7}{c}{\small{$T=30$}}\\
			\hline
			\footnotesize{RMSE} & 0.26 & \bf{0.07} & \bf{0.07} & \bf{0.07} & 0.45 & 0.10 & \bf{0.07} \\
            \footnotesize{SD} & 0.25 & 0.07 & 0.07 & 0.07 & 0.46 & \bf{0.06} & 0.07 \\
            \footnotesize{bias} & -0.05 & -0.03 & -0.02 & -0.02 & \bf{0.00} & 0.08 & \bf{0.00} \\
            \footnotesize{CI length} & 1.79 & 0.26 & 0.28 & 0.27 & 1.79 & 0.24 & 0.26 \\
            \footnotesize{coverage} & 1.00 & 0.94 & 0.96 & \bf{0.95} & \bf{0.95} & 0.73 & 0.92 \\
			\hline
			& \multicolumn{7}{c}{\small{$T=40$}}\\
			\hline
            \footnotesize{RMSE} & 0.42 & \bf{0.06} & \bf{0.06} & \bf{0.06} & 0.69 & 0.08 & \bf{0.06} \\
            \footnotesize{SD} & 0.39 & 0.06 & 0.06 & 0.06 & 0.67 & \bf{0.05} & 0.06 \\
            \footnotesize{bias} & -0.16 & -0.02 & -0.02 & -0.02 & -0.17 & 0.05 & \bf{0.00} \\
            \footnotesize{CI length} & 2.44 & 0.22 & 0.23 & 0.22 & 2.44 & 0.20 & 0.22 \\
            \footnotesize{coverage} & 1.00 & 0.93 & 0.94 & 0.93 & 0.92 & 0.82 & \bf{0.95} \\
            \hline
			& \multicolumn{7}{c}{\small{$T=50$}}\\
			\hline
            \footnotesize{RMSE} & 0.51 & \bf{0.05} & \bf{0.05} & \bf{0.05} & 0.84 & 0.06 & \bf{0.05} \\
            \footnotesize{SD} & 0.49 & 0.05 & 0.05 & 0.05 & 0.82 & \bf{0.04} & 0.05 \\
            \footnotesize{bias} & -0.15 & -0.02 & -0.01 & -0.01 & -0.15 & 0.05 & \bf{0.00} \\
            \footnotesize{CI length} & 2.73 & 0.19 & 0.20 & 0.20 & 2.73 & 0.18 & 0.19 \\
            \footnotesize{coverage} & 0.99 & \bf{0.95} & 0.96 & 0.96 & 0.88 & 0.85 & \bf{0.95} \\
            \hline
			& \multicolumn{7}{c}{\small{$T=60$}}\\
			\hline
            \footnotesize{RMSE} & 0.55 & \bf{0.04} & \bf{0.04} & \bf{0.04} & 0.89 & 0.06 & 0.05 \\
            \footnotesize{SD} & 0.53 & \bf{0.04} & \bf{0.04} & \bf{0.04} & 0.87 & \bf{0.04} & 0.05 \\
            \footnotesize{bias} & -0.17 & -0.01 & -0.01 & -0.01 & -0.19 & 0.04 & \bf{0.00} \\
            \footnotesize{CI length} & 2.83 & 0.17 & 0.18 & 0.17 & 2.83 & 0.16 & 0.17 \\
            \footnotesize{coverage} & 0.99 & \bf{0.95} & 0.96 & \bf{0.95} & 0.89 & 0.85 & 0.93 \\
			\hline\hline
    \multicolumn{8}{l}{\footnotesize{Notes: The numbers in the table are divided by 0.25 for RMSE, SD, bias, and CI length.}}\\
    \multicolumn{8}{l}{\footnotesize{Superior results are indicated in bold.}}
		\end{tabular}
	\end{center}
\end{table}

\begin{table}[htbp!]
	\begin{center}\caption{Short-run and long-run effects on COVID-19 cases}\label{app_table1}
		\begin{tabular}{r ccccc}
 			\hline\hline
			&  \footnotesize{AB-LASSO} &  \footnotesize{AB-LASSO-SS} & \footnotesize{AB}  & \footnotesize{DAB-SS}  & \footnotesize{DFE-A} \\
			& & \footnotesize{($K=2$)} & & &  \\
			\hline
			% \multicolumn{1}{r}{\footnotesize{$Y_{i,t-1}$}} & 0.75$^{***}$ & 0.76$^{***}$ & 0.78$^{***}$  & 0.79$^{***}$ & 0.72$^{***}$ & 0.76$^{***}$ \\
			% & (0.01) & (0.01)  & (0.01) & (0.01) & (0.01) & (0.01) \\
			% \multicolumn{1}{r}{\footnotesize{$Y_{i,t-2}$}} & 0.01 & 0.00 & 0.01 & 0.00 & 0.00 & -0.03$^{***}$ \\
			% & (0.01) & (0.01)  & (0.01) & (0.01) & (0.01) & (0.01) \\
			% \multicolumn{1}{r}{\footnotesize{$Y_{i,t-3}$}} & 0.06$^{***}$ & 0.06$^{***}$ & 0.06$^{***}$ & 0.06$^{***}$ & 0.05$^{***}$ & 0.05$^{***}$ \\
			% & (0.01) & (0.01)  & (0.01) & (0.01) & (0.01) & (0.01) \\
			% \multicolumn{1}{r}{\footnotesize{$Y_{i,t-4}$}} & 0.01$^{***}$ & 0.02$^{***}$ & 0.01 & 0.01 & -0.02$^{***}$ & -0.03$^{***}$ \\
			% & (0.00) & (0.00)  & (0.01) & (0.01) & (0.00) & (0.00) \\
			\multicolumn{1}{r}{\footnotesize{K-12 school opening: Short-run}} & 0.55$^{***}$ & 0.50$^{***}$  & 0.47$^{***}$  & 0.47$^{***}$ &  0.65$^{***}$ \\
			& (0.09) & (0.10)  & (0.10) & (0.10)  & (0.09) \\
			\multicolumn{1}{r}{\footnotesize{Long-run}} & 3.29$^{***}$ & 3.12$^{***}$ & 3.17$^{***}$ & 3.38$^{***}$  & 2.61$^{***}$ \\
			& (0.53) & (0.58)  & (0.69) & (0.73) &  (0.34) \\
			\multicolumn{1}{r}{\footnotesize{College visits: Short-run}} & 0.70$^{***}$ & 0.79$^{***}$& 1.14$^{***}$ & 1.26$^{***}$ &  0.87$^{***}$ \\
			& (0.18) & (0.18) & (0.34) & (0.34)  & (0.23) \\
			\multicolumn{1}{r}{\footnotesize{Long-run}} & 4.20$^{***}$ & 4.91$^{***}$ & 7.69$^{***}$ & 9.06$^{***}$  & 3.51$^{***}$ \\
			& (1.06) & (1.09) & (2.32) & (2.48) &  (0.90) \\
			\multicolumn{1}{r}{\footnotesize{Mask mandates: Short-run}} & -0.14$^{***}$ & -0.14$^{***}$ & -0.10$^{***}$ & -0.10$^{***}$ &  -0.09$^{***}$ \\
			& (0.01) & (0.01) & (0.01) & (0.01) &  (0.01) \\
			\multicolumn{1}{r}{\footnotesize{Long-run}} & -0.81$^{***}$ & -0.89$^{***}$ & -0.66$^{***}$ & -0.70$^{***}$& -0.37$^{***}$ \\
			& (0.08) & (0.09) & (0.10) & (0.11) &  (0.05) \\
			%& [0.02] & [0.01] &\\
			\multicolumn{1}{r}{\footnotesize{Stay-at-home orders: Short-run}} & -0.09$^{***}$ & -0.09$^{***}$ & -0.08$^{***}$ & -0.09$^{***}$ &  -0.06$^{***}$ \\
			& (0.02) & (0.02) & (0.02) & (0.02) &  (0.02) \\
			\multicolumn{1}{r}{\footnotesize{Long-run}} & -0.53$^{***}$ & -0.53$^{***}$ & -0.54$^{***}$ & -0.64$^{***}$ &  -0.23$^{***}$ \\
			& (0.12) & (0.12) & (0.16) & (0.17) & (0.07) \\
			\multicolumn{1}{r}{\footnotesize{Banning gatherings: Short-run}} & -0.04$^{*}$ & -0.04$^{*}$ & -0.08$^{***}$ & -0.08$^{***}$ &  -0.07$^{***}$ \\
			& (0.02) & (0.02) & (0.02) & (0.02)  & (0.02) \\
			\multicolumn{1}{r}{\footnotesize{Long-run}} & -0.23$^{*}$ & -0.22$^{*}$ & -0.52$^{***}$ & -0.58$^{***}$ &  -0.27$^{***}$ \\
			& (0.13) & (0.13) & (0.16) & (0.17) &  (0.07) \\
			\multicolumn{1}{r}{\footnotesize{Tests weekly growth}} & 0.01$^{*}$ & 0.01$^{*}$ & 0.01$^{*}$ & 0.01$^{*}$ &  0.01$^{*}$ \\
			& (0.00) & (0.00) & (0.01) & (0.01) &  (0.00) \\
			\hline\hline
            \multicolumn{6}{l}{\footnotesize{Notes: Analytical standard errors in parentheses. Significant codes: $0.01^{***}$, $0.05^{**}$, $0.1^{*}$.}}
		\end{tabular}
	\end{center}
% \caption{Parameters estimates and analytical standard errors (in parentheses). % for the model in \eqref{covid}.
% Significant codes: $0.01^{***}$, $0.05^{**}$, $0.1^{*}$.}\label{app_table1}
\end{table}

\clearpage
\bibliography{biball}
\bibliographystyle{apalike}

\newpage
    \vskip 2em \centerline{\Large \bf Appendix} \vskip -1em
    \setcounter{subsection}{0}
    \vskip 2em

    %\section{Appendix}
\begin{appendices}
    \renewcommand{\thesubsection}{A.\arabic{subsection}}
    \setcounter{equation}{0}
    \renewcommand{\theequation}{A.\arabic{equation}}
    \setcounter{theorem}{0}
    \renewcommand{\thetheorem}{A.\arabic{theorem}}
    \setcounter{lemma}{0}
    \renewcommand{\thelemma}{A.\arabic{lemma}}
    \setcounter{figure}{0}
    \renewcommand{\thefigure}{A.\arabic{figure}}
    \setcounter{table}{0}
    \renewcommand{\thetable}{A.\arabic{table}}
    \setcounter{Remark}{0}
    \renewcommand{\theRemark}{A.\arabic{Remark}}
    \setcounter{corollary}{0}
    \renewcommand{\thecorollary}{A.\arabic{corollary}}
    \setcounter{Example}{0}
    \renewcommand{\theExample}{A.\arabic{Example}}
    \setcounter{Assumption}{0}
    \renewcommand{\theAssumption}{A.\arabic{Assumption}}
    \setcounter{Proposition}{0}
    \renewcommand{\theProposition}{A.\arabic{Proposition}}

\section{Technical Proofs}

\subsection{Some Useful Lemmas and Auxiliary Results}
Define $\bm x=(x_1,\ldots,x_n)^\top, \bm y=(y_1,\ldots,y_n)^\top$, where $\{x_i\}_{i=1}^n$ and $\{y_i\}_{i=1}^n$ are sequences of independent, mean-zero, unit variance, sub-Gaussian random variables. Let $\mathcal A$ be a class of $n\times m$ matrices. For $A_1,A_2\in\mathcal A$, define the $d_q$-metric as $d_q(A_1,A_2)\defeq \|A_1-A_2\|_q$, where $\|A\|_q$ are the Schatten norms of matrix $A$: $\|A\|_q\defeq\Big(\sum\limits_{i=1}^{\min(m,n)}\sigma_i^q(A)\Big)^{1/q}$ for $1\leq q<\infty$, and $\|A\|_q\defeq\sigma_1$ for $q=\infty$, if $A$ has the singular values $\sigma_1\geq\cdots\geq\sigma_{\min(m,n)}$. For $\epsilon>0$, the $\epsilon$-covering number of $\mathcal A$ with respect to the $d_q$-metric is denoted by $\mathcal N(\epsilon,\mathcal A, d_q)$.
\begin{lemma}\label{iid}
Define $E=\sup\limits_{A\in\mathcal A}\max(\E|A^\top\bm x|_2^2,\E|A^\top\bm y|_2^2)+\gamma_2^2(\mathcal A,d_\infty)+\bar\Delta_2(\mathcal A)\gamma_2(\mathcal A,d_\infty)$, $V=\bar\Delta_\infty(\mathcal A)\big(\bar\Delta_2(\mathcal A) + \gamma_2(\mathcal A,d_\infty)\big)$, $U=\bar\Delta^2_\infty(\mathcal A)$, where $\gamma_p(\mathcal A,d_q)\lesssim\int_0^\infty\big(\log\mathcal N(\epsilon,\mathcal A,d_q)\big)^{1/p}d\epsilon$, $\bar\Delta_q(\mathcal A)=\sup\limits_{A\in\mathcal A}\|A\|_q$.
For any $u\geq0$, there exists $c_1,c_2>0$ such that
$$\P\Big(\sup_{A\in\mathcal A}|\bm x^\top AA^\top\bm y|\geq c_1E+u\Big)\leq2\exp\big(-c_2\min(u^2/V,u/U)\big).$$
\end{lemma}
A special case of Lemma \ref{iid}, with $\bm x=\bm y$, is stated in Theorem 6.2 of \cite{dirksen2015tail}. The applicability of the rate in a general context is evidently clear.
Given an $s \in \mathbb{N}$, consider $s$ probability spaces denoted as $(\Omega_{1}, \P_{1}), \ldots,(\Omega_{s}, \P_{s})$. Suppose we have a parameter set $\mathcal T$ containing $s$-tuples $\tau=(\tau_{1}, \ldots, \tau_{s})$. For each $\tau \in \mathcal T$, we have an $s$-tuple $X_{\tau}=$ $(X_{\tau_{1}}, \ldots, X_{\tau_{s}})$ of sub-exponential random variables $X_{\tau_{i}}: \Omega_{i} \rightarrow \R$, and define the sub-exponential norm as $\|X_{\tau_i}\|_{\psi_1}=\inf\{v>0:\E\exp(|X_{\tau_i}|/v)\leq 2\}$. Consider the empirical process given by
$$
W_{\tau}=\frac{1}{s} \sum_{i=1}^{s}(X_{\tau_{i}}-\E X_{\tau_{i}}),
$$
%In the terminology used here,
Bernstein's inequality (referenced, for instance, as Lemma 5.1 of \citet{dirksen2015tail}) implies that the process $(W_{\tau})_{\tau \in \mathcal T}$ exhibits a mixed tail behavior with respect to the metrics $(\frac{1}{s} d_{1}, \frac{1}{\sqrt{s}} d_{2})$, where
\begin{align*}
d_{1}(\tau, \tau')&=\max _{1 \leq i \leq s}\|X_{\tau_{i}}-X_{\tau'_{i}}\|_{\psi_{1}}, \,d_{2}(\tau, \tau')=\Big(\frac{1}{s} \sum_{i=1}^{s}\|X_{\tau_{i}}-X_{\tau'_{i}}\|_{\psi_{1}}^{2}\Big)^{1/2}.
\end{align*}

%The above theorem can directly be applied to find the following tail bound.

\begin{lemma}%(Uniform rate for the suprema of empirical processes)
[Corollary 5.2 of \cite{dirksen2015tail}, Supremum of Empirical Processes]\label{emp}
%Let $\{A_{t}\}_{t\in \mathcal T}$ be the process of averages defined as above and
Let $\sigma, K>0$ be constants such that
$$
\sup _{\tau \in \mathcal T} \frac{1}{s} \sum_{i=1}^{s} \E|X_{\tau_{i}}-\E X_{\tau_{i}}|^{r} \leq \frac{r!}{2} \sigma^{2} K^{r-2}, \quad r=2,3, \ldots.
$$
Then, for any $1 \leq q<\infty$,
$$
\Big(\E\sup_{\tau \in \mathcal T}|W_{\tau}|^{q}\Big)^{1/q} \lesssim\Big(\frac{1}{\sqrt{s}} \gamma_{2}(\mathcal T, d_{2})+\frac{1}{s} \gamma_{1}(\mathcal T, d_{1})\Big)+\sqrt{q} \frac{\sigma}{\sqrt{s}}+q \frac{K}{s} .
$$
In particular, there exist constants $c, C>0$ such that for any $u \geq 1$,
$$\P\Big(\sup _{\tau \in \mathcal T}|W_{\tau}| \geq C\Big(\frac{1}{\sqrt{s}} \gamma_{2}(\mathcal T, d_{2})+\frac{1}{s} \gamma_{1}(\mathcal T, d_{1})\Big)+c\Big(\frac{\sigma}{\sqrt{s}} \sqrt{u}+\frac{K}{s} u\Big)\Big) \leq e^{-u}.$$
\end{lemma}

\begin{lemma}[\cite{burkholder1988sharp,rio2009moment}]\label{buck} Let $q>1$, $q'=\min(q,2)$. Let $M_n = \sum^n_
{t=1} \xi_t$, where $\xi_t \in \mathcal{L}^{q}$ (i.e., $\|\xi_t\|_q<\infty$) are martingale differences. Then
\begin{equation*}
\|M_n\|_q^{q'} \leq K^{q'}_q\sum_{t=1}^n\|\xi_t\|^{q'}_q \quad \mbox{where} \quad K_q = \max((q -1)^{-1},\sqrt{q-1}).
\end{equation*}
\end{lemma}

\begin{lemma}[Tail Probabilities for High-dimensional Partial Sums] \label{exp}
% For a mean zero $p$-dimensional random variable $X_t\in \R^p$ ($p>1$), let $S_n = \sum^{n}_{t=1}X_t$ and assume that $\Phi_{\psi_\nu,\varsigma} =\max\limits_{1\leq j\leq p}\sup\limits_{q\geq2}q^{-\nu}\|X_{j,\cdot}\|_{q,\varsigma}< \infty$ for some $\nu\geq0$, and let $\gamma = 2/(1+ 2\nu)$. Then for all $x>0$, we have
% $$\P(|S_n|_\infty\geq x) \lesssim p \exp\{-C_{\gamma}x^{\gamma}/(\sqrt{n}\Phi_{\psi_\nu,0})^{\gamma}\},$$
% where $C_\gamma$ is a positive constant only depends on $\gamma$.
For a mean zero $p$-dimensional random variable $X_t\in \R^p$ ($p>1$), let $T_n = \sum^{n}_{t=1}X_t$ and $T_{n,m}=\sum_{t=1}^nX_{t,m}$, where $X_{t,m}=\E(X_t\mid\varepsilon_{t-m},\ldots,\varepsilon_t)$. Assume that $\Phi_{\psi_\nu,\varsigma} =\max\limits_{1\leq j\leq p}\sup\limits_{q\geq2}q^{-\nu}\|X_{j,\cdot}\|_{q,\varsigma}< \infty$ for some $\nu\geq0$, and let $\gamma = 2/(1+ 2\nu)$. Then for all $x>0$, we have
$$\P(|T_n - T_{n,m}|_\infty\geq x) \lesssim p \exp\{-C_{\gamma}x^{\gamma}m^{\varsigma\gamma}/(\sqrt{n}\Phi_{\psi_\nu,\varsigma})^{\gamma}\},$$
where $C_\gamma$ is a positive constant only depends on $\gamma$.

%	 Let $X_{ij}$ be random variables with mean zero and define $\gamma_{j,\alpha} =  \|X_{.j}\|_{\psi_r,\alpha}$. Assume that $\|X_{.j}\|_{\psi_r,\alpha}< \infty$.  Then, we have
%	\begin{equation*}
%	\P(\max_j \sum^n_{i=1} X_{ij}\geq u) \leq p \max_j c_{\alpha} \exp (-(u)^{\alpha}/(\sqrt{n}^{\alpha}\gamma_{j,\alpha}^{\alpha} c'_{\alpha})),
%	\end{equation*}
%	where $\alpha =2$ corresponds for the sub exponential case, and $\alpha = 1$ corresponds to the sub Gaussian case.
\end{lemma}

Lemma \ref{exp} follows from Lemma C.3 of \citet{ZW_sup} and applying the Bonferroni inequality.
In particular, $\nu = 1$ corresponds to the sub-exponential case, and $\nu = 1/2$
corresponds to the sub-Gaussian case.

% \begin{lemma}[Freedman's inequality]\label{free}
% 		Let $\{\xi_{a,t}\}_{t=1}^n$ be a martingale difference sequence with respect to the filtration $\{\mathcal F_t\}_{t=1}^n$. Let $V_a = \sum^n_{t=1}\E(\xi^2_{a,t}| \mathcal{F}_{t-1})$ and $M_a = \sum^n_{t=1} \xi_{a,t}$. Then, for $x,u,v>0$, we have
% 		\begin{equation*}
% 		\P(\max_{a \in \mathcal{A}} |M_a|\geq x) \leq \sum^{n}_{t=1}\P(\max_{a\in \mathcal{A}} \xi_{a,t}\geq u)+ 2 \P(\max_{a \in \mathcal{A} } V_a \geq v )+ 2|\mathcal{A}| e^{-x^2/(2xu+ 2v)},
% 		\end{equation*}
% 		where $\mathcal A$ is an index set with $|\mathcal A|<\infty$.
% 	\end{lemma}
% Lemma \ref{free} is a maximal form of Freedman's inequality \citep{freedman1975tail}.

\subsubsection{Relaxing the Tail Assumption}
Lemma \ref{iid} primarily focuses on sub-Gaussian random variables. Now, we explore ways to ease this assumption. Specifically, we examine the case of $m=1$, i.e., $A$ is a vector of dimension $n\times1$.
% Lemma \ref{iid} is majorly concerning the sub-Gaussian random variables. We now discuss how to relax such assumption.We maintain the notations in Lemma \ref{iid}, with $q_n = (E +\sqrt{V}+U) $, a slowly varying constant $\gamma_n$, and $A = \delta$, with $|\delta|_{\max}\leq c_n$.
\begin{theorem}\label{fattertail}
Assume that $x_i,y_i\in\mathcal L^q$ for some $q>2$, $\sup_{A\in\mathcal A}|A|_2\lesssim\sqrt{n}c_n$. Then, we have with probability $1-\smallO(1)$,
$$\sup_{A\in\mathcal A}|\bm x^\top AA^\top\bm y|\lesssim(E+\sqrt{V}+U)n^{2r}\gamma_n + (n^{-(q-2)r/2+1}c_n)^2\gamma_n,$$
where $E,V,U$ are defined in Lemma \ref{iid}, $0<r\leq1/2$, and $\gamma_n$ is a slowly growing sequence of positive constants.
% Assume that $\max_i\|X_i\|_{q}$ are bounded with $q>6$, then we have, exist a $0<r \leq 1/2$ and $qr >3/2${\color{red}dominate by the first term },
% \begin{equation}
% \mbox{sup}_{A \in \mathcal{A}} |X^{\top}AA^{\top}X|/c_n^2 \lesssim_p  [E + \sqrt{V}+ U] \gamma_n n^{2r} + (n^{-qr/2 + 1} \gamma_n)^2.
% \end{equation}
% $|\delta|_2\leq \sqrt{n}c_{\max}c_n$, where $c_{\max}$ is a positive constant.
\end{theorem}

\begin{proof}
We proceed without loss of generality by assuming $\bm x=\bm y$. Let $\bm z=(z_1,\ldots,z_n)^\top$, where $z_i=x_i\mathbf 1(|x_i|\leq M)$ represents the truncated random variables, with $M=cn^r$ for constants $c>0$ and $0<r\leq1/2$. Denote $\bm w=(w_1,\ldots,w_n)^\top$, where $w_i=x_i\mathbf 1(|x_i|> M)=x_i-z_i$. It follows that
\begin{align*}
\sup_{A\in\mathcal A}|\bm x^\top AA^\top\bm x|&\leq \sup_{A\in\mathcal A}|\bm z^\top AA^\top\bm z| + \sup_{A\in\mathcal A}|\bm w^\top AA^\top\bm z| + \sup_{A\in\mathcal A}|\bm x^\top AA^\top\bm w|\\
&\leq \sup_{A\in\mathcal A}|\bm z^\top AA^\top\bm z| + 2|\bm w|_2\sup_{A\in\mathcal A}|AA^\top\bm x|_2\\
&\leq \sup_{A\in\mathcal A}|\bm z^\top AA^\top\bm z| + 2|\bm w|_2\sup_{A\in\mathcal A}|A|_2\sup_{A\in\mathcal A}|A^\top\bm x|.
\end{align*}
Utilizing Lemma \ref{iid}, we bound the first term as $\sup_{A\in\mathcal A}|\bm z^\top AA^\top\bm z|\lesssim_\P(E+\sqrt{V}+U)n^{2r}\gamma_n$, where $\gamma_n$ is a sequence of positive numbers growing slowly. Moreover, by Markov inequality, we have
$$\P(|\bm w|_2^2>s^2)\leq n\E[x_i^2\mathbf 1(|x_i|>M)]/s^2\leq n\E|x_i|^{q}/(s^2M^{q-2}),$$
Note that $\E|x_i|^{q}$ is bounded for some $q>2$. By letting $s^2=n^{-(q-2)r+1}\gamma_n$, we have the tail probability tends to zero as $n\to\infty$, that is $|\bm w|_2^2\lesssim_\P n^{-(q-2)r+1}\gamma_n$. Given $\sup_{A\in\mathcal A}|A|_2\lesssim\sqrt{n}c_n$, it follows that
$$\sup_{A\in\mathcal A}|\bm x^\top AA^\top\bm x|\lesssim_\P(E+\sqrt{V}+U)n^{2r}\gamma_n + (n^{-(q-2)r/2+1}c_n)^2\gamma_n.$$
\end{proof}

When $r=\log\log n/\log n$ (implying $n^r=\log n$), the second term in the bound becomes $(\log n)^{-(q-2)+2/r}c_n^2\gamma_n$. For sufficiently large $q$, this term would be dominated by the first term, resulting in a similar rate as in the sub-Gaussian case, albeit subject to a scaling factor up to a logarithmic order.

\subsubsection{Relaxing the Independence Assumption}
%{The dependence case }
Consider two processes $\{x_i\}_{i=1}^n$ and $\{y_i\}_{i=1}^n$ (with $x_i,y_i\in\R$), both of which are stationary with zero mean and unit variance, and admit the representations: $x_i=f(\mathcal F_i)$ and $y_i=g(\mathcal F_i)$, where $f,g$ are measurable functions, and $\mathcal F_i\defeq\{\ldots,\eta_{i-1},\eta_i\}$, with $\eta_i$ for $i\in\mathbb Z$ being i.i.d.\ random elements.

Define the projector operator $\mathcal P_\ell(x_iy_j)\defeq \E(x_iy_j\mid\mathcal F_{s-\ell})-\E(x_iy_j\mid\mathcal F_{s-\ell-1})$, where $s=\min(i,j)$. Note that $\mathcal P_\ell(x_iy_j)$ is m.d.s.\ with respect to $\mathcal F_{s-1}$. For $q\geq1$, $\varsigma>0$, we introduce the norm $$\Theta_{q,\varsigma}\defeq\sup_{d\geq0}(d+1)^{\varsigma}\sum_{\ell\geq d}\max_{1\leq i,j\leq n} \|\mathcal P_\ell(x_iy_j)\|_q,$$
to measure the degree of dependence. This norm is directly linked to the dependence adjusted norm for $x_iy_j$. Additionally, we denote a truncation argument as $x_i^m\defeq\E(x_i\mid \eta_{i-m},\ldots,\eta_{i})$.

%Let $X_i = g(\eta_i, \eta_{i-1},\cdots, \eta_{-\infty})$ and $Y_j = f(\eta_i, \eta_{i-1},\cdots, \eta_{-\infty})$ be mean zero stationary processes with the following dependence adjusted norm condition $\Phi_{\phi_v, \xi}$.
%Assume the $\varepsilon-$ covering of $\delta = (\delta_1, \delta_2, \cdots, \delta_p)$ is $N_n$ (with countable element and the constant $\varepsilon>0$).
%Let $P_{l} X_iY_j = \E(X_i Y_j|\mathcal{F}_{i-l}) - \E(X_iY_j|\mathcal{F}_{i-l-1})$.
%Denote $\sup_{d\geq 0} \sum_{l\geq d} d^{-\xi}\|P_{l} X_i Y_j\|_q \defeq \Theta_{x,y,\xi}(i,j,q)$ for a positive constant $\xi>0$.

%Recall that we define for a integer $m$, a truncated lag version, $X_{i,m}= \E( X_{i}|\eta_i, \eta_{i-1},\cdots, \eta_{i-m})$.
%Let $c_n>0$ be a positive constant.
\begin{Assumption}\phantomsection\label{reff}
\begin{enumerate}
\item[(i)] Assume that $\sup_{q\geq2}q^{-\nu}\Theta_{q,\varsigma}<\infty$ for some $\nu\geq0$, $\varsigma>0$.%, and
%$$\sup_{d\geq0}(d+1)^{\varsigma}\sum_{\ell\geq d}\max_{1\leq i,j\leq n}\Big\|\mathcal P_{\ell}(x_iy_j-x_{i}^my_{j}^m)\Big\|_q \leq m^{-\varphi}\Theta_{q,\varsigma},\,\text{ for some } \varphi\geq0,\varsigma>0.$$
%$\|P_l(X_iY_j-X_{i,m} Y_{j,m})\|_{q} \leq m^{-\alpha}\Theta_{x,y,\xi}(i,j,q)$ for $\alpha> 0 $.
% For some constant $\alpha>0$, $\Theta_{x,y,\xi}(i,j,q)\leq \infty$, $\|\sum_{l\geq d}  P_{l}  (X_iY_j-X_{i,m}Y_{j,m})\|_q \leq m^{-\alpha}\Theta_{x,y,\xi}(i,j,q)$.
%The subGaussain norm dependence adjusted norm for $X_iY_j$ is finite.
\item[(ii)] The sub-Gaussian norm $\|x_iy_j\|_{\psi_{1/2}}<\infty$, for all $i,j=1,\ldots,n$.
\item[(ii')] The sub-exponential norm $\|x_iy_j\|_{\psi_{1}}<\infty$, for all $i,j=1,\ldots,n$.
%For all $i,j$, we have  $\|X_iY_j \|_{\psi_2} \leq 1$. \quad ii)'  $\|X_iY_j \|_{\psi_1} \leq 1$.
\item[(iii)] There exists a finite set $\mathcal A=\{\delta\in\R^n:|\delta|_\infty\leq c_{\max}\}$, for some $c_{\max}>0$,
%The $\epsilon$-covering number of $\mathcal A$ with respect to $d_\infty$ is bounded by $N_n$,
such that
$$\int_0^\infty\big(\log\mathcal N(\epsilon,\mathcal A,d_\infty)\big)^{1/2}d\epsilon\lesssim(\log N_n)^{1/2},$$
where $N_n$ is a sequence of positive constants greater than 1.
%The cardinality of the functional class $\mathcal{A}= \{\delta: \max_i |\delta_i| \leq c_{max} \}$ is $N_n$, with a constant $c_{\max}>0$.
\end{enumerate}
\end{Assumption}

\begin{theorem}\label{label:Ubound}
Under Assumption \ref{reff} with (ii), we have
%for a constant $c_2>0$, we shall prove the following bounds,
\begin{eqnarray*}
&&\sup_{\delta\in\mathcal A}\Big|\sum_{i=1}^n\sum_{j=1}^n \delta_i x_i\delta_j y_j \Big|\big/c_{\max}^2
\lesssim_\P n+\log N_n + \sqrt{n} (\log N_n)^{1/2} +nm^{-\varsigma}\log N_n.
\end{eqnarray*}
Under Assumption \ref{reff} with (ii'), we have
$$\sup_{\delta\in\mathcal A}\Big|\sum_{i=1}^n\sum_{j=1}^n \delta_ix_i\delta_j y_j\Big| \big/ c_{max}^2 \lesssim_\P \{n+\log N_n + \sqrt{n} (\log N_n)^{1/2}\}(\log n)^2\gamma_n + nm^{-\varsigma}(\log N_n)^{3/2},$$
%$$\sup_{\delta\in\mathcal A}\Big|\sum_{i=1}^n\sum_{j=1}^n \delta_ix_i\delta_j y_j\Big| \Big/ c_{max}^2 \lesssim_\P n (\log N_n)^{1/2} \gamma_n,$$
where $\gamma_n$ is a slowly growing sequence of positive constants, see the proof of Theorem \ref{fattertail}.
\end{theorem}
\begin{proof}
\underline{Step 1:} First, we need to prove that the deviation $\sup_{\delta\in\mathcal A}\big|\sum_{i=1}^n\sum_{j=1}^n(\delta_ix_i\delta_jy_j - \delta_ix_i^m\delta_jy_j^m)\big|$, subject to the truncation error, is sufficiently small.

For each $i=1,\ldots,n$, define $z_i(\delta)\defeq\sum_{j=1}^n(\delta_ix_i\delta_jy_j - \delta_ix_i^m\delta_jy_j^m)$.
% For any $\delta\in\mathcal A$, by applying Lemma \ref{buck} for $q\geq 2$, we have
% \begin{align*}
% \sup_{d\geq0}(d+1)^{\varsigma}\sum_{\ell\geq d}\Big\|\mathcal P_\ell(z_i(\delta)\Big\|_q &=\sup_{d\geq0}(d+1)^{\varsigma}\sum_{\ell\geq d}\Big\|\sum_{j=1}^n\mathcal P_\ell(\delta_ix_i\delta_jy_j - \delta_ix_i^m\delta_jy_j^m)\Big\|_q\\
% &\lesssim\sqrt{n}\sup_{d\geq0}(d+1)^{\varsigma}\sum_{\ell\geq d}\max_{1\leq j\leq n}\Big\|\mathcal P_\ell(\delta_ix_i\delta_jy_j - \delta_ix_i^m\delta_jy_j^m)\Big\|_q\\
% &\leq \sqrt{n}c_{\max}^2\sup_{d\geq0}(d+1)^{\varsigma}\sum_{\ell\geq d}\max_{1\leq j\leq n}\Big\|\mathcal P_\ell(x_iy_j - x_i^my_j^m)\Big\|_q\\
% &\leq \sqrt{n}c_{\max}^2m^{-\varphi}\Theta_{q,\varsigma},
% \end{align*}
% where the last inequality is ensured by Assumption \ref{reff}(i). Next,
On Assumption \ref{reff}(i), by applying Lemma \ref{exp} on the summation $\sum_{i=1}^nz_{i}(\delta)$, we obtain that
$$\sup_{\delta\in\mathcal A}\Big|\sum_{i=1}^n\sum_{j=1}^n(\delta_ix_i\delta_jy_j - \delta_ix_i^m\delta_jy_j^m)\Big|\lesssim_{\P}nc_{\max}^2m^{-\varsigma}(\log N_n)^{1/\gamma},$$
where $\gamma=1$ for the sub-Gaussian case, and $\gamma=2/3$ for the sub-exponential case.

\underline{Step 2:} Next, we shall bound the sum of the truncated terms. Divide the sample $\{1,\ldots,n\}$ into $L=\lfloor n/m\rfloor$ blocks: $A_l$, $l=1,\ldots,L$, each of size $m$. Without loss of generality, assume that $L$ is an even number and let $B_o,B_e\subseteq\{1,\ldots,L\}$ be the indices sets for the odd and even blocks, respectively.

For each block $l=1,\ldots,L$, define $\tilde x_l^m(\delta)\defeq\sum_{i\in A_l}\delta_ix_i^m$, $\tilde y_l^m(\delta)\defeq\sum_{i\in A_l}\delta_iy_i^m$. It follows that
\begin{align*}
\sum_{i=1}^n\sum_{j=1}^n\delta_ix_i^m\delta_jy_j^m &=\sum_{l=1}^L\sum_{l'=1}^L\sum_{i\in A_l}\sum_{j\in A_{l'}}\delta_ix_i^m\delta_jy_j^m\\
&=\sum_{l\in B_o}\sum_{l'\in B_o}\tilde x_l^m(\delta)\tilde y_{l'}^m(\delta) + \sum_{l\in B_o}\sum_{l'\in B_e}\tilde x_l^m(\delta)\tilde y_{l'}^m(\delta) + \sum_{l\in B_e}\sum_{l'\in B_e}\tilde x_l^m(\delta)\tilde y_{l'}^m(\delta).
\end{align*}
It is worth noting that $\{\tilde x_l^m(\delta)\}_{l\in B_o}$ and $\{\tilde x_l^m(\delta)\}_{l\in B_e}$ are sequences of independent random variables, similarly for $\{\tilde y_l^m(\delta)\}_{l\in B_o}$ and $\{\tilde y_l^m(\delta)\}_{l\in B_e}$. We shall apply Lemma \ref{iid} to bound each of the terms above. Taking the first term as an example, we denote a scaling constant $c_l(\delta)>0$ such that $\operatorname{Var}(\tilde x_l^m(\delta)/c_l(\delta))=\operatorname{Var}(\tilde y_l^m(\delta)/c_l(\delta))=1$, and stack them into a vector $\bm c(\delta)=(c_l(\delta))_{l\in B_o}$. %To correspond with the set $\mathcal A$, we define the collection $\mathcal C(\mathcal A)$ to encompass all $\bm c(\delta)$ associated with each $\delta\in\mathcal A$.
Using $\bm c(\delta)$ as the weights in the quadratic form, it follows that for any $u\geq0$, there exist $c_1,c_2>0$ such that
$$\P\Big(\sup_{\delta\in\mathcal A}\Big|\sum_{l\in B_o}\sum_{l'\in B_o}\tilde x_l^m(\delta)\tilde y_{l'}^m(\delta)\Big|\big/c_{\max}^2\geq c_1E+u\Big)\leq2\exp\big(-c_2\min(u^2/V,u/U)\big),$$
where $E=n + \log N_n + \sqrt{n}(\log N_n)^{1/2}$, $V=n+\sqrt{n}(\log N_n)^{1/2}$, $U=n$.

To prove the case of (ii'), we just need to replace Lemma \ref{iid} with Theorem \ref{fattertail}. In particular, we choose $r$ such that $n^r=\log n$ and assume the moments condition holds with a sufficiently high order to absorb the second term in the bound.

By combining the results of Step 1 and Step 2, we can conclude the proof.
\end{proof}

For a class of measurable functions $\mathcal G$ mapping to the real space $\R$. let the $d_1$-metric and $d_2$-metric be denoted as $d_1(f,g)=\max_{1\leq i\leq n}\|f(x_i)-g(x_i)\|_{\psi_1}$ and $d_2(f,g)=\{\E\|f(x_i)-g(x_i)\|_{\psi_1}^2\}^{1/2}$, $\forall f,g\in\mathcal G$.
%For $\epsilon>0$, define the $\epsilon$-covering number with respect to the $d_r$-metric as $\mathcal N(\epsilon, \mathcal G, d_r)$, for $r=1,2$.
In the following theorem, we will derive a bound for the empirical process $\sup_{f\in\mathcal G}|n^{-1}\sum_{i=1}^n\{f(x_i)-\E f(x_i)\}|$ under certain conditions.

Consider analogous definitions as mentioned earlier:
\begin{align*}
&\mathcal P_\ell(f(x_i))\defeq\E(f(x_i)\mid\mathcal F_{i-\ell}) - \E(f(x_i)\mid\mathcal F_{i-\ell-1}),\, \Theta_{f,q,\varsigma}\defeq\sup_{d\geq0}(d+1)^{\varsigma}\sum_{\ell\geq d}\max_{1\leq i\leq n} \|\mathcal P_\ell(f(x_i))\|_q,\\
&f^m(x_i)\defeq\E(f(x_i)\mid\eta_{i-m},\ldots,\eta_i).
\end{align*}

% Denote $\mathcal{G}_n(\vps)$ as a functional class with $\vps-$ covering number as $N_n(\varepsilon)$ utilizing a certain norm.
% We shall study the object $\sup_{f(.)\in \mathcal{G}_n(\vps)}|n^{-1}\sum_i \{f(x_i)-\E(f(x_i)\}|.$
\begin{Assumption}\phantomsection\label{reff1}
\begin{itemize}
\item[(i)] The function class $\mathcal G$ is enveloped with $F=\sup_{f\in\mathcal G}|f|$, with $\max_{1\leq i\leq n}(\E|F(x_i)|^2)^{1/2}<c_n$. Additionally, assume that there exists a sequence of positive constants (greater than 1) $N_n$, such that
$$\int_0^\infty\log\mathcal N(\epsilon,\mathcal G,d_1)d\epsilon\lesssim\log N_n,\quad \int_0^\infty\big(\log\mathcal N(\epsilon,\mathcal G,d_2)\big)^{1/2}d\epsilon\lesssim(\log N_n)^{1/2}.$$
% Assume the functional class $\mG_n(\vps)$ has envelop $F_n$. The covering number with respect to $d_2(f,g)= \{\E((f(x)-g(x))^2)\}^{1/2}$ is $N_n(\vps)$. The envelope function is $F_n(x_i)$.
\item[(ii)] For any $f\in\mathcal G$, assume that $\sup_{q\geq 1} q^{-\nu}\Theta_{f,q,\varsigma} < \infty$ for some $\nu\geq0$, $\varsigma>0$.%, and
%$$\sup_{d\geq0}(d+1)^{\varsigma}\sum_{\ell\geq d}\max_{1\leq i\leq n}\Big\|\mathcal P_{\ell}(f(x_i)-f^m(x_i))\Big\|_q \leq m^{-\varphi}\Theta_{f,q,\varsigma},\,\text{ for some } \varphi\geq0,\varsigma>0.$$
% For all integer  $p\geq 2$ with finite $p-$ moment, we let $[\E|f(x_i)-\E(f(x_i)|\eta_i,\cdots, \eta_{i-m})|^p]^{1/p}\lesssim m^{-\alpha} \theta_{x,F}$, where $\theta_{x,F}$ is a constant, for all $f(.)\in \mathcal{G}_n(\vps) $.
% Assume that $\sup_{q\geq 1} q^{-v}\Theta_{x,F}(q) < \infty $, for some constant $v>0$.
% {\color{red}
% We introduce the norm $\Theta_{f, q,\varsigma}\defeq\sup_{d\geq0}d^{-\varsigma}\sum_{\ell\geq d} \|\mathcal P_\ell(f(x_i))\|_q$.
% Assume that $\sup_{q\geq2}q^{-\nu}\Theta_{f, q,\varsigma}<\infty$ for some $\nu\geq0$, $\varsigma>0$, and
% $$\sup_{f(.)}\sup_{d\geq0}d^{-\varsigma}\Big\|\sum_{\ell\geq d}\big|\mathcal P_{\ell}(f(x_i)-f(x_i)^m)\big|\Big\|_q \leq m^{-\varphi}\Theta_{f,q,\varsigma},\,\text{ for some } \varphi\geq0,\varsigma>0.$$
% }
\item[(iii)] There exist constants $\sigma,K>0$ such that
$$\sup_{f\in\mathcal G}\frac{1}{n} \sum_{i=1}^{n} \E\left|f(x_i)-\E f(x_i)\right|^{q} \leq \frac{q !}{2} \sigma^{2} K^{q-2}, \quad(q=2,3, \ldots).$$
% The subGaussian/subExponential dependence adjusted norm of $f(x_i)$ is finite, for a positive constant $\sigma>0$. Let $q\geq 2$ be an integer. Moreover,  exists a $K$ such that $$
% \sup _{f(.) \in \mG_n(\varepsilon)} \frac{1}{n} \sum_{i=1}^{n} \E\left|f(x_i)-\E f(x_i)\right|^{q} \leq \frac{q !}{2} \sigma^{2} K^{q-2}, \quad(q=2,3, \ldots)
% .$$
\end{itemize}
\end{Assumption}
% Let $\gamma = 2/(1+ 2v)$.
\begin{theorem}\label{emplemma}
%Let $\max_i(\E|F_n(x_i)^2|)^{1/2} \leq c_n$.
Under Assumption \ref{reff1}, we have
%for a constant $c_3>0$, we shall prove the following bounds,
\begin{eqnarray*}
\sup_{f\in\mathcal G}\Big|\frac{1}{n}\sum_{i=1}^n\{f(x_i)-\E f(x_i)\}\Big|\big/c_n\lesssim_\P m^{-\varsigma}(\log N_n)^{1/\gamma}/\sqrt{n} + \sqrt{(\log N_n)/n} + \sqrt{m}(\log N_n)/n,
% &&\sup_{f(.)\in \mG_n(\vps)}|n^{-1}\sum_i \{f(x_i)-\E(f({x_i})\}| \lesssim_p c_n (\sqrt{\log N_n(\vps)}/{\sqrt{n}} + {\log N_n(\vps)}/{(\sqrt{m}L)}
% \\&&+c_n(\log(N_n(\vps)))^{1/\gamma} m^{-\alpha}\sqrt{n}^{-1}).
\end{eqnarray*}
where $\gamma=1$ for the sub-Gaussian case, and $\gamma=2/3$ for the sub-exponential case.
\end{theorem}
\begin{proof}
To begin, we decompose the process as:
$$n^{-1}\sum_{i=1}^n\{f(x_i)-\E f(x_i)\}=n^{-1}\sum_{i=1}^n\{f(x_i) - f^m(x_i) + f^m(x_i) - \E f^m(x_i)\}.$$
In the subsequent two steps, we will analyze each part of the deviations.

\underline{Step 1:} Given Assumptions \ref{reff1}(i)-(ii), applying Lemma \ref{exp} yields:
$$\sup_{f\in\mathcal G}|n^{-1}\sum_{i=1}^n\{f(x_i) - f^m(x_i)\}|\lesssim_\P n^{-1/2}m^{-\varsigma}(\log N_n)^{1/\gamma}c_n,$$
where $\gamma=1$ for the sub-Gaussian case, and $\gamma=2/3$ for the sub-exponential case.

\underline{Step 2:} We partition the sample into blocks. The definitions of $L$, $A_l$, $B_o$, and $B_e$ remain consistent with those in the proof of Theorem \ref{label:Ubound}. For each block $l=1,\ldots,L$, define $z_{l,f}^m\defeq\sum_{i\in A_l}\{f^m(x_i) - \E f^m(x_i)\}$. It follows that
$$\sum_{i=1}^n\{f^m(x_i) - \E f^m(x_i)\}=\sum_{l=1}^Lz_{l,f}^m=\sum_{l\in B_o}z_{l,f}^m + \sum_{l\in B_e}z_{l,f}^m.$$
Note that $\{z_{l,f}^m\}_{l\in B_o}$ and $\{z_{l,f}^m\}_{l\in B_o}$ are sequences of independent random variables. We shall apply Lemma \ref{emp} to bound each of them.

Recalling the definition of the projector operation and utilizing Lemma \ref{buck}, we observe that
\begin{align*}
\|z_{l,f}^m\|_q&=\Big\|\sum_{\ell\geq0}\mathcal P_\ell(z_{l,f}^m)\Big\|_q\\
&\leq\sum_{\ell\geq0}\Big\|\sum_{i\in A_l}\mathcal P_\ell(f^m(x_i)-\E f^m(x_i))\Big\|_q\\
&\lesssim\sqrt{m}\sum_{\ell\geq0}\max_{1\leq i\leq n}\|\mathcal P_\ell(f^m(x_i))\|_q\leq\sqrt{m}\Theta_{f,q,\varsigma}.
\end{align*}
Then, after proper scaling on the block sums, we can evoke Lemma \ref{emp} to obtain that
$$\sup_{f\in\mathcal G}\Big|n^{-1}\sum_{l=1}^Lz_{l,f}^m\Big|\big/c_n\lesssim_{\P}\sqrt{\log N_n}/\sqrt{n} + \sqrt{m}\log N_n/n.$$

By combining the results of Step 1 and Step 2, we can conclude the proof.
\end{proof}

\subsection{Proofs for Section \ref{firststep}}In this subsection, we present the proofs for Section \ref{firststep} for the model without time effects $\gamma_t$. When time effects are present, cross-sectional demeaning of the variables introduces weak cross-sectional dependence of order $1/N$. This leads to corrections in higher-order terms but does not affect the consistency of step 1 or the convergence rates.

\subsubsection{Oracle Order of $s_t^{*}$} \label{oracle}
Recall that $V_{it}$ is a vector of length $m_t$ that gathers the instruments for each $t=1,\ldots,T-1$. Let $J_t\subseteq\{1,\ldots,m_t\}$ be a set of indices with cardinality $|J_t|\leq s_t\leq m_t$, and let $J_t^c=\{k\in\{1,\ldots,m_t\}:k\notin J_t\}$ be the complement set. Denote $V_{it,J_t}$ and $V_{it,J_t^c}$ as the sub-vectors of $V_{it}$ corresponding to $J_t$ and $J_t^c$, respectively. Consider the linear projection
$$W_{it} = V_{it}^\top\Pi_t^0 + \eta_{it} = V_{it,J_t}^\top\Pi_{t,J_t}^0 + \underbrace{V_{it,J_t^c}^\top\Pi_{t,J_t^c}^0 + \eta_{it}}_{=:\check\eta_{it}},$$
where $\Pi_t^{0} \defeq \bar\E(V_{it}V_{it}^{\top})^{-1}\bar\E(V_{it}W_{it})$, $\Pi^0_{t,J_t}$ and $\Pi^0_{t,J_t^c}$ are the sub-vectors of $\Pi^0_{t}$ corresponding to $J_t$ and $J_t^c$, respectively. Let $\delta_{-J_t}^0\defeq\Pi_t^0 - \Pi^{0}_{J_t}$, where $\Pi_{J_t}^{0}$ is a vector of length $m_t$ with elements corresponding to $J_t$ being $\Pi_{t,J_t}^{0}$ and zeros elsewhere. Moreover, define $\Pi_{t,J_t}^{\dagger} \defeq \bar\E(V_{it,J_t}V_{it,J_t}^{\top})^{-1}\bar\E(V_{it,J_t}W_{it})$, and $\widehat\Pi_{t,J_t}\defeq \big(\sum_{i=1}^NV_{it,J_t}V_{it,J_t}^{\top}\big)^{-1}\big(\sum_{i=1}^NV_{it,J_t}W_{it}\big)$.

To choose the optimal value of $s_t$, we consider the oracle risk minimization problem:
\begin{equation}\label{st}
s_t^* = \arg\min_{s_t}\Big\{\min_{J_t:|J_t|\leq s_t}N^{-1}\sum_{i=1}^N(V_{it}^\top\Pi_t^0 - V_{it,J_t}^\top\Pi_{t,J_t}^\dagger)^2 + s_t\check\sigma_t^2/N\Big\},
\end{equation}
where $\check\sigma_t^2\defeq\bar\E(\check\eta^2_{it}\mid V_{it,J_t})$. In the following theorem, we shall illustrate the oracle order of $s_t^*$ for a specific case.

\begin{theorem}[Oracle Order of $s_t^*$]\label{bound}
Under Assumptions \ref{a1}--\ref{a2}, and assuming that $\bar\E(V_{it}^\top\delta_{-J_t}^0)\lesssim c^{-s_t}$ for some constant $c>1$, we can conclude that the optimal $s_t^*$ defined in \eqref{st} is bounded as $s_t^*\asymp\log N\wedge t$.
\end{theorem}

\begin{proof}%[Proof of Theorem \ref{bound}]
Let $e_{it}\defeq V_{it,J_t}^\top\big(\sum_{i=1}^NV_{it,J_t}V_{it,J_t}^{\top}\big)^{-1}\big(\sum_{i=1}^NV_{it,J_t}\check\eta_{it}\big)$. Observe that
\begin{eqnarray*}
&&\frac{1}{N}\sum_{i=1}^N\E\{(V_{it}^\top\Pi_t^0 - V_{it,J_t}^\top\widehat\Pi_{t,J_t})^2\mid V_{it,J_t}\} \\
&=& \frac{1}{N}\sum_{i=1}^N\E\Big(\Big[V_{it}^\top\Pi_t^0 - V_{it,J_t}^\top\Big(\sum_{i=1}^NV_{it,J_t}V_{it,J_t}^{\top}\Big)^{-1}\Big(\sum_{i=1}^NV_{it,J_t}(V_{it,J_t}^\top\Pi_{t,J_t}^0 + \check\eta_{it})\Big)\Big]^2\mid V_{it,J_t}\Big)\\
&=& \frac{1}{N}\sum_{i=1}^N\delta_{-J_t}^{0\top}V_{it}V_{it}^\top\delta_{-J_t}^{0} - \frac{2}{N}\sum_{i=1}^N\E(e_{it}V_{it}^\top\delta_{-J_t}^0\mid V_{it,J_t}) + \frac{1}{N}\sum_{i=1}^N\E(e_{it}^2\mid V_{it,J_t}).
\end{eqnarray*}
By applying the concentration inequality, we can bound the first term in probability by $c^{-2s_t}$. Under the cross-sectional independence assumption, the last term is bounded as
\begin{align*}
\frac{1}{N}\sum_{i=1}^N\E(e_{it}^2\mid V_{it,J_t})&= \frac{1}{N^2}\sum_{i=1}^N\check\sigma_t^2\operatorname{tr}\Big(\Big(\frac{1}{N}\sum_{i=1}^NV_{it,J_t}V_{it,J_t}^{\top}\Big)^{-1}V_{it,J_t}V_{it,J_t}^{\top}\Big) + \smallO_\P(1) \lesssim_\P (N^{-1}+ c^{-2s_t})s_t.
\end{align*}
As for the second term, we have $\frac{1}{N}\sum_{i=1}^N\E(e_{it}V_{it}^\top\delta_{-J_t}^0\mid V_{it,J_t})\lesssim_\P c^{-2s_t}\vee c^{-2s_t}\sqrt{s_t/N}\vee c^{-s_t}\sqrt{s_t}/N$. 

Overall, to minimize the order of $c^{-2s_t}+c^{-2s_t}\sqrt{s_t/N}+ c^{-s_t}\sqrt{s_t}/N+N^{-1}s_t+c^{-2s_t}s_t$,
we find that the oracle order of the optimal $s_t^*$ is approximately given by $s_t^*\asymp\log N\wedge t$.
\end{proof}

As shown in the proof of Theorem \ref{bound}, the approximate sparsity error, as quantified by \eqref{cs}, is bounded as $C_{s_t^*}\lesssim_\P (N^{-1}s_t^*+c^{-2s_t^*}s_t^*)^{1/2}$ is this particular case. Therefore, we can further deduce the oracle bound for $C_{s_t^*}\lesssim_\P\sqrt{(\log N\wedge t)/N}$ based on the optimal $s_t^*\asymp\log N\wedge t$ derived.
The required condition $\bar\E(V_{it}^\top\delta_{-J_t}^0)\lesssim c^{-s_t}$ in the theorem is linked to the decay structure of the coefficients, which is verified in the proof of Proposition \ref{prop:decayK0} for a class of panel AR(1) models.

\begin{proof}[Proof of Proposition \ref{prop:decayK0}]
Consider the panel AR(1) model:
$$
Y_{it}=\alpha_i+\theta_1^0 Y_{i,t-1}+\theta_2^0 D_{it}+\varepsilon_{it}, \quad |\theta_1^0|<1.
$$
Iterating the model forward yields
$$Y_{it}=\frac{\alpha_i}{1-\theta_1^0} + \sum_{\ell\geq0}(\theta_1^0)^{\ell}(\theta_2^0 D_{i,t-\ell}+\varepsilon_{i,t-\ell}).$$
Recall that $\Pi_t^{0}=\bar\E(V_{it}V_{it}^{\top})^{-1}\bar\E(V_{it}W_{it})$, where $W_{it}\in\{\Delta Y_{i,t-1},\Delta D_{it}\}$ and $V_{it} = (1,X_{i1}^\top,\ldots,X_{it}^\top)^\top=(1,Y_{i0},D_{i1},\ldots,Y_{i,t-1},D_{it})^\top$. We focus on the component $W_{it}=\Delta Y_{i,t-1}$.

Without loss of generality, we set the mean of $D_{it}$ to zero, since a nonzero mean does not affect the decay structure. Similarly, as the individual effects $\alpha_i$ are treated as fixed parameters in the theoretical analysis, we set $\alpha_i=0$ in this proof. Any nonzero shift would be absorbed into the intercept term of the projection.
Let $\operatorname{Var}(D_{it})=\sigma_D^2$, $\operatorname{Var}(\vps_{it})=\sigma_\vps^2$, and $\E(D_{it}\vps_{i,t-1})=\kappa$. It follows that
\begin{align*}
    \E(Y_{it}D_{is})&=(\theta_1^0)^{t-s}\theta_2^0\sigma^2_D\mathbf 1(t>s-1)+\kappa(\theta_1^0)^{t-s+1}\mathbf 1(t\geq s-1),\quad \E(D_{it}D_{is})=\sigma_D^2\mathbf 1(t=s),\\
    \E(Y_{it}Y_{is})&=\begin{cases}
			\frac{(\theta_2^0)^2\sigma_D^2+\sigma_\vps^2+2\theta_1^0\theta_2^0\kappa}{1-(\theta_1^0)^2}, & \text{if } t=s,\\
            \frac{\{(\theta_2^0)^2\sigma_D^2+\sigma_\vps^2\}(\theta_1^0)^{|t-s|}+\theta_2^0\kappa\{(\theta_1^0)^{|t-s|-1}+(\theta_1^0)^{|t-s|+1}\}}{1-(\theta_1^0)^2}, & \text{if } t\neq s.
		 \end{cases}
\end{align*}
We observe that the covariances among the instruments decay with the time distance. The matrix $\E(V_{it}V_{it}^{\top})$ belongs to the class of well-localized matrices with exponentially decaying off-diagonal entries. Standard results for such matrices (see e.g.,\cite{grochenig2006symmetry}) imply that the element in the $l$-th row and $r$-th column of the inverse matrix satisfies
$$\big|\big(\E(V_{it}V_{it}^{\top})^{-1}\big)_{lr}\big|\lesssim \rho^{|l-r|},\quad l,r=1,\ldots,2t+1,$$
for some $\rho\in(0,1)$.

On the other hand, based on the recursive representation of the model and the definition of the FOD transformation, we find that
\begin{align*}
    \E(\Delta Y_{i,t-1}D_{is})&=\begin{cases}
			c_{t-1}\Big[(\theta_1^0)^{t-s}\theta_2^0\sigma_D^2 + (\theta_1^0)^{t-s+1}\kappa - \frac{(\theta_2^0\sigma_D^2+\theta_1^0\kappa)(\theta_1^0)^{t-s}\{1-(\theta_1^0)^{T-t+1}\}}{(T-1+1)(1-\theta_1^0)}\Big], \, \text{if } 1\leq s<t,\\
            c_{t-1}\Big[\kappa - \frac{(\theta_2^0\sigma_D^2+\theta_1^0\kappa)\{1-(\theta_1^0)^{T-t+1}\}}{(T-1+1)(1-\theta_1^0)}\Big], \, \text{if } s=t,
		 \end{cases}\\
    \E(\Delta Y_{i,t-1}Y_{is})&=\begin{cases}
			c_{t-1}\Big[\frac{(\theta_2^0)^2\sigma_D^2+\sigma_\vps^2+2\theta_1^0\theta_2^0\kappa}{1-(\theta_1^0)^2} - \frac{\{1-(\theta_1^0)^{T-t+1}\}\{\theta_1^0(\theta_2^0)^2\sigma_D^2+\theta_2^0\kappa+(\theta_1^0)^2\theta_2^0\kappa+\theta_1^0\sigma_\vps^2\}}{(T-1+1)(1-\theta_1^0)\{1-(\theta_1^0)^2\}}\Big],  \,\text{if } s=t-1,\\
            c_{t-1}\Big[\frac{\{(\theta_2^0)^2\sigma_D^2+\sigma_\vps^2\}(\theta_1^0)^{t-s-1}+\theta_2^0\kappa(\theta_1^0)^{t-s-2}\{1+(\theta_1^0)^{2}\}}{1-(\theta_1^0)^2} &\\
            \hspace{1cm} - \frac{\{1-(\theta_1^0)^{T-t+1}\}(\theta_1^0)^{t-s-1}\{\theta_1^0(\theta_2^0)^2\sigma_D^2+\theta_2^0\kappa+(\theta_1^0)^2\theta_2^0\kappa+\theta_1^0\sigma_\vps^2\}}{(T-1+1)(1-\theta_1^0)\{1-(\theta_1^0)^2\}}\Big],  \,\text{if } 0\leq s<t-1,
		 \end{cases}
\end{align*}
where $c_{t-1}=\sqrt{(T-t+1)/(T-t+2)}$. Both $\E(\Delta Y_{i,t-1}D_{is})$ and $\E(\Delta Y_{i,t-1}Y_{is})$ exhibit geometric decay in the time distance $t-s$, at a rate determined by $\theta_1^0$.

Combining these findings, we conclude that the coefficients in $\Pi_t^0$ also inherit the decay structure, and hence $\Pi_t^0$ is approximately sparse.

\end{proof}

\subsubsection{The RE Condition}
%\begin{equation}
%\min_{\delta_{\Pi,t}\neq 0, |\delta_{\Pi,t}|_0\leq s_t,|\delta_{\Pi,J_t^c}|_1\leq c_0 |\delta_{\Pi,J_t}|_1 } \frac{|X_t^{\top}{\delta}_{\Pi}|_{2,n}^2}{|\delta_{\Pi, J_t}|_2^2}\geq \kappa_t^2(c_0,s_t)
%\end{equation}

%\begin{lemma}
% Define $\Omega \defeq \{\delta_{\Pi,t}/|\delta_{\Pi,t}|_2: \delta_{\Pi,t}\neq 0, |\delta_{\Pi,t}|_0\leq s_t,|\delta_{\Pi,J_t^c}|_1\leq c_0 |\delta_{\Pi,J_t}|_1\}$.
%Assume that
%$c_{\min}\leq \min_{\delta\in \Omega} n^{-1}\sum_t{\delta^{\top}\E(X_tX_t^{\top})}\delta \leq C_{\max}$.
%Also $X_{it}$ are iid over $i$ and have finite $q$ the moment with $q>4$.
%Assume $\delta_{n,t}=c_{\delta}(\sqrt{n_t}^{-1} \sqrt{s_t \log p} \vee n^{1/q-1} (s_t \log p)^{3/2} l_s^{1/q} \sqrt{s_t})\to 0$.
%$C_{\max} = \kappa_t^2(c_0,s_t)+ \delta_{n,t}$.
%Then we have the event
%\begin{equation}
%\min_{\delta_{\Pi,t}\neq 0, |\delta_{\Pi,t}|_0\leq s_t,|\delta_{\Pi,J_t^c}|_1\leq c_0 |\delta_{\Pi,J_t}|_1 } \frac{|X_t^{\top}{\delta}_{\Pi}|_{2,n}^2}{|\delta_{\Pi, J_t}|_2^2}\geq \kappa_t^2(c_0,s_t).
%\end{equation}
%happens with high probability.
%\end{lemma}
\begin{proof}[Proof of Lemma \ref{idlemma}]
Recall that $\delta_{t,J_t}$ is a sub-vector of $\delta_t$ corresponding to the indices set $J_t$. It suffices to show that the event
$$\min_{\delta_{t}\neq 0, |\delta_{t}|_0\leq s_t^*,|\delta_{t,J_t^c}|_1\leq c_0 |\delta_{t,J_t}|_1 } \frac{|\bm V_t{\delta}_{t}|_{2}}{\sqrt{N}|\delta_{t}|_2}\geq \kappa_t(c_0,s_t^*)$$
occurs with probability $1-\smallO(1)$ as it is less likely to hold than the original event required for identification.

Consider the sphere $\Omega_t(c_0,s_t^*)$ defined in Assumption \ref{33}. To simplify notation, we will henceforth refer to it as $\Omega_t$. Let $\Omega^\epsilon_{t}$ be the $\epsilon$-net of $\Omega_t$ with respect to $|\cdot|_2$. According to \citet{rudelson2012reconstruction}, the cardinality of $\Omega^\epsilon_{t}$ is bounded as $|\Omega^\epsilon_{t}|\lesssim {m_t \choose s_t^*}(c/\epsilon)^{s_t^*}\leq\{cem_t/(s_t^*\epsilon)\}^{s_t^*}$, where $m_t$ is the dimension of $V_{it}$ (i.e. the number of instruments for each period $t$), %(i.e. $m_t\asymp t$)
and $c>0$ is an absolute constant.
Moreover, for any point $\delta\in\Omega_t$, let $\pi(\delta)$ denote the closest point to $\delta$ within $\Omega^\epsilon_{t}$.

We first observe that
\begin{align*}
\min_{\delta_t\in\Omega_t}|\bm V_t\delta_t|_2^2&\geq\min_{\delta_t\in\Omega_t}|N\delta_t^\top\bar\E(V_{it}V_{it}^\top)\delta_t|-\max_{\delta_t\in\Omega_t}|\delta_t^\top\{\bm V_t^\top \bm V_t-N\bar\E(V_{it}V_{it}^\top)\}\delta_t|\\
&\geq\min_{\delta_t\in\Omega_t}|N\delta_t^\top\bar\E(V_{it}V_{it}^\top)\delta_t|-\max_{\delta_t\in\Omega_t}|\bm V_t\delta_t|_2^2-\max_{\delta_t\in\Omega_t}|N\delta_t^\top\bar\E(V_{it}V_{it}^\top)\delta_t|.
\end{align*}
Next, we shall derive the upper bound for the second term. By defining $C_t\defeq\max\limits_{\delta_t\in\Omega_t}|\bm V_t(\delta_t-\pi(\delta_t))|_2$, and $D_t\defeq\max\limits_{\pi(\delta_t)\in\Omega_t^\epsilon}|\bm V_t\pi(\delta_t)|_2$, we have the following inequalities
$$D_t-C_t\leq\max_{\delta_t\in\Omega_t}|\bm V_t\delta_t|_2\leq D_t+C_t.$$
Since $C_t\leq \epsilon\max\limits_{\delta_t\in\Omega_t}|\bm V_t\delta_t|_2$, we can further obtain that
$$D_t/(1+\epsilon)\leq\max_{\delta_t\in\Omega_t}|\bm V_t\delta_t|_2\leq D_t/(1-\epsilon).$$
To bound $D_t$, we apply the tail probability inequality in Lemma \ref{emp} on the mean zero random variable $z_{it}(\delta_t)\defeq\{\pi(\delta_t)^\top V_{it}\}^2 -\bar\E\{\pi(\delta_t)^\top V_{it}\}^2$, over all $\pi(\delta_t)\in\Omega_t^\epsilon$. It follows that $D_t^2\lesssim_\P\sqrt{N}\{s_t^*\log(cem_t/(s_t^*\epsilon))\}^{1/2}$.
%$D_t^2\lesssim_\P\sqrt{N}\{s_t^*\log(cem_t/(s_t^*\epsilon))\}^{1/\gamma}\Psi_{\psi_\nu,0}$, where $\gamma = 2/(2\nu+1)$ with $\nu\geq0$ such that $\Psi_{\psi_\nu,0}\defeq\max\limits_{\pi(\delta_t)\in\Omega_t^\epsilon}\sup\limits_{r\geq2}r^{-\nu}\|z_{\cdot}(\delta_t)\|_{r,0}<\infty$.
% \begin{align*}
%     \P\Big(\max_{\pi(\delta_t)\in\Omega_t^\epsilon}|\pi(\delta_t)^\top B_t\pi(\delta_t)|\geq z\Big)=\P\Big(\max_{\pi(\delta_t)\in\Omega_t^\epsilon}\Big|\sum_{i=1}^N\big[\{\pi(\delta_t)^\top V_{it}\}^2 -\E\{\pi(\delta_t)^\top V_{it}\}^2\big]\Big|\geq z\Big)
% \end{align*}

%By letting {\red$\epsilon= $ (need to cancel the order of $D_t^2/(1-\epsilon)^2$ with $\sqrt{Ns_t^*}\log(N\vee T)$)}
By choosing a sufficiently small $\epsilon$, e.g., $\epsilon=1/(\log s_t^*)^{1/4}$, and given Assumption \ref{33} with the particular $\kappa_t(\cdot)$ specified in the conditions of the lemma, we have $$\min_{\delta_t\in\Omega_t}N^{-1/2}|\bm V_t\delta_t|_2=\min_{\delta_{t}\neq 0, |\delta_{t}|_0\leq s_t^*,|\delta_{t,J_t^c}|_1\leq c_0 |\delta_{t,J_t}|_1 } \frac{|\bm V_t{\delta}_{t}|_{2}}{\sqrt{N}|\delta_{t}|_2}\geq \kappa_t(c_0,s_t^*),$$
which holds with probability $1-\smallO(1)$, as $N\to\infty$.
\end{proof}

\subsubsection{Prediction Performance of $\widehat{\Pi}_t$}

\begin{proof}[Proof of Theorem \ref{lassobound}]
First observe that
\begin{align*}
    |\bm W_t - \bm V_t\widehat \Pi_t|_2^2 &= |\bm V_t(\Pi_t^0-\widehat\Pi_t)+\bm\eta_t|_2^2= N|\Pi_t^0-\widehat\Pi_t|_{2,N}^2 + |\bm\eta_t|_2^2 + 2\langle\bm V_t(\Pi_t^0-\widehat\Pi_t),\bm\eta_t\rangle,\\
    |\bm W_t - \bm V_t\Pi_t^*|_2^2 &= |\bm V_t(\Pi_t^0-\Pi_t^*)+\bm\eta_t|_2^2= N|\Pi_t^0-\Pi_t^*|_{2,N}^2 + |\bm\eta_t|_2^2 + 2\langle\bm V_t(\Pi_t^0-\Pi_t^*),\bm\eta_t\rangle.
\end{align*}
By the definition of the LASSO estimator, we have
$$|\bm W_t - \bm V_t\widehat \Pi_t|_2^2 + \lambda_t|\bm\omega_t\circ\widehat\Pi_t|_1 \leq |\bm W_t - \bm V_t\Pi_t^*|_2^2 + \lambda_t|\bm\omega_t\circ \Pi_t^*|_1,$$
where $\circ$ denotes the Hadamard product.
It follows that
\begin{align*}
    |\delta_{\Pi,t}|_{2,N}^2 -2C_{s_t^*}|\delta_{\Pi,t}|_{2,N}&\leq |\Pi_t^0-\widehat\Pi_t|_{2,N}^2 - |\Pi_t^0-\Pi_t^*|_{2,N}^2\\
    &= N^{-1}|\bm W_t - \bm V_t\widehat \Pi_t|_2^2 - N^{-1}|\bm W_t - \bm V_t\Pi_t^*|_2^2 + 2N^{-1}\langle\bm V_t\delta_{\Pi,t},\bm\eta_t\rangle\\
    &\leq 2N^{-1}\langle\bm V_t\delta_{\Pi,t},\bm\eta_t\rangle + N^{-1}\lambda_t(|\bm\omega_t\circ \Pi_t^*|_1 - |\bm\omega_t\circ\widehat\Pi_t|_1).
\end{align*}
If $|\delta_{\Pi,t}|_{2,N}^2 -2C_{s_t^*}|\delta_{\Pi,t}|_{2,N}\leq0$, then we have $ |\delta_{\Pi,t}|_{2,N}\leq 2C_{s_t^*}$ and the desired bound in the conclusion holds. In the following, we shall derive the bound for the case of $|\delta_{\Pi,t}|_{2,N}^2 -2C_{s_t^*}|\delta_{\Pi,t}|_{2,N}>0$.

On the event $\mathcal A_{2t}$, %(with $c=4$),
we have
$$\langle\bm V_t\delta_{\Pi,t},\bm\eta_t\rangle=\langle\bm\eta_t^\top\bm V_t\oslash\bm\omega_t,\bm\omega_t\circ\delta_{\Pi,t}\rangle\leq\lambda_t|\bm\omega_t\circ\delta_{\Pi,t}|_1/c,$$
which implies that
\begin{eqnarray}\label{66}
    &&|\delta_{\Pi,t}|_{2,N}^2 -2C_{s_t^*}|\delta_{\Pi,t}|_{2,N} + (2N)^{-1}\lambda_t|\bm\omega_t\circ\delta_{\Pi,t}|_1\notag\\
    &\leq& \frac{4+c}{2c}N^{-1}\lambda_t|\bm\omega_t\circ\delta_{\Pi,t}|_1 + N^{-1}\lambda_t(|\bm\omega_t\circ \Pi_t^*|_1 - |\bm\omega_t\circ\widehat\Pi_t|_1).
\end{eqnarray}
Recall that $J_t$ is the indices set for nonzero elements in $\Pi_t^*$, and $J_t^c$ is the complement set for zero ones. Let $\delta_{\Pi,t,J_t}$ and $\delta_{\Pi,t,J_t^c}$ denote the sub-vectors of $\delta_{\Pi,t}$ corresponding to $J_t$ and $J_t^c$, similarly for $\Pi_t^*$, $\widehat\Pi_t$, and $\bm\omega_t$. Then, by \eqref{66} we obtain that
\begin{eqnarray*}
   &&|\delta_{\Pi,t}|_{2,N}^2 -2C_{s_t^*}|\delta_{\Pi,t}|_{2,N} + (2N)^{-1}\lambda_t|\bm\omega_t\circ\delta_{\Pi,t}|_1\\
   &\leq& \frac{4+c}{2c}N^{-1}\lambda_t|\bm\omega_{t,J_t}\circ\delta_{\Pi,t,J_t}|_1 + N^{-1}\lambda_t|\bm\omega_{t,J_t}\circ\Pi^*_{t,J_t}|_1 - N^{-1}\lambda_t|\bm\omega_{t,J_t}\circ\widehat\Pi_{t,J_t}|_1\\
   &\leq& C_1N^{-1}\lambda_t|\bm\omega_{t,J_t}\circ\delta_{\Pi,t,J_t}|_1,
\end{eqnarray*}
where $C_1>0$ is a constant depending only on the $c$ that satisfies Assumption \ref{weights}.

Without loss of generality, we set the penalty wights to be 1. Given that $|\delta_{\Pi,t}|_{2,N}^2 -2C_{s_t^*}|\delta_{\Pi,t}|_{2,N}\geq0$, \eqref{66} also implies that $\delta_{\Pi,t}$ satisfies $|\delta_{\Pi,t,J_t^c}|_1<6|\delta_{\Pi,t,J_t}|_1$ for $c>1$. On the event $\mathcal A_{1t}$, we get $|\delta_{\Pi,t}|_{2,N}\geq\kappa_t(3,s_t^*)|\delta_{\Pi,t,J_t}|_2$. Therefore, on the events $\mathcal A_{1t}$ and $\mathcal A_{2t}$, we have
$$|\delta_{\Pi,t}|_{2,N}^2 -2C_{s_t^*}|\delta_{\Pi,t}|_{2,N}\leq C_1N^{-1}\sqrt{s_t^*}\lambda_t|\delta_{\Pi,t,J_t}|_2\leq C_1N^{-1}\sqrt{s_t^*}\lambda_t|\delta_{\Pi,t}|_{2,N}/\kappa(3,s_t^*),$$
which gives the bound for prediction norm:
$$|\delta_{\Pi,t}|_{2,N}\lesssim 2C_{s_t^*}+N^{-1}\sqrt{s_t^*}\lambda_t/\kappa(3,s_t^*),$$
with probability at least $1-\alpha-\smallO(1)$.

Next, we shall derive the bound for $|\delta_{\Pi,t}|_1$. When %$|\delta_{\Pi,t}|_{2,N}^2 -2C_{s_t^*}|\delta_{\Pi,t}|_{2,N}\geq0$, we have
$|\delta_{\Pi,t,J_t^c}|_1<6|\delta_{\Pi,t,J_t}|_1$ is satisfied, we have
%which gives
$$|\delta_{\Pi,t}|_1<7|\delta_{\Pi,t,J_t}|_1\leq7\sqrt{s_t^*}|\delta_{\Pi,t,J_t}|_2\leq7\sqrt{s_t^*}|\delta_{\Pi,t}|_{2,N}/\kappa(3,s_t^*).$$
When $|\delta_{\Pi,t,J_t^c}|_1\geq6|\delta_{\Pi,t,J_t}|_1$, we have $|\delta_{\Pi,t}|_{2,N}^2 -2C_{s_t^*}|\delta_{\Pi,t}|_{2,N}<0$. Then, by \eqref{66} we can find that $C_{s_t^*}^2\geq C_2N^{-1}\lambda_t|\delta_{\Pi,t,J_t^c}|_1$, where $C_2>0$ is a constant only depending on $c$. Hence,
$$|\delta_{\Pi,t}|_1\leq\frac{7}{6}|\delta_{\Pi,t,J_t^c}|_1\lesssim NC_{s_t^*}^2/\lambda_t,$$
with probability at least $1-\alpha$.
Overall, $|\delta_{\Pi,t}|_1$ is bounded as
$$|\delta_{\Pi,t}|_1\lesssim 7\sqrt{s_t^*} \{2C_{s_t^*}+  N^{-1}\sqrt{s_t^*}\lambda_t/\kappa_t(3,s_t^*) \}/\kappa_t(3,s_t^*) + NC_{s_t^*}^2/\lambda_t,$$
with probability at least $1-\alpha-\smallO(1)$.
\end{proof}

\subsection{Proofs for Section \ref{secondstep}}
In this subsection, we first present the proofs for Section \ref{secondstep} for the model without time effects $\gamma_t$ (in Appendices \ref{app:var}--\ref{app:hdm}). The cross-sectional demeaning required when time effects are present introduces weak cross-sectional dependence, which makes the inference analysis more cumbersome but does not materially affect the results. %(under suitable conditions). 
In Appendix \ref{add}, we indicate how the proofs can be adapted to accommodate $\gamma_t$. %in an example with %a dynamic panel model with one lagged dependent variable as a covariate, that is, 

\subsubsection{Asymptotic Normality of AB-LASSO and AB-LASSO-SS}\label{app:main}
\begin{proof}[Proof of Theorem \ref{main}]
\underline{For AB-LASSO}: Recall the expression
$$\widehat\theta - \theta^0 = \bigg(\frac{1}{NT}\sum_{i=1}^N \sum_{t=1}^{T-1} \widehat{\bm\Theta}_tV_{it}\Delta X_{it}^{\top} \bigg)^{-1} \bigg(\frac{1}{NT}\sum_{i=1}^N \sum_{t=1}^{T-1} \widehat{\bm\Theta}_tV_{it} \Delta\varepsilon_{it}\bigg),$$
and rewrite the first summation as
\begin{eqnarray*}
&&(NT)^{-1}\sum_{i=1}^N \sum_{t=1}^{T-1} \widehat{\bm\Theta}_tV_{it}\Delta X_{it}^{\top}\\
&=&(NT)^{-1}\sum_{i=1}^N \sum_{t=1}^{T-1} (\widehat{\bm\Theta}_t - \bm\Theta_t^*)V_{it}\Delta X_{it}^{\top} + (NT)^{-1}\sum_{i=1}^N \sum_{t=1}^{T-1} \bm\Theta_t^*\{V_{it}\Delta X_{it}^{\top}-\E(V_{it}\Delta X_{it}^{\top})\} \\
&&+ \,(NT)^{-1}\sum_{i=1}^N \sum_{t=1}^{T-1} \bm\Theta_t^*\E(V_{it}\Delta X_{it}^{\top})\\
&=:&I_1 + I_2 + I_3.
\end{eqnarray*}
To deal with the inverse of the sum, we consider the following expansion:
\begin{align*}
(I_1+I_2+I_3)^{-1}&=[I_3\{\mathbf{I}_{d\times d}+I_3^{-1}(I_1+I_2)\}]^{-1}\\
&=\bigg[\sum_{k=0}^\infty\big\{-I_3^{-1}(I_1+I_2)\big\}^k\bigg]I_3^{-1}=\{\mathbf I_{d\times d}-I_3^{-1}(I_1+I_2)+R_n\}I_3^{-1},
\end{align*}
where $\mathbf I_{d\times d}$ represents the $d\times d$ identity matrix and $R_n$ denotes the remainder of order $\smallO_\P((NT)^{-1/2})$. We observe an approximate sparse error between $\bm\Theta^0_t$ and $\bm\Theta^*_t$, with the average error rate being of a small order $T^{-1}\sum_{t=1}^{T-1}|\bm\Theta^*_t-\bm\Theta^0_t|_\infty=\smallO(1)$. This generally holds true under a decaying temporal dependence structure. The error is deemed negligible in the subsequent asymptotic analysis. Based on Assumption \ref{identification}, we assert that $I_3$ is invertible in the limit, with the presence of only a negligible error.

The second summation in the expression of $(\widehat\theta - \theta^0)$ is rewritten as
\begin{eqnarray*}
&&(NT)^{-1}\sum_{i=1}^N \sum_{t=1}^{T-1} \widehat{\bm\Theta}_tV_{it} \Delta\varepsilon_{it}\\
&=&(NT)^{-1}\sum_{i=1}^N \sum_{t=1}^{T-1} (\widehat{\bm\Theta}_t-\bm\Theta_t^*)V_{it} \Delta\varepsilon_{it}+(NT)^{-1}\sum_{i=1}^N \sum_{t=1}^{T-1} \bm\Theta_t^*V_{it} \Delta\varepsilon_{it}=:P_1 + P_2
\end{eqnarray*}
It follows that
\beq\label{expansion}
\widehat\theta - \theta^0=I_3^{-1}P_2 +I_3^{-1}P_1- I_3^{-1}(I_1+I_2)I_3^{-1}(P_1+P_2)+\smallO_\P((NT)^{-1/2}).
\eeq
We will prove that $I_3^{-1}P_2$ is the leading term, with $\sqrt{NT}I_3^{-1}P_2$ exhibiting asymptotic Gaussianity and analyze the orders of the other terms.
%as proved in Lemma \ref{Central_limit_theorem}.
%in Lemma \ref{22} and \ref{11}, respectively.

% Then with the nominator, we can see that
% \begin{align*}
% (NT)^{-1}\sum_{i=1}^{N}\sum_{t=3}^{T}\widehat{\Theta}_{t}V_{it}\Delta\tilde{\varepsilon}_{it} & =(NT)^{-1}\sum_{i=1}^{N}\sum_{t=3}^{T}(\widehat{\Theta}_{t}-\Theta_{t}^{*})V_{it}\Delta\tilde{\varepsilon}_{it}+(NT)^{-1}\sum_{i=1}^{N}\sum_{t=3}^{T}\Theta_{t}^{*}V_{it}\Delta\tilde{\varepsilon}_{it}\\
%  & =P_{1}+P_{2}.
% \end{align*}\\
% {By Lemma \ref{11} and the rate therein.}
% To analyze the rate, we see the decomposition of the terms as follows,
% \begin{align*}
% \hat{\theta}-\theta & =[((I-I_{3}^{-1}(I_{1}+I_{2})+I_{3}^{-1}(I_{1}+I_{2})I_{3}^{-1}(I_{1}+I_{2}))+...)I_{3}^{-1}](P_{1}+P_{2})\\
%  & =\underbrace{I_{3}^{-1}P_{2}}_{(NT)^{-1/2}}-I_{3}^{-1}(I_{1}+I_{2})I_{3}^{-1}(P_{1}+P_{2})+I_{3}^{-1}P_{1}+ \Co_p(\sqrt{NT}^{-1}).
% \end{align*}

We refer to a central limit theorem for stationary random field (Theorem 1 of \citet{el2013central}) to establish asymptotic normality. To achieve this, we must verify the necessary conditions outlined below.

Define an index set $\mathcal J_{N,T}\defeq\{(i,t):1\leq i\leq N,1\leq t\leq T-1\}$. As $N,T\to\infty$, it follows that the cardinality $|\mathcal J_{N,T}|\to\infty$, while the ratio $|\partial\mathcal J_{N,T}|/|\mathcal J_{N,T}|\to0$, where $\partial \mathcal J_{N,T}$ contains the boundary points of $\mathcal J_{N,T}$. Under Assumptions \ref{a1} and \ref{a2}(ii), the $d$-dimensional process $Z_{(i,t)}\defeq\bm\Theta^*_tV_{it}\Delta\varepsilon_{it}$ is i.i.d.\ across $i$ and exhibits weak serial dependence over $t$, as $\E(\Delta\varepsilon_{it}\Delta\varepsilon_{i,t-\ell}\mid V_{it})\lesssim\frac{1}{T-t+\ell}$ for any $\ell\geq 1$. For $k=1,\ldots,d$, $Z_{(i,t),k}$ can be represented as $Z_{(i,t),k}=h_{(i,t),k}(\ldots,\eta_{(i,t-1)},\eta_{(i,t)})$, where $h_{(i,t),k}$ are measurable functions, and $\eta_{(i,t)}$ for $i\in\mathbb N$, $t\in\mathbb Z$, are i.i.d.\ random elements. By Definition \ref{dep} and Assumption \ref{a2}(i), it follows that
$$\sum_{(i,t)\in\mathcal J_{N,T}}\|Z_{(i,t),k}^*-Z_{(i,t),k}\|_2<\infty.$$

Moreover, Assumption \ref{a2}(ii), along with the cross-sectional independence assumption, implies that for $k,k'=1,\ldots,d$, the variance
\begin{align*}
\E\bigg[\bigg(\sum_{(i,t)\in\mathcal J_{N,T}}Z_{(i,t),k)}\bigg)\bigg(\sum_{(i,t)\in\mathcal J_{N,T}}Z_{(i,t),k'}\bigg)\bigg]=\sum_{(i,t)\in\mathcal J_{N,T}}\E(Z_{(i,t),k}Z_{(i,t),k'}) + \smallO(NT)
%\quad+\sum_{(i,t)\in\mathcal J'_{N,T}}\E(Z_{(i,t),k}Z_{(i,t-1),k'})+\sum_{(i,t)\in\mathcal J'_{N,T}}\E(Z_{(i,t-1),k}Z_{(i,t),k'})
\end{align*}
is of order $NT$. %, where $\mathcal J'_{N,T}\defeq\{(i,t):1\leq i\leq N,3\leq t\leq T\}$.
%{\red Here $\E$ denotes unconditional expectation. }
Therefore, based on Assumption \ref{identification}, by applying Theorem 1 of \citet{el2013central} and Slutsky's theorem, we deduce that
$$\sqrt{NT}I_3^{-1}P_2\stackrel{\mathcal{L}}{\to} \operatorname{N}(0, Q^{-1}\bm\Sigma(Q^{-1})^\top).$$

% {\color{red}
% Recalling the subspace $\Omega_t(c_0,s_t^*)$ defined in Assumption \ref{33}, and considering the
% % function class of the stacked vector $(\delta_t)_{t=2}^T$ given by
% % $$\mathcal H\defeq\{(\delta_t)_{t=2}^T\in\R^{m}:\delta_t\in\Omega_t(c_0,s_t^*),t=2,\ldots,T\},$$
% % where $m=\sum_{t=2}^Tm_t$ represents the total number of instruments for all $t=2,\ldots,T$.
% entropy condition (with respect to the $d_2$-metric):
% $$\operatorname{ent}\big(\epsilon,\bigcup\nolimits_{2\leq t\leq T}\Omega_t(c_0,s_t^*)\big)\lesssim T\max_{2\leq t\leq T}s_t^*\log(m_t/\epsilon), \,\text{ for all } 0<\epsilon<1,$$
% by employing Theorem \ref{emplemma} based on Assumptions \ref{a1}--\ref{a2}, we obtain:
% \begin{equation}\label{empbound}
% \sup_{\delta_t\in\Omega_t(c_0,s_t^*),\,t=2,\ldots,T}%{(\delta_t)_{t=2}^T\in\mathcal H}
% \bigg|(NT)^{-1}\sum_{i=1}^N\sum_{t=2}^T\delta_t^\top V_{it}\Delta\varepsilon_{it}\bigg|\lesssim_\P\sqrt{\max_{2\leq t\leq T}s_t^*\log m_t}/\sqrt{N}.   
% \end{equation}
% Consequently, we bound $|I_3^{-1}P_1|_\infty$ as:
% $$|I_3^{-1}P_1|_\infty \leq |I_3^{-1}|_\infty\bigg|(NT)^{-1}\sum_{i=1}^N\sum_{t=2}^T(\widehat{\bm\Theta}_t-\bm\Theta_t^*)V_{it} \Delta\varepsilon_{it}\bigg|_\infty\lesssim_\P\frac{\max_{2\leq t\leq T}s_t^*\log m_t}{N}.$$}

Next, we bound $|I_3^{-1}P_1|_\infty$. For $j=1,\ldots,d$, let $e_j$ denote the $j$-th canonical basis vector, so that the $j$-th component of $P_1$ can be written as $P_{1,j}=(NT)^{-1}\sum_{i=1}^{N}\sum_{t=1}^{T-1}
e_j^\top(\widehat{\bm{\Theta}}_{t}-\bm{\Theta}_{t}^{*})V_{it}\Delta\varepsilon_{it}$. To apply Markov's inequality, we study the order of $\E(P_{1,j}^2)$, where 
$$\E(P_{1,j}^2)=(NT)^{-2}\E\bigg(\sum_{i=1}^{N}\sum_{t=1}^{T-1}
a_{it,j}\Delta\varepsilon_{it}\bigg)^2, \quad a_{it,j}\defeq e_j^\top(\widehat{\bm{\Theta}}_{t}-\bm{\Theta}_{t}^{*})V_{it}.$$

We first consider the squared terms in the expansion. Recalling the filtration $\mathcal F_{it}$ in Assumption \ref{a2}(ii), we have 
$$\E(a_{it,j}^2(\Delta\varepsilon_{it})^2\mid \mathcal F_{it}, \widehat{\bm\Theta}_t)=a_{it,j}^2\E((\Delta\varepsilon_{it})^2\mid \mathcal F_{it}, \widehat{\bm\Theta}_t)\leq \sigma^2a_{it,j}^2,$$
where $\sigma^2=\max_{i,t}\E((\Delta\varepsilon_{it})^2\mid \mathcal F_{it}, \widehat{\bm\Theta}_t)<\infty$. Applying the prediction performance bound in Theorem \ref{lassobound} then yields
$$\sum_{i=1}^{N}\sum_{t=1}^{T-1}\E(a_{it,j}^2(\Delta\varepsilon_{it})^2)\lesssim T\max_{1\leq t\leq T-1}s_t^*(\log m_t)^2.$$

Concerning the cross terms, note that $\E(\Delta\varepsilon_{it}\Delta\varepsilon_{it'})=0$ for $1\leq t\neq t'\leq T-1$, and $\lim\limits_{N,T\to\infty}(NT)^{-1}\sum\limits_{i=1}^N\sum\limits_{1\leq s<t\leq T-1}\E(\Delta\varepsilon_{it}\Delta\varepsilon_{is} \mid V_{it})=0$. Therefore, $\Delta\varepsilon_{it}$ can be treated as approximately a martingale difference sequence over $t$. Moreover, $\Delta Y_{i,t-1}$ and $\Delta\varepsilon_{it}$ are correlated through the future average component, and this dependence is of order $1/(T-t)$ due to the harmonic weights in the FOD transformation. Since the generated errors $(\widehat{\bm{\Theta}}_{t}-\bm{\Theta}_{t}^{*})$ involve cross-sectional averages, cross-sectional independence together with Assumption \ref{a2}(iv) implies that 
$$\sum_{i=1}^{N}\sum_{1\leq t\neq t'\leq T-1}\E(a_{it,j}a_{it',j}\Delta\varepsilon_{it}\Delta\varepsilon_{it'})\lesssim\frac{T}{N},\quad \sum_{1\leq i\neq i'\leq N}\sum_{t=1}^{T-1}\E(a_{it,j}a_{i't,j}\Delta\varepsilon_{it}\Delta\varepsilon_{i't})=\bigO(1),$$
$$\sum_{1\leq i\neq i'\leq N}\sum_{1\leq t\neq t'\leq T-1}\E(a_{it,j}a_{i't',j}\Delta\varepsilon_{it}\Delta\varepsilon_{i't'})\lesssim (\log T)^2.$$

As a result, we obtain $\E(P_{1,j}^2)\lesssim \max\limits_{1\leq t\leq T-1}s_t^*(\log m_t)^2/(N^2T)$, and therefore $|I_3^{-1}P_1|_\infty\lesssim_\P\max\limits_{1\leq t\leq T-1}\sqrt{s_t^*}\log m_t/(N\sqrt{T})$. In particular, under the condition that $\max\limits_{1\leq t\leq T-1}\sqrt{s_t^*}\log m_t/\sqrt{N}\to0$ as $N,T\to\infty$, we conclude that $|I_3^{-1}P_1|_\infty=\smallO_\P(1/\sqrt{NT})$.

Next, we proceed to bound $|I_3^{-1}(I_1+I_2)I_3^{-1}(P_1+P_2)|_\infty$. Observe that
\begin{eqnarray}\label{expand}
|I_3^{-1}(I_1+I_2)I_3^{-1}(P_1+P_2)|_\infty&\leq &|I_3^{-1}I_1I_3^{-1}P_1|_\infty + |I_3^{-1}I_2I_3^{-1}P_1|_\infty \notag\\
&&+\, |I_3^{-1}I_1I_3^{-1}P_2|_\infty + |I_3^{-1}I_2I_3^{-1}P_2|_\infty.
\end{eqnarray}

We first examine the rate of $|I_3^{-1}I_1|_\infty$. Write $\Delta X_{it}=\mu_{it} + u_{it}$, where $\mu_{it}=\E(\Delta X_{it}\mid \mathcal F_{it})$ and $u_{it}=\Delta X_{it} - \E(\Delta X_{it}\mid \mathcal F_{it})$. For $j=1,\ldots,d$, the $j$-th component of $I_1$ can be written as $I_{1,j}=(NT)^{-1}\sum_{i=1}^{N}\sum_{t=1}^{T-1}
e_j^\top(\widehat{\bm{\Theta}}_{t}-\bm{\Theta}_{t}^{*})V_{it}\Delta X_{it}^\top=(NT)^{-1}\sum_{i=1}^{N}\sum_{t=1}^{T-1}
e_j^\top(\widehat{\bm{\Theta}}_{t}-\bm{\Theta}_{t}^{*})V_{it}\mu_{it}^\top + (NT)^{-1}\sum_{i=1}^{N}\sum_{t=1}^{T-1}
e_j^\top(\widehat{\bm{\Theta}}_{t}-\bm{\Theta}_{t}^{*})V_{it}u_{it}^\top$. Since $\E(u_{it}\mid\mathcal F_{it})=0$, the second term is asymptotically negligible relative to the first. It therefore suffices to bound $\big|(NT)^{-1}\sum_{i=1}^{N}\sum_{t=1}^{T-1}
e_j^\top(\widehat{\bm{\Theta}}_{t}-\bm{\Theta}_{t}^{*})V_{it}\mu_{it}^\top\big|_1$.

For each $t=1,\ldots,T-1$, applying the Cauchy-Schwarz inequality yields 
$$\bigg|\frac{1}{N}\sum_{i=1}^Ne_j^\top(\widehat{\bm{\Theta}}_{t}-\bm{\Theta}_{t}^{*})V_{it}\mu_{it}^\top\bigg|_1\leq \sqrt{d}\bigg\{\frac{1}{N}\sum_{i=1}^N(e_j^\top(\widehat{\bm{\Theta}}_{t}-\bm{\Theta}_{t}^{*})V_{it})^2\bigg\}^{1/2}\bigg(\frac{1}{N}\sum_{i=1}^N|\mu_{it}|_2^2\bigg)^{1/2}.$$
Utilizing the prediction performance bound, it follows that 
\begin{align*}
|I_{1,j}|_1\lesssim_\P\max_{1\leq t\leq T-1}\sqrt{s_t^*}\log m_t/\sqrt{N},
\end{align*}
where we use the facts that $\frac{1}{T}\sum_{t=1}^{T-1}\big(\frac{1}{N}\sum_{i=1}^N|\mu_{it}|_2^2\big)^{1/2}=\bigO_\P(1)$ and that the dimension $d$ is fixed. As a result, we bound $|I_3^{-1}I_1I_3^{-1}P_1|_\infty$ as
$$|I_3^{-1}I_1I_3^{-1}P_1|_\infty\lesssim_\P\max_{1\leq t\leq T-1}s_t^*(\log m_t)^2/\sqrt{N^3T}.$$

Regarding the remaining terms on the right-hand side of \eqref{expand}, note that $|I_3^{-1}I_2|_{\infty}=\bigO_\P(1/\sqrt{NT})$ and $|I_3^{-1}P_2|_\infty=\bigO_\P(1/\sqrt{NT})$, which imply that $|I_3^{-1}I_2I_3^{-1}P_2|_\infty=\bigO_\P(1/(NT))$. Combining the bounds for $|I_3^{-1}P_1|_\infty$ and $|I_3^{-1}I_1|_\infty$ established above, we further deduce that 
$$|I_3^{-1}I_2I_3^{-1}P_1|_\infty\lesssim_\P\max\limits_{1\leq t\leq T-1}\sqrt{s_t^*}\log m_t/(\sqrt{N^3}T),\quad |I_3^{-1}I_1I_3^{-1}P_2|_\infty\lesssim_\P\max\limits_{1\leq t\leq T-1}\sqrt{s_t^*}\log m_t/(N\sqrt{T}).$$ 

Under the condition that $\max\limits_{1\leq t\leq T-1}\sqrt{s_t^*}\log m_t/\sqrt{N}\to0$ as $N,T\to\infty$, we conclude that $\sqrt{NT}|I_3^{-1}(I_1+I_2)I_3^{-1}(P_1+P_2)|_\infty=\smallO_\P(1)$. This completes the proof for the AB-LASSO estimator.

\underline{AB-LASSO-SS}: Analogous to \eqref{expansion}, we derive expansions for $\widehat{\theta}_{A,B}-\theta^0$ and $\widehat{\theta}_{B,A}-\theta^0$, respectively. With sample splitting, a key difference arises in analyzing $(NT)^{-1/2}\sum_{i\in\mathbb I_s}\sum_{t=1}^{T-1}(\widehat{\bm\Theta}_t-\bm\Theta_t^*)V_{it} \Delta\varepsilon_{it}$, where $s\in\{A,B\}$. Since $\widehat{\bm\Theta}_t$ is obtained from a sub-sample independent of the one used in the summation, we directly obtain  
$$(NT)^{-1/2}\bigg|\sum_{i\in\mathbb I_s}\sum_{t=1}^{T-1}(\widehat{\bm\Theta}_t-\bm\Theta_t^*)V_{it} \Delta\varepsilon_{it}\bigg|_\infty\lesssim_\P\max_{1\leq t\leq T-1}\sqrt{s_t^*}\log m_t/\sqrt{N}.$$
%Thus, the required condition reduces to $\max\limits_{2\leq t\leq T}\sqrt{s_t^*}\log m_t/\sqrt{N}\to0$, as $N,T\to\infty$. 
Consequently, the required condition remains $\max\limits_{1\leq t\leq T-1}\sqrt{s_t^*}\log m_t/\sqrt{N}\to0$, as $N,T\to\infty$. The remainder of the proof follows that of AB-LASSO estimator in a similar manner.
\end{proof}

\subsubsection{Consistency of the Variance Estimator}\label{app:var}
\begin{proof}[Proof of Lemma \ref{lem:Omega-consistency}]
Define the following sequences of matrices: 
\begin{align*}
    Q_{N,T}=\frac{1}{NT}\sum_{i=1}^N\sum_{t=2}^T \bm\Theta_t^0E(V_{it}\Delta X_{it}^\top),\quad \Sigma_{0,t,N}=\frac{1}{N}\sum_{i=1}^N \E(V_{it}V_{it}^\top(\Delta\varepsilon_{it})^2).
\end{align*}
According to the definitions of the limit matrices in Section \ref{secondstep}, we have $Q_{N,T}\to Q$ as $N,T\to\infty$ and $\Sigma_{0,t,N}\to\Sigma_{0,t}$ as $N\to\infty$. Let 
$$\bm\Sigma_{N,T}=\frac{1}{T}\sum_{t=1}^{T-1}\bm\Theta_t^0\Sigma_{0,t,N}\bm\Theta_t^{0\top},$$
then $\bm\Sigma_{N,T}\to\bm\Sigma$ as $N,T\to\infty$. By Assumption \ref{identification}, $Q_{N,T}$ is nonsingular for sufficiently large $N$ and $T$. As a result, $|\Omega_{N,T}-\Omega|_{\max}=\smallO(1)$, where $\Omega_{N,T}=Q_{N,T}^{-1}\Sigma_{N,T}(Q_{N,T}^{-1})^\top$. Therefore, it suffices to show that $|\widehat\Omega - \Omega_{N,T}|_{\max}=\smallO_\P(1)$.

Recall the sandwich form of variance estimator $\widehat\Omega = \widehat Q^{-1}\widehat{\bm\Sigma}(\widehat Q^{-1})^\top$. It follows that 
$$\widehat\Omega-\Omega_{N,T}
=\underbrace{\widehat Q^{-1}(\widehat{\bm\Sigma}-\bm\Sigma_{N,T})(\widehat Q^{-1})^\top}_{=:\bm\Delta_1} + \underbrace{\widehat Q^{-1}\bm{\Sigma}_{N,T}(\widehat Q^{-1})^\top - Q_{N,T}^{-1}\bm{\Sigma}_{N,T}(Q_{N,T}^{-1})^\top}_{=:\bm\Delta_2}.$$
For the first term, we have $|\bm\Delta_1|_{\max}\leq|\widehat Q^{-1}|_\infty|\widehat{\bm\Sigma}-\bm\Sigma_{N,T}|_{\max}|(\widehat Q^{-1})^\top|_1$. Concerning $\bm\Delta_2$, observe that $\widehat{Q}^{-1} = Q^{-1}_{N,T} - {Q}_{N,T}^{-1} (\widehat Q - Q_{N,T})\widehat{Q}^{-1}$, which gives 
$$\widehat Q^{-1}\bm{\Sigma}_{N,T}(\widehat Q^{-1})^\top= \{Q^{-1}_{N,T} - {Q}_{N,T}^{-1} (\widehat Q - Q_{N,T}) \widehat{Q}^{-1}\}\bm\Sigma_{N,T}\{Q^{-1}_{N,T} - {Q}_{N,T}^{-1} (\widehat Q - Q_{N,T}) \widehat{Q}^{-1}\}^\top.$$
Note that the leading term in $\widehat Q^{-1}\bm{\Sigma}_{N,T}(\widehat Q^{-1})^\top$ is expressed by
\begin{eqnarray*}
&&\{Q^{-1}_{N,T} - {Q}_{N,T}^{-1} (\widehat Q - Q_{N,T}) {Q}_{N,T}^{-1}\}\bm\Sigma_{N,T}\{Q^{-1}_{N,T} - {Q}_{N,T}^{-1} (\widehat Q - Q_{N,T}) {Q}_{N,T}^{-1}\}^\top\\
&=& Q^{-1}_{N,T} \bm{\Sigma}_{N,T} (Q^{-1}_{N,T})^\top + Q^{-1}_{N,T} (\widehat Q - Q_{N,T})Q^{-1}_{N,T} \bm{\Sigma}_{N,T}\{Q^{-1}_{N,T} (\widehat Q - Q_{N,T}) Q^{-1}_{N,T}\}^\top\\
&& -\, Q^{-1}_{N,T} (\widehat Q - Q_{N,T}) Q^{-1}_{N,T} \bm{\Sigma}_{N,T} (Q^{-1}_{N,T})^\top - Q^{-1}_{N,T} \bm{\Sigma}_{N,T} \{Q^{-1}_{N,T} (\widehat Q - Q_{N,T}) Q^{-1}_{N,T}\}^\top.
\end{eqnarray*}
As a result, we obtain that 
\begin{align*}
|\bm\Delta_2|_{\max} &\leq \big|Q^{-1}_{N,T} (\widehat Q - Q_{N,T}) Q^{-1}_{N,T} \bm{\Sigma}_{N,T} (Q^{-1}_{N,T})^\top\big|_{\max} + \big|Q^{-1}_{N,T} \bm{\Sigma}_{N,T} \{Q^{-1}_{N,T} (\widehat Q - Q_{N,T}) Q^{-1}_{N,T}\}^\top\big|_{\max} \\
& \quad + \bigO_\P\big(|\widehat Q - Q_{N,T}|^2_{\max}\big),    
\end{align*}
where 
\begin{align*}
\big|Q^{-1}_{N,T} (\widehat Q - Q_{N,T}) Q^{-1}_{N,T} \bm{\Sigma}_{N,T} (Q^{-1}_{N,T})^\top\big|_{\max}&\leq|Q^{-1}_{N,T}|_\infty|\widehat Q - Q_{N,T}|_{\max}|Q^{-1}_{N,T}|_1|\bm{\Sigma}_{N,T}|_1|(Q^{-1}_{N,T})^\top|_1,\\
\big|Q^{-1}_{N,T} \bm{\Sigma}_{N,T} \{Q^{-1}_{N,T} (\widehat Q - Q_{N,T}) Q^{-1}_{N,T}\}^\top\big|_{\max}&\leq |Q^{-1}_{N,T}|_\infty|\bm{\Sigma}_{N,T}|_\infty|(Q^{-1}_{N,T})^\top|_\infty|\widehat Q - Q_{N,T}|_{\max}|(Q^{-1}_{N,T})^\top|_1
\end{align*}

Then, we proceed to show that $|\widehat Q - Q_{N,T}|_{\max}=\smallO_\P(1)$ and $|\widehat{\bm\Sigma}-\bm\Sigma_{N,T}|_{\max}=\smallO_\P(1)$. Concerning $\widehat Q$, observe that
\begin{eqnarray*}
&&\bigg|(NT)^{-1}\sum_{i=1}^N\sum_{t=1}^{T-1}\{\widehat{\bm\Theta}_tV_{it}\Delta X_{it}^\top - \bm\Theta_t^0\E(V_{it}\Delta X_{it}^\top)\}\bigg|_{\max}\\
&\leq&\bigg|(NT)^{-1}\sum_{i=1}^N\sum_{t=1}^{T-1}(\widehat{\bm\Theta}_t-\bm\Theta^0_t)V_{it}\Delta X_{it}^\top\bigg|_{\max} + \bigg|(NT)^{-1}\sum_{i=1}^N\sum_{t=1}^{T-1}\bm\Theta^0_t\{V_{it}\Delta X_{it}^\top-\E(V_{it}\Delta X_{it}^\top)\}\bigg|_{\max}.
\end{eqnarray*}
As shown in the proof of Theorem \ref{main}, the first term is bounded by $\max\limits_{1\leq t\leq T-1}\sqrt{s_t^*}\log m_t/\sqrt{N}$, with probability tending to 1, and the second term is of order $\bigO_\P(1/\sqrt{NT})$. Thus, $\widehat Q$ is consistent in the sense that $|\widehat Q - Q|_{\max}=\smallO_\P(1)$. Since the matrix $Q$ has fixed dimension $d\times d$, it also follows that $|\widehat Q^{-1}- Q^{-1}|_\infty=\smallO_\P(1)$. 

Recall that $\widehat{\bm\Sigma}= (T-1)^{-1} \sum_{t=1}^{T-1}\widehat{\bm\Theta}_t \widehat \Sigma_{0,t} \widehat {\bm\Theta}_t^{\top}$, where $\widehat \Sigma_{0,t} = N^{-1}\sum_{i=1}^N V_{it}V_{it}^\top(\Delta\widehat\varepsilon_{it})^2$, with $\Delta\widehat{\varepsilon}_{it}=\Delta Y_{it}-\Delta X_{it}^\top\widehat\theta$. Accordingly, we define $\Sigma^\ast_{0,t} = N^{-1}\sum_{i=1}^N V_{it}V_{it}^\top(\Delta\varepsilon_{it})^2$. For each $t=1,\ldots,T-1$, observe that
\begin{align*}
\Sigma_{0,t}^\ast-\widehat\Sigma_{0,t} &= N^{-1}\sum_{i=1}^NV_{it}V_{it}^\top\{(\Delta\varepsilon_{it})^2 - (\Delta\widehat\varepsilon_{it})^2\}\\
&= N^{-1}\sum_{i=1}^NV_{it}V_{it}^\top(\Delta\varepsilon_{it} - \Delta\widehat\varepsilon_{it})(\Delta\varepsilon_{it} + \Delta\widehat\varepsilon_{it})\\
&=N^{-1}\sum_{i=1}^NV_{it}V_{it}^\top\{-\Delta X_{it}^\top(\theta^0-\widehat\theta)\}\{2\Delta\varepsilon_{it}+\Delta X_{it}^\top(\theta^0-\widehat\theta)\}
\end{align*}
Utilizing the fact that $|\Delta X_{it}^\top(\widehat\theta-\theta^0)|\leq|\Delta X_{it}|_\infty|\widehat\theta - \theta^0|_1$, we get
\begin{align*}
\max_{1\leq t\leq T-1}|\Sigma_{0,t}^\ast-\widehat\Sigma_{0,t}|_{\max}&\leq \max_{1\leq t\leq T-1}2\Big|N^{-1}\sum_{i=1}^NV_{it}V_{it}^\top\Delta X_{it}^\top\Delta\varepsilon_{it}\Big|_{\max}|\widehat\theta - \theta^0|_1 \\
&\quad + \max_{1\leq t\leq T-1}\Big|N^{-1}\sum_{i=1}^NV_{it}V_{it}^\top\Big|_{\max}|\Delta X_{it}|_\infty^2|\widehat\theta - \theta^0|_1^2\\
&\leq 2C|\widehat\theta - \theta^0|_1 + C|\widehat\theta - \theta^0|_1^2,
\end{align*}
for some constant $C>0$. Similarly, there exists $C'>0$ such that 
\begin{align*}
\max_{1\leq t\leq T-1}|\bm\Theta_t^0(\Sigma_{0,t}^\ast-\widehat\Sigma_{0,t})\bm\Theta_t^{0\top}|_{\max} %&\leq \max_{1\leq t\leq T-1}2|N^{-1}\sum_{i=1}^N\bm\Theta_t^0V_{it}V_{it}^\top\bm\Theta_t^{0\top}\Delta X_{it}^\top\Delta\varepsilon_{it}|_{\max}|\widehat\theta - \theta^0|_1\\
%&\quad + \max_{1\leq t\leq T-1}|N^{-1}\sum_{i=1}^N\bm\Theta_t^0V_{it}V_{it}^\top\bm\Theta_t^{0\top}|_{\max}|\Delta X_{it}|_\infty^2|\widehat\theta - \theta^0|_1^2\\
&\leq 2C'|\widehat\theta - \theta^0|_1 + C'|\widehat\theta - \theta^0|_1^2.
\end{align*}
Combining with the consistency of the final estimator $\widehat\theta$ (shown in Theorem \ref{main}) and the consistency of $\widehat{\bm\Theta}_t$ (shown in Theorem \ref{lassobound}), we have the following expansion
\begin{eqnarray*}
&&\widehat{\bm\Theta}_t \widehat \Sigma_{0,t} \widehat{\bm\Theta}_t^{\top} - \bm\Theta^0_t \Sigma_{0,t}^\ast \bm\Theta_t^{0\top} \\
&=& (\widehat{\bm\Theta}_t - \bm\Theta_t^0)\Sigma^\ast_{0,t}\bm\Theta_t^{0\top} + \widehat{\bm\Theta}_t\Sigma^\ast_{0,t}(\widehat{\bm\Theta}_t - \bm\Theta_t^0)^\top + \widehat{\bm\Theta}_t(\widehat\Sigma_{0,t} - \Sigma^\ast_{0,t})\widehat{\bm\Theta}_t^{\top}\\
&=& (\widehat{\bm\Theta}_t - \bm\Theta_t^0)\Sigma^\ast_{0,t}\bm\Theta_t^{0\top} + \bm\Theta_t^{0}\Sigma^\ast_{0,t}(\widehat{\bm\Theta}_t - \bm\Theta_t^0)^\top + \bm\Theta_t^{0}(\widehat\Sigma_{0,t} - \Sigma^\ast_{0,t})\bm\Theta_t^{0\top} + \smallO_\P(1),
\end{eqnarray*}
which implies that 
\begin{align*}
|\widehat{\bm\Theta}_t \widehat \Sigma_{0,t} \widehat{\bm\Theta}_t^{\top} - \bm\Theta^0_t \Sigma_{0,t}^\ast \bm\Theta_t^{0\top}|_{\max}&\leq 2|\Sigma_{0,t}^\ast\bm\Theta_t^0|_{\max}|\widehat{\bm\Theta}_t - \bm\Theta_t^0|_1 + |\bm\Theta_t^0(\Sigma_{0,t}^\ast-\widehat\Sigma_{0,t})\bm\Theta_t^{0\top}|_{\max}+\smallO_\P(1)\\
&=\smallO_\P(1).    
\end{align*}
Lastly, taking average over $t$ and applying law of large numbers yield that $|\widehat{\bm\Sigma}-\bm\Sigma_{N,T}|_{\max}=\smallO_\P(1)$. 

Hence, combining these results, we conclude that $|\bm\Delta_1|_{\max}=\smallO_\P(1)$ and $|\bm\Delta_2|_{\max}=\smallO_\P(1)$, which completes the proof.

\end{proof}

\subsubsection{Asymptotic Normality for the General Model Estimator}\label{app:hdm}
\begin{proof}[Proof of Theorem \ref{maindiverged}]
Observe that the estimator $\widehat\theta_1$ obtained by \eqref{diverged1}, regardless of the demeaning transformation due to the presence of time effects, can be expressed as
\begin{align*}\widehat\theta_1-\theta_1^0&=\bigg(\frac{1}{NT}\sum_{i=1}^N\sum_{t=1}^{T-1}\widehat{\mathcal W}_{t}^{\top}\widehat{U}_{it}{\Delta X}_{1,it}^{\top}\bigg)^{-1}\bigg(\frac{1}{NT}\sum_{i=1}^N\sum_{t=1}^{T-1}\widehat{\mathcal W}_{t}^{\top}\widehat{U}_{it}\Delta X_{2,it}^\top\theta_2^0\bigg)\\
&\quad+\bigg(\frac{1}{NT}\sum_{i=1}^N\sum_{t=1}^{T-1}\widehat{\mathcal W}_{t}^{\top}\widehat{U}_{it}{\Delta X}_{1,it}^{\top}\bigg)^{-1}\bigg(\frac{1}{NT}\sum_{i=1}^N\sum_{t=1}^{T-1}\widehat{\mathcal W}_{t}^{\top}\widehat{U}_{it}\Delta \varepsilon_{it}\bigg)\\
&=:L_1+L_2.
\end{align*}
We will demonstrate that $|L_1|_\infty$ is negligible asymptotically, and that $L_2$ approaches\\ $(NT)^{-1}\sum_{i=1}^N\sum_{t=1}^{T-1}\mathcal W_{t}^\top U_{it}\Delta\varepsilon_{it}$ closely enough. This will lead to the conclusion upon application of the central limit theorem.

\underline{Step 1:} We first show that the Dantzig estimator $\widehat{\mathcal W}_t$ is close to the desired weighting matrix $\mathcal W_t$ with respect to various norms.
According to the definitions of $\mathcal W_t$ and $\widehat{\mathcal W}_t$ provided in Section \ref{debias}, we derive the following results:
\begin{align*}
|\widehat{\mathcal W}_t - \mathcal W_t|_{\max}&=|M_t^{-1}|_\infty|\widehat M_t\widehat{\mathcal W}_t + (M_t-\widehat M_t)\widehat{\mathcal W}_t - M_t\mathcal W_t|_{\max}\\
&\leq |M_t^{-1}|_\infty\{|\widehat M_t\widehat{\mathcal W}_t - \mathbf I_{d\times d_1}|_{\max} + |(M_t-\widehat M_t)\widehat{\mathcal W}_t|_{\max}\}\\
&\leq|M_t^{-1}|_\infty(\ell_t + |M_t-\widehat M_t|_{\max}|\widehat{\mathcal W}_t|_{1,1})\\
&\leq|M_t^{-1}|_\infty(\ell_t + |M_t-\widehat M_t|_{\max}|{\mathcal W}_t|_{1,1}).
\end{align*}
For each $t=1,\ldots,T-1$, we have $|M_{t}-\widehat{M}_{t}|_{\max}\lesssim_\P\rho_{N,t}$, for some $\rho_{N,t}\to0$ as $N\to\infty$. This result follows from the consistency of the LASSO estimators in step 1. 
By parts (i) and (iii) of Assumption \ref{assum}, we obtain that
\beq\label{max}
\max_{1\leq t\leq T-1}|\widehat{\mathcal W}_t - \mathcal W_t|_{\max}\lesssim_\P c_n\max_{1\leq t\leq T-1}(\ell_t + \rho_{N,t}c_n)\lesssim c_n^2\sqrt{v_n/N}.
\eeq
The last inequality is implied by bounding $|M_{t}-\widehat{M}_{t}|_{\max}$ using the concentration inequality in Lemma \ref{emp}, under Assumptions \ref{a1}--\ref{a2}.

To analyze $|\widehat{\mathcal W}_t - \mathcal W_t|_{1,1}$, we consider a truncation argument with $\tau_n=\sqrt{v_n/N}$. Specifically, we have
\begin{align*}
|\widehat{\mathcal W}_t - \mathcal W_t|_{1,1}&=|\widehat{\mathcal W}_t - \mathcal W_t|_{1,1}\mathbf 1(|\mathcal W_t|_{\max}\leq\tau_n) + |\widehat{\mathcal W}_t - \mathcal W_t|_{1,1}\mathbf 1(|\mathcal W_t|_{\max}>\tau_n)\\
&\leq2|\mathcal W_t|_{1,1}\mathbf 1(|\mathcal W_t|_{\max}\leq\tau_n) + \sum_{i=1}^d\sum_{j=1}^{d_1}|\widehat{\mathcal W}_{t,ij}-\mathcal W_{t,ij}|\mathbf 1(|\mathcal W_t|_{\max}>\tau_n).
\end{align*}
Utilizing the bound in \eqref{max} and Assumption \ref{assum}, it follows that
\begin{align*}
\max_{1\leq t\leq T-1}|\widehat{\mathcal W}_t - \mathcal W_t|_{1,1}&\lesssim_\P2\max_{1\leq t\leq T-1}|\mathcal W_{t,ij}|^r\mathbf 1(|\mathcal W_t|_{\max}\leq\tau_n)/|\mathcal W_{t,ij}|^{r-1} \\
&\quad\,+ c_n^2\sqrt{v_n/N}\sum_{i=1}^d\sum_{j=1}^{d_1}|\mathcal W_{t,ij}/\tau_n|^r\mathbf 1(|\mathcal W_t|_{\max}>\tau_n)\\
&\leq2\tau_n^{1-r}w_n + \tau_n^{-r}w_nc_n^2\sqrt{v_n/N}\lesssim c_n^2w_n(v_n/N)^{\frac{1-r}{2}},
\end{align*}
where the parameter $r\in[0,1)$ is chosen to ensure that Assumption \ref{assum}(i) holds.

\underline{Step 2:} Next, we analyze the rate of $|L_1|_\infty$. Note that the dimension of $X_{1,it}$, i.e., $d_1$, is fixed. The max norm, spectral norm (denoted by $|\cdot|_2$), and infinity norm of a $d_1\times d_1$ matrix are equivalent up to a constant factor that depends on $d_1$. Similarly, the $\ell_1$-norm and $\ell_\infty$-norm of a $d_1\times 1$ vector exhibit the same relationship. Recall the constraints in the Dantzig estimator in \eqref{penalize}, we find that
\begin{align*}
\bigg|N^{-1}\sum_{i=1}^N\widehat{\mathcal W}_{t}^{\top}\widehat{U}_{it}{\Delta X}_{1,it}^{\top} - \mathbf I_{d_1\times d_1}\bigg|_2\lesssim\ell_t,\quad\bigg|N^{-1}\sum_{i=1}^N\widehat{\mathcal W}_{t}^{\top}\widehat{U}_{it}{\Delta X}_{2,it}^{\top}\bigg|_{\max}\leq\ell_t.
\end{align*}

Denote the smallest and largest singular values of a matrix by $\sigma_{\min}(\cdot)$ and $\sigma_{\max}(\cdot)$, respectively. Applying Weyl's inequality for singular values, $\sigma_{\min}(C+D)\geq\sigma_{min}(C) - \sigma_{\max}(D)$, we can bound $|L_1|_\infty$ as follows:
\begin{align*}
|L_1|_\infty&\leq\bigg|\bigg(\frac{1}{NT}\sum_{i=1}^N\sum_{t=1}^{T-1}\widehat{\mathcal W}_{t}^{\top}\widehat{U}_{it}{\Delta X}_{1,it}^{\top}\bigg)^{-1}\bigg|_\infty\bigg|\frac{1}{NT}\sum_{i=1}^N\sum_{t=1}^{T-1}\widehat{\mathcal W}_{t}^{\top}\widehat{U}_{it}\Delta X_{2,it}^\top\theta_2^0\bigg|_\infty\\
&\lesssim\bigg|\bigg(\frac{1}{NT}\sum_{i=1}^N\sum_{t=1}^{T-1}\widehat{\mathcal W}_{t}^{\top}\widehat{U}_{it}{\Delta X}_{1,it}^{\top} - \mathbf I_{d_1\times d_1} + \mathbf I_{d_1\times d_1}\bigg)^{-1}\bigg|_2\bigg|\frac{1}{NT}\sum_{i=1}^N\sum_{t=1}^{T-1}\widehat{\mathcal W}_{t}^{\top}\widehat{U}_{it}{\Delta X}_{2,it}^{\top}\bigg|_{\max}|\theta_2^0|_1\\
&\leq \bigg\{\sigma_{\min}(\mathbf I_{d_1\times d_1}) - \frac{1}{T}\sum_{t=1}^{T-1}\sigma_{\max}\bigg(\frac{1}{N}\sum_{i=1}^N\widehat{\mathcal W}_{t}^{\top}\widehat{U}_{it}{\Delta X}_{1,it}^{\top} - \mathbf I_{d_1\times d_1}\bigg)\bigg\}^{-1}\max_{1\leq t\leq T-1}\ell_t\vartheta_n\\
&\lesssim\max_{1\leq t\leq T-1}\ell_t\vartheta_n=\smallO(1/\sqrt{NT}) \quad\text{ (by Assumption \ref{assum}(iii))}.
\end{align*}

\underline{Step 3:} Lastly, we establish the asymptotic normality of the leading term $L_2$. To achieve this, we verify that $L_2$ is sufficiently close to $(NT)^{-1}\sum_{i=1}^N\sum_{t=1}^{T-1}\mathcal W_{t}^\top U_{it}\Delta\varepsilon_{it}$. Based on the findings in Step 2, we observe that
\begin{eqnarray*}
&&\bigg|\frac{1}{NT}\sum_{i=1}^N\sum_{t=1}^{T-1}\mathcal W_{t}^\top U_{it}\Delta\varepsilon_{it} - L_2\bigg|_\infty\\
&\leq&\bigg|\bigg(\frac{1}{NT}\sum_{i=1}^N\sum_{t=1}^{T-1}\widehat{\mathcal W}_{t}^{\top}\widehat{U}_{it}{\Delta X}_{1,it}^{\top}\bigg)^{-1}\bigg|_{\max}\bigg|\frac{1}{NT}\sum_{i=1}^N\sum_{t=1}^{T-1}(\widehat{\mathcal W}_{t}^{\top}\widehat{U}_{it}\Delta\varepsilon_{it} - \mathcal W_{t}^\top U_{it}\Delta\varepsilon_{it})\bigg|_1\\
&&+ \bigg|\bigg(\frac{1}{NT}\sum_{i=1}^N\sum_{t=1}^{T-1}\widehat{\mathcal W}_{t}^{\top}\widehat{U}_{it}{\Delta X}_{1,it}^{\top}\bigg)^{-1}-\mathbf I_{d_1\times d_1}\bigg|_{\max}\bigg|\frac{1}{NT}\sum_{i=1}^N\sum_{t=1}^{T-1}\mathcal W_{t}^\top U_{it}\Delta\varepsilon_{it}\bigg|_1 \\
&\lesssim&\bigg|\frac{1}{NT}\sum_{i=1}^N\sum_{t=1}^{T-1}(\widehat{\mathcal W}_{t}^{\top}\widehat{U}_{it}\Delta\varepsilon_{it} - \mathcal W_{t}^\top U_{it}\Delta\varepsilon_{it})\bigg|_\infty + \smallO_\P(1/\sqrt{NT})\\
&\leq&\bigg|\frac{1}{NT}\sum_{i=1}^N\sum_{t=1}^{T-1}(\widehat{\mathcal W}_{t} - \mathcal W_{t})^\top \widehat U_{it}\Delta\varepsilon_{it})\bigg|_\infty + \bigg|\frac{1}{NT}\sum_{i=1}^N\sum_{t=1}^{T-1}\mathcal W_{t}^\top (\widehat U_{it} - U_{it})\Delta\varepsilon_{it}\bigg|_\infty + \smallO_\P(1/\sqrt{NT}).
\end{eqnarray*}

Using the results obtained in Step 1, we find that
$$\bigg|\frac{1}{NT}\sum_{i=1}^N\sum_{t=1}^{T-1}(\widehat{\mathcal W}_{t} - \mathcal W_{t})^\top \widehat U_{it}\Delta\varepsilon_{it})\bigg|_\infty\lesssim_{\P}c_n^2w_n(v_n/N)^{\frac{1-r}{2}}/\sqrt{NT}.$$
Furthermore, following similar steps as in bounding $|I_3^{-1}P_1|_\infty$ in the proof of Theorem \ref{main}, we obtain that
$$\bigg|\frac{1}{NT}\sum_{i=1}^N\sum_{t=1}^{T-1}\mathcal W_{t}^\top (\widehat U_{it} - U_{it})\Delta\varepsilon_{it}\bigg|_\infty\lesssim_{\P}c_n\max_{1\leq t\leq T-1}\sqrt{s_t^*}\log m_t/(N\sqrt{T}).$$
Combining these findings with Assumption \ref{assum}(ii), we conclude that
$$\bigg|\frac{1}{NT}\sum_{i=1}^N\sum_{t=1}^{T-1}\mathcal W_{t}^\top U_{it}\Delta\varepsilon_{it} - L_2\bigg|_\infty=\smallO_\P(1/\sqrt{NT}).$$
The proof is then completed by applying a central limit theorem to $\frac{1}{\sqrt{NT}}\sum\limits_{i=1}^N\sum\limits_{t=1}^{T-1}\mathcal W_{t}^\top U_{it}\Delta\varepsilon_{it}$, following a similar approach as in the proof of Theorem \ref{main}.

It is noteworthy that with the implementation of a sample-splitting procedure, where $\widehat{\mathcal W}_t$ and $\widehat U_{it}$ are estimated from an auxiliary sub-sample, we could get %sharper 
same bounds as follows:
\begin{align*}
\bigg|\frac{1}{NT}\sum_{i\in\mathbb I_s}\sum_{t=1}^{T-1}(\widehat{\mathcal W}_{t} - \mathcal W_{t})^\top \widehat U_{it}\Delta\varepsilon_{it})\bigg|_\infty&\lesssim_{\P}c_n^2w_n(v_n/N)^{\frac{1-r}{2}}/\sqrt{NT},\\
\bigg|\frac{1}{NT}\sum_{i\in\mathbb I_s}\sum_{t=1}^{T-1}\mathcal W_{t}^\top (\widehat U_{it} - U_{it})\Delta\varepsilon_{it}\bigg|_\infty&\lesssim_{\P}c_n\max_{1\leq t\leq T-1}\sqrt{s_t^*}\log m_t/(N\sqrt{T}),
\end{align*}
where $s\in\{A,B\}$. As a result, the required rate condition in Assumption \ref{assum}(ii) remains as %can be improved to
$c_n^2w_n(v_n/N)^{\frac{1-r}{2}} + c_n\max\limits_{1\leq t\leq T-1}\sqrt{s_t^*}\log m_t/\sqrt{N}=\smallO(1)$.
\end{proof}

\subsubsection{Time Effects}\label{add}
In this subsection, we discuss how the inference theory, such as Theorem \ref{main}, adapts in the presence of time effects $\gamma_t$. In particular, we compare the leading term with and without cross-sectional demeaning.
%What is the change of order after demean?

Recall the definition $\Delta\widetilde\varepsilon_{it} = \Delta\varepsilon_{it} - \sum_{j=1}^N \Delta \varepsilon_{jt}/N$, $\Delta\varepsilon_{it} = c_t\{\varepsilon_{it} - \sum_{s=1}^{T-t}\varepsilon_{i,t+s}/(T-t)\}$, and $c_t = \sqrt{(T-t)/(T-t+1)}$. Observe that $\sum_{i=1}^N\Delta\widetilde\varepsilon_{it}=0$ and $\Delta\widetilde\varepsilon_{it}=\Delta\varepsilon_{it}-\Delta\bar\varepsilon_t$, where $\bar\varepsilon_t=\sum_{j=1}^N\varepsilon_{jt}/N$. Let $\mu_t^V\defeq\E(V_{it})$. It follows that 
\begin{align*}
\frac{1}{NT}\sum_{i=1}^N\sum_{t=1}^{T-1}\bm\Theta_t^0V_{it}\Delta\widetilde{\varepsilon}_{it}&=\frac{1}{NT}\sum_{i=1}^N\sum_{t=1}^{T-1}\bm\Theta_t^0V_{it}\Delta\varepsilon_{it} - \frac{1}{NT}\sum_{i=1}^N\sum_{t=1}^{T-1}\bm\Theta_t^0\mu_t^V\Delta\varepsilon_{it} \\
&\quad - \frac{1}{NT}\sum_{i=1}^N\sum_{t=1}^{T-1}\bm\Theta_t^0(V_{it}-\mu_t^V)\Delta\bar\varepsilon_{t}\\
&=:P_1+R_1+R_2.
\end{align*}
In particular, %we find the terms on the right-hand side have the following orders:
$$P_1+R_1=\frac{1}{NT}\sum_{t=1}^{T-1}\bm\Theta_t^0\sum_{i=1}^N(V_{it}-\mu_t^V)\Delta\varepsilon_{it}=\bigO_\P(1/\sqrt{NT}),$$
$$R_2=-\frac{1}{T}\sum_{t=1}^{T-1}\bm\Theta_t^0\frac{1}{N}\sum_{i=1}^N(V_{it}-\mu_t^V)\Delta\bar\varepsilon_t=\bigO_\P(1/(N\sqrt{T})),$$ 
as $N^{-1}\sum_{i=1}^N(V_{it}-\mu_t^V)=\bigO_\P(1/\sqrt{N})$ and $\Delta\bar\varepsilon_t=\bigO_\P(1/\sqrt{N})$. Therefore, cross-sectional demeaning introduces only an additional higher-order component. The leading term in the expansion \eqref{expansion} has the same stochastic order with or without cross-sectional demeaning. The analysis of the remaining terms proceeds analogously.

\subsubsection{Efficiency}\label{app_ebound}
In this subsection, we discuss the connection between our results and efficiency in dynamic panel models with fixed effects. Specifically, under a univariate panel AR(1) model with certain conditions, we demonstrate that our estimator is efficient by invoking the results in \citet{hahn2002asymptotically}.
\begin{Proposition}[Attainability of the Efficiency Bound]\label{ebound}
Consider the univariate panel AR(1) model:
$$Y_{it} = \alpha_i + \theta^0 Y_{i,t-1} + \varepsilon_{it},\quad i=1,\ldots,N,t=1,\ldots,T.$$
Suppose that the conditions of Theorem 3 in \citet{hahn2002asymptotically} are satisfied: (i) $\varepsilon_{it}\stackrel{\operatorname{i.i.d.}}{\sim}\operatorname{N}(0,\sigma^2_\varepsilon)$; (ii) $0<\lim\limits_{N,T\to\infty}\frac{N}{T}<\infty$; (iii) $\lim\limits_{N\to\infty}(\theta^0)^N=0$; and (iv) $\frac{1}{N}\sum_{i=1}^N|Y_{i0}|^2=\bigO(1)$ and $\frac{1}{N}\sum_{i=1}^N|\alpha_{i}|^2=\bigO(1)$.\footnote{We can allow $\alpha_i$ to be random, provided that $\E(\varepsilon_{it}|\alpha_i)=0$. The proof then applies by conditioning on the realization of these effects.} Then the asymptotic variance of the proposed estimator $\widehat\theta$, as given in \eqref{var}, simplifies to $1-(\theta^0)^2$ and therefore the efficiency bound is attained. 
\end{Proposition}

\begin{proof}
By recursive substitution, the model can be rewritten as $Y_{it} = \mu_i + u_{it}$, where $\mu_i=\alpha_i/(1-\theta^0)$, $u_{it}=\sum_{\ell\geq0}(\theta^0)^\ell\varepsilon_{i,t-\ell}$. %, and $\E(u_{it}|\alpha_i)=0$. 
Observe that the FOD transformation yields
$$\Delta Y_{i,t-1}=a_{t}Y_{i,t-1} - b_{t}\sum_{s=t}^TY_{is}=a_{t}u_{i,t-1} - b_{t}\sum_{s=t}^Tu_{is},$$
where $a_{t}=\sqrt{(T-t+1)/(T-t+2)}$ and $b_{t}=a_{t}/(T-t+1)$. 

Let $V_{it}=(1,Y_{i,t-1},\ldots,Y_{i0})^\top$ and $U_{it}=(1,u_{i,t-1},\ldots,u_{i0})^\top$. It follows that $\E(V_{it}\Delta Y_{i,t-1})=\E(U_{it}\Delta Y_{i,t-1})$ and $\E(V_{it}V_{it}^\top)=\E(U_{it}U_{it}^\top) + \mu_i^2\mathbf 1\mathbf 1^\top$, where $\mathbf{1}$ denotes a vector of ones with the same dimension as $V_{it}$. The fixed effects therefore contribute only a rank-one component in the direction of the constant regressor.
In the linear projection of $\Delta Y_{i,t-1}$ onto $V_{it}$, this rank-one term affects only the intercept and cancels from the slope coefficients. Consequently, 
$$Q_t\defeq\E(V_{it}^\top\Delta Y_{i,t-1})\E(V_{it}V_{it}^\top)^{-1}\E(V_{it}\Delta Y_{i,t-1})=\E(V_{it}^\top\Delta Y_{i,t-1})\E(U_{it}U_{it}^\top)^{-1}\E(V_{it}\Delta Y_{i,t-1}),$$
where $\E(U_{it}U_{it}^\top)$ is a Toeplitz matrix whose inverse is tridiagonal. As a result, we obtain 
$$Q_t=\frac{\sigma^2_\varepsilon(a_{t} - \theta^0b_{t}S_t)^2}{1-(\theta^0)^2},$$
where $S_t=\frac{1-(\theta^0)^{T-t+1}}{1-\theta^0}$. 
Combining this with the identity $\E(V_{it}V_{it}^\top(\Delta\varepsilon_{it})^2)=\sigma_\varepsilon^2\E(V_{it}V_{it}^\top)$, which is ensured by the homoskedasticity of the errors, the asymptotic variance of the proposed estimator, as given in \eqref{var}, can be expressed as $\Omega = \lim\limits_{T\to\infty}\frac{1}{T}\sum_{t=1}^{T-1}Q_t^{-1}\sigma^2_\varepsilon$. 

As $T\to\infty$, we have
$$\max_{1\leq t\leq T-1}|(a_{t} - \theta^0b_{t}S_t) - 1|\leq \max_{1\leq t\leq T-1}|a_{t}-1| + |\theta^0|\max_{1\leq t\leq T-1}|b_{t}S_t| \to 0,$$
which implies that the asymptotic variance simplifies to $\Omega = 1-(\theta^0)^2$. Finally, by invoking Theorem 3 of \citet{hahn2002asymptotically}, we conclude that the proposed estimator is efficient in this context. 
\end{proof}

\subsection{Main Results for Estimator using First Differences}\label{app:FD}
In this subsection, we consider taking first differences (FD) over time to remove individual. Specifically, the transformed model is given by 
$$\Delta Y_{it} =  \Delta X_{it}^\top\theta^0 + \Delta\varepsilon_{it}, \quad 2\leq t\leq T,$$
where $\Delta Z_{it} = Z_{it} - Z_{i,t-1}$, for $Z_{it} \in \{Y_{it},  X_{it},  \varepsilon_{it}\}$.  The transformed error, $\Delta\varepsilon_{it}$, satisfies the moment conditions
$$
\E (X_{is}\Delta\varepsilon_{it}) = 0, \text{ for all } 1\leq s \leq t-1, \, 2\leq t\leq T.
$$
The estimation procedure remains largely unchanged, except for a minor adjustment to the time indexing. A key difference arises in the inference theory for step 2, particularly in the asymptotic variance of the estimator. 

Recall that the AB-LASSO estimator can be expressed by 
$$\widehat\theta - \theta^0 = \bigg(\sum_{i=1}^N \sum_{t=2}^{T} \widehat{\bm\Theta}_tV_{it}\Delta X_{it}^{\top} \bigg)^{-1} \bigg(\sum_{i=1}^N \sum_{t=2}^{T} \widehat{\bm\Theta}_tV_{it} \Delta\varepsilon_{it}\bigg).$$
Under assumptions analogous to those imposed in the main text, the asymptotic variance of $\widehat\theta$ admits a sandwich form. In particular, the martingale difference property of the errors implies that $\E(\varepsilon_{i,t-1}\varepsilon_{is} \mid V_{it})=0$, for $1\leq s<t-1\leq T$, which in turn yields $\E(\Delta\varepsilon_{it}\Delta\varepsilon_{i,t-\ell} \mid V_{it})=0$, for $\ell>1$. By defining $\Sigma_{0,t}\defeq\lim\limits_{N\to\infty}N^{-1}\sum_{i=1}^N\E(V_{it}V_{it}^\top(\Delta\varepsilon_{it})^2)$ and $\Sigma_{1,t}\defeq\lim\limits_{N\to\infty}N^{-1}\sum_{i=1}^N\E(V_{it}V_{i,t-1}^\top\Delta\varepsilon_{it}\Delta\varepsilon_{i,t-1})$, we can express the asymptotic variance of $\widehat\theta$ in the form of 
\begin{align*}
\Omega = Q^{-1}\bm\Sigma(Q^{-1})^\top, 
\end{align*}
where
$$\bm\Sigma = \lim_{T\to\infty}\Big(\frac{1}{T}\sum_{t=2}^T{\bm\Theta}_t^0\Sigma_{0,t}{\bm\Theta}_t^{0\top} + \frac{1}{T} \sum_{t=3}^T{\bm\Theta}_t^0\Sigma_{1,t}{\bm\Theta}_{t-1}^{0\top} + \frac{1}{T}\sum_{t=3}^T{\bm\Theta}_{t-1}^0\Sigma_{1,t}^\top{\bm\Theta}_{t}^{0\top}\Big).$$
Accordingly, the empirical analog of $\Omega$ is 
\begin{align*}
\widehat\Omega = \widehat Q^{-1}\widehat{\bm\Sigma}(\widehat Q^{-1})^\top, 
\end{align*}
where $\widehat Q = (NT)^{-1}\sum_{i=1}^N\sum_{t=2}^T\widehat{\bm\Theta}_tV_{it}\Delta X_{it}^{\top}$, and 
$$\widehat{\bm\Sigma}= \frac{1}{T} \sum_{t=2}^T\widehat {\bm\Theta}_t \widehat \Sigma_{0,t} \widehat {\bm\Theta}_t^{\top} + \frac{1}{T} \sum_{t=3}^T \widehat {\bm\Theta}_t \widehat \Sigma_{1,t} \widehat {\bm\Theta}_{t-1}^{\top} + \frac{1}{T} \sum_{t=3}^T \widehat {\bm\Theta}_{t-1} \widehat \Sigma_{1,t}^\top \widehat {\bm\Theta}_{t}^{\top},$$
with $\widehat \Sigma_{0,t} = N^{-1}\sum_{i=1}^N V_{it}V_{it}^\top(\Delta\widehat \varepsilon_{it})^2$ and $\widehat \Sigma_{1,t} = N^{-1} \sum_{i=1}^N V_{it}V_{i,t-1}^\top\Delta \widehat \varepsilon_{it}\Delta \widehat \varepsilon_{i,t-1}$.

\begin{Theorem}[Asymptotic Normality of AB-LASSO with FD]\label{main.fd}
	Under Assumptions \ref{a1}--\ref{identification}, where the transformation $\Delta$ denotes FD, suppose that the asymptotic variance $\Omega$ is a positive definite matrix, and that $\max\limits_{2\leq t\leq T} s_t^*\log m_t\sqrt{T/N}\to 0$ as $N,T\to\infty$. Then the AB-LASSO estimator $\widehat{\theta}$ is consistent for $\theta^0$, and
	$$\sqrt{NT}(\widehat{\theta}- \theta^0 )\stackrel{\mathcal{L}}{\to} \operatorname{N}(0, \Omega).$$
\end{Theorem}

\begin{proof}
	As shown at the beginning of the proof of Theorem \ref{main}, we obtain the analogous expansion under the FD transformation $\Delta$:
	$$\widehat\theta - \theta^0=I_3^{-1}P_2 +I_3^{-1}P_1- I_3^{-1}(I_1+I_2)I_3^{-1}(P_1+P_2)+\smallO_\P((NT)^{-1/2}).$$
	%We will prove that $I_3^{-1}P_2$ is the leading term, with $\sqrt{NT}I_3^{-1}P_2$ exhibiting asymptotic Gaussianity and analyze the orders of the other terms.
	
	We refer to a central limit theorem for stationary random field (Theorem 1 of \citet{el2013central}) to establish asymptotic normality for the scaled leading term $\sqrt{NT}I_3^{-1}P_2$. To achieve this, we must verify the necessary conditions outlined below.
	
	Define an index set $\mathcal J_{N,T}\defeq\{(i,t):1\leq i\leq N,2\leq t\leq T\}$. As $N,T\to\infty$, it follows that the cardinality $|\mathcal J_{N,T}|\to\infty$, while the ratio $|\partial\mathcal J_{N,T}|/|\mathcal J_{N,T}|\to0$, where $\partial \mathcal J_{N,T}$ contains the boundary points of $\mathcal J_{N,T}$. Under Assumption \ref{a1}, the $d$-dimensional process $Z_{(i,t)}\defeq\bm\Theta^*_tV_{it}\Delta\varepsilon_{it}$ is stationary over $t$ and i.i.d. across $i$. For $k=1,\ldots,d$, $Z_{(i,t),k}$ can be represented as $Z_{(i,t),k}=h_{(i,t),k}(\ldots,\eta_{(i,t-1)},\eta_{(i,t)})$, where $h_{(i,t),k}$ are measurable functions, and $\eta_{(i,t)}$ for $i\in\mathbb N$, $t\in\mathbb Z$, are i.i.d. random elements. By Definition \ref{dep} and Assumption \ref{a2}(i), it follows that 
	$$\sum_{(i,t)\in\mathcal J_{N,T}}\|Z_{(i,t),k}^*-Z_{(i,t),k}\|_2<\infty.$$
	Moreover, Assumption \ref{a2}(ii), along with the cross-sectional independence assumption, implies that for $k,k'=1,\ldots,d$, the variance
	\begin{eqnarray*}
		&&\E\bigg[\bigg(\sum_{(i,t)\in\mathcal J_{N,T}}Z_{(i,t),k)}\bigg)\bigg(\sum_{(i,t)\in\mathcal J_{N,T}}Z_{(i,t),k'}\bigg)\bigg]\\
		&=&\sum_{(i,t)\in\mathcal J_{N,T}}\E(Z_{(i,t),k}Z_{(i,t),k'})\quad+\sum_{(i,t)\in\mathcal J'_{N,T}}\E(Z_{(i,t),k}Z_{(i,t-1),k'})+\sum_{(i,t)\in\mathcal J'_{N,T}}\E(Z_{(i,t-1),k}Z_{(i,t),k'})
	\end{eqnarray*}
	is of order $NT$, where $\mathcal J'_{N,T}\defeq\{(i,t):1\leq i\leq N,3\leq t\leq T\}$.
	Therefore, based on Assumption \ref{identification}, by applying Theorem 1 of \citet{el2013central} and Slutsky's theorem, we deduce that
	$$\sqrt{NT}I_3^{-1}P_2\stackrel{\mathcal{L}}{\to} \operatorname{N}(0, Q^{-1}\bm\Sigma(Q^{-1})^\top).$$
	
	Recalling the subspace $\Omega_t(c_0,s_t^*)$ defined in Assumption \ref{33}, and considering the 
	entropy condition (with respect to the $d_2$-metric): 
	$$\operatorname{ent}\big(\epsilon,\bigcup\nolimits_{2\leq t\leq T}\Omega_t(c_0,s_t^*)\big)\lesssim T\max_{2\leq t\leq T}s_t^*\log(m_t/\epsilon), \,\text{ for all } 0<\epsilon<1,$$
	by employing Theorem \ref{emplemma} based on Assumptions \ref{a1}--\ref{a2}, we obtain:
	\begin{equation}\label{empbound}
	\sup_{\delta_t\in\Omega_t(c_0,s_t^*),\,t=2,\ldots,T}%{(\delta_t)_{t=2}^T\in\mathcal H}
	\bigg|(NT)^{-1}\sum_{i=1}^N\sum_{t=2}^T\delta_t^\top V_{it}\Delta\varepsilon_{it}\bigg|\lesssim_\P\sqrt{\max_{2\leq t\leq T}s_t^*\log m_t}/\sqrt{N}.   
	\end{equation}
	Consequently, we bound $|I_3^{-1}P_1|_\infty$ as:
	$$|I_3^{-1}P_1|_\infty \leq |I_3^{-1}|_\infty\bigg|(NT)^{-1}\sum_{i=1}^N\sum_{t=2}^T(\widehat{\bm\Theta}_t-\bm\Theta_t^*)V_{it} \Delta\varepsilon_{it}\bigg|_\infty\lesssim_\P\max_{2\leq t\leq T}s_t^*\log m_t/N.$$
	
	Next, we proceed to bound $|I_3^{-1}(I_1+I_2)I_3^{-1}(P_1+P_2)|_\infty$. Observe that 
	\begin{eqnarray*}
	|I_3^{-1}(I_1+I_2)I_3^{-1}(P_1+P_2)|_\infty&\leq &|I_3^{-1}I_1I_3^{-1}P_1|_\infty + |I_3^{-1}I_2I_3^{-1}P_1|_\infty \notag\\
	&&+\, |I_3^{-1}I_1I_3^{-1}P_2|_\infty + |I_3^{-1}I_2I_3^{-1}P_2|_\infty.
	\end{eqnarray*}
	
	We first look at the rate of $|I_3^{-1}I_1I_3^{-1}P_1|_\infty$. By letting $\bm D_{t}\defeq\widehat{\bm\Theta}_t-\bm\Theta_t^*$, we have 
	$$I_3^{-1}I_1I_3^{-1}P_1=(NT)^{-2}\sum_{i=1}^N\sum_{t=2}^T\sum_{i'=1}^N\sum_{t'=2}^TI_3^{-1}\bm D_{t}V_{it}\Delta X_{it}^{\top}I_3^{-1}\bm D_{t'}V_{i't'}\Delta\varepsilon_{i't'}.$$
	In particular, for $k=1,\ldots,d$, the $k$th element of $I_3^{-1}I_1I_3^{-1}P_1$ is given by
	\begin{eqnarray*}
		&&(NT)^{-2}\sum_{i=1}^N\sum_{t=2}^T\sum_{i'=1}^N\sum_{t'=2}^T[I_3^{-1}]_{k\cdot}\bm D_{t}V_{it}\Delta X_{it}^{\top}I_3^{-1}\bm D_{t'}V_{i't'}\Delta\varepsilon_{i't'}\\
		&=&(NT)^{-2}\sum_{k'=1}^d\sum_{i=1}^N\sum_{t=2}^T\sum_{i'=1}^N\sum_{t'=2}^T[I_3^{-1}\bm D_{t}]_{k\cdot}[V_{it}\Delta X_{it}^\top]_{\cdot k'}[I_3^{-1}\bm D_{t'}]_{k'\cdot}V_{i't'}\Delta\varepsilon_{i't'}\\
		&=&(NT)^{-2}\sum_{k'=1}^d\sum_{i=1}^N\sum_{t=2}^T\sum_{i'=1}^N\sum_{t'=2}^T[\Delta X_{it}V_{it}^\top]_{k'\cdot}[I_3^{-1}\bm D_{t}]_{k\cdot}^\top[I_3^{-1}\bm D_{t'}]_{k'\cdot}V_{i't'}\Delta\varepsilon_{i't'}\\
		&=&(NT)^{-2}\sum_{k'=1}^d\sum_{i=1}^N\sum_{t=2}^T\sum_{i'=1}^N\sum_{t'=2}^T[\Delta X_{it}V_{it}^\top-\E(\Delta X_{it}V_{it}^\top)]_{k'\cdot}[I_3^{-1}\bm D_{t}]_{k\cdot}^\top[I_3^{-1}\bm D_{t'}]_{k'\cdot}V_{i't'}\Delta\varepsilon_{i't'}\\
		&&+\,\, (NT)^{-2}\sum_{k'=1}^d\sum_{i=1}^N\sum_{t=2}^T\sum_{i'=1}^N\sum_{t'=2}^T[\E(\Delta X_{it}V_{it}^\top)]_{k'\cdot}[I_3^{-1}\bm D_{t}]_{k\cdot}^\top[I_3^{-1}\bm D_{t'}]_{k'\cdot}V_{i't'}\Delta\varepsilon_{i't'}, 
	\end{eqnarray*}
	where $[\cdot]_{k\cdot}$ denotes the $k$th row and $[\cdot]_{\cdot k}$ denotes the $k$th column of the matrix, respectively.
	
	Consider the class of functions
	$$\mathcal A_{t,t'}\defeq\{A_{t,t'}=a_ta_{t'}^\top:a_t\in\Omega(c_0,s_t^*),a_{t'}\in\Omega(c_0,s_{t'}^*)\},$$
	with the entropy condition (with respect to the $d_\infty$-metric):
	$$\operatorname{ent}\big(\epsilon,\bigcup\nolimits_{2\leq t,t'\leq T}\mathcal A_{t,t'}\big)\lesssim T^2\big\{\max_{2\leq t\leq T}s_t^*\log(m_t/\epsilon)\big\}^2,\,\text{ for all }0<\epsilon<1.$$
	Applying Theorem \ref{label:Ubound} based on Assumptions \ref{a1}--\ref{a2}, we find that
	\begin{eqnarray*}
		&&\sup_{A_{t,t'}\in\mathcal A_{t,t'},\,t,t'=2,\ldots,T}\bigg|(NT)^{-2}\sum_{i=1}^N\sum_{t=2}^T\sum_{i'=1}^N\sum_{t'=2}^T[\Delta X_{it}V_{it}^\top-\E(\Delta X_{it}V_{it}^\top)]_{k'\cdot}A_{t,t'}V_{i't'}\Delta\varepsilon_{i't'}\bigg|\\
		&\lesssim_\P&(NT)^{-1}\log^3(NT).
	\end{eqnarray*}
	Moreover, combining \eqref{empbound} with the fact that
	$$\sup_{a_t\in\Omega_t(c_0,s_t^*),\,t=2,\ldots,T}\bigg|(NT)^{-1}\sum_{i=1}^N\sum_{t=2}^T[\E(\Delta X_{it}V_{it}^\top)]_{k'\cdot} a_t\bigg|\lesssim\max_{2\leq t\leq T}s_t^*\log m_t/\sqrt{N},$$
	we achieve:
	\begin{eqnarray*}
		&&\sup_{A_{t,t'}\in\mathcal A_{t,t'},\,t,t'=2,\ldots,T}\bigg|(NT)^{-2}\sum_{i=1}^N\sum_{t=2}^T\sum_{i'=1}^N\sum_{t'=2}^T[\E(\Delta X_{it}V_{it}^\top)]_{k'\cdot}A_{t,t'}V_{i't'}\Delta\varepsilon_{i't'}\bigg|\\
		&\lesssim_\P&\max_{2\leq t\leq T}(s_t^*\log m_t)^{3/2}/N.
	\end{eqnarray*}
	As a result, we bound $|I_3^{-1}I_1I_3^{-1}P_1|_\infty$ as 
	$$|I_3^{-1}I_1I_3^{-1}P_1|_\infty \lesssim_\P\max_{2\leq t\leq T}(s_t^*\log m_t)^{2}/N^{3/2}.$$
	
	Regarding the other higher-order terms, note that $|I_3^{-1}I_2|_{\infty}=\bigO_\P(1/\sqrt{NT})$ and $|I_3^{-1}P_2|_\infty=\bigO_\P(1/\sqrt{NT})$, which imply that $|I_3^{-1}I_2I_3^{-1}P_2|_\infty=\bigO_\P(1/(NT))$. Combining the bound for $|I_3^{-1}P_1|_\infty$ we have found above, we also deduce that $|I_3^{-1}I_2I_3^{-1}P_1|_\infty\lesssim_\P\max\limits_{2\leq t\leq T}s_t^*\log m_t/\sqrt{N^3T}$. Lastly, as for $|I_3^{-1}I_1|_{\infty}$, applying Theorem \ref{emplemma} gives that
	\begin{align*}
	&\sup_{\delta_t\in\Omega_t(c_0,s_t^*),\,t=2,\ldots,T}\bigg|(NT)^{-1}\sum_{i=1}^N\sum_{t=2}^T\delta_t^\top\{V_{it}\Delta X_{it}^{\top}-\E(V_{it}\Delta X_{it}^{\top})\}\bigg|_\infty\lesssim_\P\sqrt{\max_{2\leq t\leq T}s_t^*\log m_t}/\sqrt{N}.
	\end{align*}
	Combining this with the fact that
	\begin{align*}
	&\sup_{\delta_t\in\Omega_t(c_0,s_t^*),\,t=2,\ldots,T}\bigg|(NT)^{-1}\sum_{i=1}^N\sum_{t=2}^T\delta_t^\top\E(V_{it}\Delta X_{it}^{\top})\bigg|_\infty\lesssim_\P\max_{2\leq t\leq T}s_t^*\log m_t/\sqrt{N},
	\end{align*}
	we conclude that $|I_3^{-1}I_1|_{\infty}\lesssim_\P\max\limits_{2\leq t\leq T}s_t^*\log m_t/\sqrt{N}$, which implies that $|I_3^{-1}I_1I_3^{-1}P_2|_\infty\lesssim_\P\max\limits_{2\leq t\leq T}s_t^*\log m_t/(N\sqrt{T})$.
	
	Under the condition $\max\limits_{2\leq t\leq T}s_t^*\log m_t\sqrt{T}/\sqrt{N}\to0$ as $N,T\to\infty$, we have $\sqrt{NT}|I_3^{-1}P_1|_\infty=\smallO_\P(1)$, as well as $\sqrt{NT}|I_3^{-1}(I_1+I_2)I_3^{-1}(P_1+P_2)|_\infty=\smallO_\P(1)$. Thus, the proof is concluded.
\end{proof}

The asymptotic normality of AB-LASSO-SS under FD is analogous to that under FOD. In particular, the conclusion and required conditions for AB-LASSO-SS remain unchanged from those in Theorem \ref{main} and are therefore omitted.

% \begin{Theorem}[Asymptotic Normality of AB-LASSO-SS with First Differences]\label{main.fdss}
% 	Under Assumptions \ref{a1}--\ref{identification}, where the transformation $\Delta$ denotes first differencing, suppose that the asymptotic variance $\Omega$ is a positive definite matrix, and that $\max\limits_{2\leq t\leq T} \sqrt{s_t^*}\log m_t/\sqrt{N}\to0$ as $N,T\to\infty$. Then the AB-LASSO-SS estimator $\widehat{\theta}_{SS}$ is consistent for $\theta^0$, and
% 	$$\sqrt{NT}(\widehat{\theta}_{SS}- \theta^0 )\stackrel{\mathcal{L}}{\to} \operatorname{N}(0, \Omega).$$
% \end{Theorem}

% \begin{proof}
% 	Analogous to \eqref{expansion}, we derive expansions for $\widehat{\theta}_{A,B}-\theta^0$ and $\widehat{\theta}_{B,A}-\theta^0$, respectively. With sample-splitting, a significant difference arises in the convergence rate of $(NT)^{-1/2}\sum_{i\in\mathbb I_s}\sum_{t=2}^T(\widehat{\bm\Theta}_t-\bm\Theta_t^*)V_{it} \Delta\varepsilon_{it}$, where $s\in\{A,B\}$. As $\widehat{\bm\Theta}_t$ is obtained from a sub-sample uncorrelated with the one considered in the summation, we have 
% 	$$(NT)^{-1/2}\bigg|\sum_{i\in\mathbb I_s}\sum_{t=2}^T(\widehat{\bm\Theta}_t-\bm\Theta_t^*)V_{it} \Delta\varepsilon_{it}\bigg|_\infty\lesssim_\P\max_{2\leq t\leq T}\sqrt{s_t^*}\log m_t/\sqrt{N}.$$
% 	Thus, the required condition reduces to $\max\limits_{2\leq t\leq T}\sqrt{s_t^*}\log m_t/\sqrt{N}\to0$, as $N,T\to\infty$. The remainder of the proof follows that of Theorem \ref{main.fd} in a similar manner.
% \end{proof}

\renewcommand{\thesubsection}{B.\arabic{subsection}}
\setcounter{equation}{0}
\renewcommand{\theequation}{B.\arabic{equation}}
\setcounter{theorem}{0}
\renewcommand{\thetheorem}{B.\arabic{theorem}}
\setcounter{lemma}{0}
\renewcommand{\thelemma}{B.\arabic{lemma}}
\setcounter{figure}{0}
\renewcommand{\thefigure}{B.\arabic{figure}}
\setcounter{table}{0}
\renewcommand{\thetable}{B.\arabic{table}}
\setcounter{Remark}{0}
\renewcommand{\theRemark}{B.\arabic{Remark}}
\setcounter{corollary}{0}
\renewcommand{\thecorollary}{B.\arabic{corollary}}
\setcounter{Example}{0}
\renewcommand{\theExample}{B.\arabic{Example}}

\section{Supplementary Results for Simulation Study}\label{app.sim}
\subsection{Application-based Calibration}\label{app.cab}
We carry out an application-based calibration to validate our method in the context of the empirical study presented in Section \ref{app}. Specifically, we generate the response variable $Y_{it}$, for $i=1,\ldots,N=2,510$ and $t=2,\ldots,T=32$, using the model:
$$Y_{it} = \alpha_i + \gamma_t + \theta_0 D_{i,t-1} + \beta_1Y_{i,t-1} + \theta_1^\top C_{1i,t-1} + \theta_2 C_{2it}  + \varepsilon_{it},  \quad \varepsilon_{it}\stackrel{\operatorname{i.i.d.}}{\sim}\operatorname{N}(0,\sigma^2),$$
where the covariates $D_{it}$, $C_{1i,t-1}$, $C_{2it}$ and the initial value $Y_{i1}$ are taken directly from the real dataset. In the calibration, the values of the individual and time effects, as well as the coefficients, are set to the FE estimates. The sample standard deviation of the residuals estimated by FE is used to set $\sigma$.\footnote{We also consider a heteroskedastic case where the variance depends on the treatment variable $D_{i,t-1}$, and a dynamic model including four lags of the dependent variable. The results confirm the validity of our method in these settings and are available from the authors upon request.}

Table \ref{table:calibration} presents the calibration results for the coefficients of the treatment and policy variables using AB-LASSO, AB-LASSO-SS ($K=2$), FE, and DFE-A. We report the RMSE, SD, bias, and CI length as percentages of the absolute ``true" parameter values, along with empirical coverage, based on 100 replications. Inference is assessed at the nominal 95\% confidence level.\footnote{Here, AB and DAB are not included in the comparison because running AB for a single replication (with 3,375 moment conditions) already requires a total memory usage of approximately 2.3 GB, with a peak usage of 1.0 TB. For the real data application in Section \ref{app}, we used a server with 1.5 TB of memory and a time limit of 2 days, which would not be feasible on a standard PC.}

We find that AB-LASSO and AB-LASSO-SS consistently yield lower RMSEs than FE and DFE-A, primarily due to lower bias in most cases. Their empirical coverage rates are generally more accurate, with notable improvements for school opening, college visits, and mask mandates. In addition, our methods deliver more precise estimates, with narrower confidence intervals for all coefficients. Overall, sample splitting does not substantially alter the results of AB-LASSO, likely because FOD helps mitigate overfitting bias.

\newpage

\begin{table}[htbp!]
	\begin{center}	\caption{Calibration results for the effects on COVID-19 cases}\label{table:calibration}
		\begin{tabular}{p{1.5cm} cccc}
			\hline\hline
			&  \footnotesize{AB-LASSO} &  \footnotesize{AB-LASSO-SS} & \footnotesize{FE} & \footnotesize{DFE-A}\\
			& &  \footnotesize{($K=2$)} & & \\
			\cline{2-5}
			& \multicolumn{4}{c}{\small{K-12 school opening}}\\
			\hline
            \footnotesize{RMSE} & \bf{0.08} & 0.09 & 0.12 & 0.09 \\
            \footnotesize{SD} & 0.05 & 0.05 & 0.05 & 0.05 \\
            \footnotesize{bias} & \bf{-0.07} & \bf{-0.07} & 0.10 & \bf{0.07} \\
            \footnotesize{CI length} & 0.19 & 0.19 & 0.21 & 0.21 \\
            \footnotesize{coverage} & \bf{0.74} & 0.69 & 0.50 & 0.70 \\
			\hline
			& \multicolumn{4}{c}{\small{College visits}}\\
			\hline
            \footnotesize{RMSE} & \bf{0.08} & \bf{0.08} & 0.11 & 0.09 \\
            \footnotesize{SD} & \bf{0.06} & \bf{0.06} & 0.07 & 0.07 \\
            \footnotesize{bias} & \bf{-0.04} & \bf{-0.04} & 0.08 & 0.06 \\
            \footnotesize{CI length} & 0.29 & 0.29 & 0.33 & 0.33 \\
            \footnotesize{coverage} & 0.93 & \bf{0.94} & 0.87 & 0.92 \\
			\hline
			& \multicolumn{4}{c}{\small{Mask mandates}}\\
			\hline
            \footnotesize{RMSE} & \bf{0.06} & \bf{0.06} & 0.09 & 0.07 \\
            \footnotesize{SD} & \bf{0.05} & \bf{0.05} & 0.06 & 0.06 \\
            \footnotesize{bias} & \bf{0.03} & \bf{0.03} & -0.06 & -0.04 \\
            \footnotesize{CI length} & 0.21 & 0.21 & 0.24 & 0.24 \\
            \footnotesize{coverage} & \bf{0.90} & \bf{0.90} & 0.79 & 0.88 \\
            \hline
			& \multicolumn{4}{c}{\small{Stay-at-home orders}}\\
			\hline
            \footnotesize{RMSE} & \bf{0.20} & \bf{0.20} & 0.23 & 0.22 \\
            \footnotesize{SD} & \bf{0.20} & \bf{0.20} & 0.21 & 0.21 \\
            \footnotesize{bias} & \bf{0.00} & \bf{0.00} & -0.08 & -0.06 \\
            \footnotesize{CI length} & 0.80 & 0.80 & 0.87 & 0.87 \\
            \footnotesize{coverage} & 0.97 & 0.97 & 0.94 & \bf{0.96} \\
            \hline
			& \multicolumn{4}{c}{\small{Banning gatherings}}\\
			\hline
            \footnotesize{RMSE} & \bf{0.23} & \bf{0.23} & 0.25 & 0.24 \\
            \footnotesize{SD} & \bf{0.23} & \bf{0.23} & 0.25 & 0.24 \\
            \footnotesize{bias} & 0.03 & 0.03 & \bf{0.02} & \bf{0.02} \\
            \footnotesize{CI length} & 0.93 & 0.93 & 1.04 & 1.04 \\
            \footnotesize{coverage} & \bf{0.97} & \bf{0.97} & 0.98 & 0.98 \\
			\hline\hline
	   \multicolumn{5}{l}{\footnotesize{Notes: Superior results are indicated in bold.}}
		\end{tabular}
	\end{center}
\end{table}

\subsection{Supplementary Tables for Simulation Study}\label{app.supp}
\mbox{}

\small\begin{table}[htbp!]
	\begin{center} 	\caption{Results for $\theta_2=0.25$ with $N=100$, homoskedastic case}\label{table:n100t2.homo}
		\begin{tabular}{p{1.5cm} cccccccc}
			\hline\hline
			& \footnotesize{AB} &  \footnotesize{AB-LASSO} & \footnotesize{AB-LASSO-SS} &  \footnotesize{AB-LASSO-SS} & \footnotesize{DAB-SS} & \footnotesize{DFE-A} & \footnotesize{ML}\\
			& &  & \footnotesize{($K=2$)} &  \footnotesize{($K=5$)} &  \\
			\cline{2-8}
			& \multicolumn{7}{c}{\small{$T=20$}}\\
			\hline
            \footnotesize{RMSE} & 0.26 & 0.12 & 0.15 & 0.15 & 0.44 & 0.11 & \bf{0.10} \\
            \footnotesize{SD} & 0.25 & 0.10 & 0.14 & 0.14 & 0.44 & \bf{0.08} & 0.10 \\
            \footnotesize{bias} & -0.08 & -0.06 & -0.06 & -0.07 & -0.05 & 0.07 & \bf{0.00} \\
            \footnotesize{CI length} & 1.72 & 0.39 & 0.70 & 0.58 & 1.72 & 0.34 & 0.38 \\
            \footnotesize{coverage} & 1.00 & 0.91 & 0.99 & 0.97 & \bf{0.95} & 0.87 & 0.94 \\
			\hline
			& \multicolumn{7}{c}{\small{$T=30$}}\\
			\hline
			\footnotesize{RMSE} & 0.48 & \bf{0.08} & 0.10 & 0.09 & 0.79 & 0.09 & \bf{0.08} \\
            \footnotesize{SD} & 0.44 & \bf{0.07} & 0.09 & 0.09 & 0.75 & \bf{0.07} & 0.08 \\
            \footnotesize{bias} & -0.20 & -0.03 & -0.03 & -0.03 & -0.25 & 0.05 & \bf{0.00} \\
            \footnotesize{CI length} & 2.61 & 0.29 & 0.42 & 0.35 & 2.61 & 0.27 & 0.29 \\
            \footnotesize{coverage} & 1.00 & 0.93 & 0.98 & \bf{0.95} & 0.91 & 0.88 & 0.92 \\
			\hline
			& \multicolumn{7}{c}{\small{$T=40$}}\\
			\hline
			\footnotesize{RMSE} & 0.51 & \bf{0.06} & 0.07 & 0.07 & 0.86 & 0.07 & 0.07 \\
            \footnotesize{SD} & 0.47 & \bf{0.06} & 0.07 & 0.07 & 0.82 & \bf{0.06} & 0.07 \\
            \footnotesize{bias} & -0.20 & -0.02 & -0.02 & -0.01 & -0.25 & 0.04 & \bf{0.00} \\
            \footnotesize{CI length} & 2.78 & 0.24 & 0.31 & 0.27 & 2.78 & 0.23 & 0.24 \\
            \footnotesize{coverage} & 0.99 & 0.94 & 0.98 & \bf{0.95} & 0.89 & 0.89 & 0.91 \\
            \hline
			& \multicolumn{7}{c}{\small{$T=50$}}\\
			\hline
            \footnotesize{RMSE} & 0.55 & \bf{0.06} & \bf{0.06} & \bf{0.06} & 0.91 & \bf{0.06} &  \\
            \footnotesize{SD} & 0.54 & 0.06 & 0.06 & 0.06 & 0.90 & \bf{0.05} &  \\
            \footnotesize{bias} & -0.14 & \bf{-0.01} & \bf{-0.01} & \bf{-0.01} & -0.12 & 0.03 &  \\
            \footnotesize{CI length} & 2.85 & 0.21 & 0.26 & 0.23 & 2.85 & 0.20 &  \\
            \footnotesize{coverage} & 0.99 & 0.93 & \bf{0.96} & \bf{0.94} & 0.89 & 0.89 &  \\
            \hline
			& \multicolumn{7}{c}{\small{$T=60$}}\\
			\hline
            \footnotesize{RMSE} & 0.51 & \bf{0.05} & 0.06 & \bf{0.05} & 0.85 & \bf{0.05} &  \\
            \footnotesize{SD} & 0.50 & \bf{0.05} & \bf{0.05} & \bf{0.05} & 0.85 & \bf{0.05} &  \\
            \footnotesize{bias} & -0.12 & \bf{-0.01} & \bf{-0.01} & \bf{-0.01} & -0.04 & 0.02 &  \\
            \footnotesize{CI length} & 2.82 & 0.19 & 0.23 & 0.20 & 2.82 & 0.18 &  \\
            \footnotesize{coverage} & 1.00 & \bf{0.94} & \bf{0.96} & \bf{0.94} & 0.89 & 0.92 &  \\
			\hline\hline
    \multicolumn{8}{l}{\footnotesize{Notes: The numbers in the table are divided by 0.25 for RMSE, SD, bias, and CI length.}}\\
    \multicolumn{8}{l}{\footnotesize{Superior results are indicated in bold.}}
		\end{tabular}
	\end{center}
\end{table}

\small\begin{table}[htbp!]
	\begin{center} 	\caption{Results for $\theta_2=0.25$ with $N=200$, homoskedastic case}\label{table:n200t2.homo}
		\begin{tabular}{p{1.5cm} cccccccc}
			\hline\hline
			& \footnotesize{AB} &  \footnotesize{AB-LASSO} & \footnotesize{AB-LASSO-SS} &  \footnotesize{AB-LASSO-SS} & \footnotesize{DAB-SS} & \footnotesize{DFE-A} & \footnotesize{ML}\\
			& &  & \footnotesize{($K=2$)} &  \footnotesize{($K=5$)} &  \\
			\cline{2-8}
			& \multicolumn{7}{c}{\small{$T=20$}}\\
			\hline
            \footnotesize{RMSE} & 0.09 & 0.09 & 0.10 & 0.10 & 0.18 & 0.09 & \bf{0.07} \\
            \footnotesize{SD} & 0.09 & 0.07 & 0.08 & 0.08 & 0.18 & \bf{0.06} & 0.07 \\
            \footnotesize{bias} & -0.03 & -0.06 & -0.06 & -0.05 & 0.05 & 0.07 & \bf{0.00} \\
            \footnotesize{CI length} & 0.54 & 0.27 & 0.34 & 0.30 & 0.54 & 0.24 & 0.27 \\
            \footnotesize{coverage} & 1.00 & 0.86 & 0.93 & 0.90 & 0.85 & 0.79 & \bf{0.96} \\
			\hline
			& \multicolumn{7}{c}{\small{$T=30$}}\\
			\hline
			\footnotesize{RMSE} & 0.22 & 0.06 & 0.06 & 0.06 & 0.33 & 0.07 & \bf{0.05} \\
            \footnotesize{SD} & 0.17 & \bf{0.05} & \bf{0.05} & \bf{0.05} & 0.33 & \bf{0.05} & \bf{0.05} \\
            \footnotesize{bias} & -0.13 & -0.03 & -0.03 & -0.03 & -0.07 & 0.05 & \bf{0.00} \\
            \footnotesize{CI length} & 1.35 & 0.20 & 0.23 & 0.21 & 1.35 & 0.19 & 0.20 \\
            \footnotesize{coverage} & 1.00 & 0.92 & \bf{0.94} & \bf{0.94} & \bf{0.96} & 0.84 & \bf{0.96} \\
			\hline
			& \multicolumn{7}{c}{\small{$T=40$}}\\
			\hline
			\footnotesize{RMSE} & 0.39 & \bf{0.05} & \bf{0.05} & \bf{0.05} & 0.62 & 0.06 & \bf{0.05} \\
            \footnotesize{SD} & 0.32 & \bf{0.04} & \bf{0.04} & \bf{0.04} & 0.56 & \bf{0.04} & 0.05 \\
            \footnotesize{bias} & -0.21 & -0.02 & -0.01 & -0.02 & -0.28 & 0.04 & \bf{0.00} \\
            \footnotesize{CI length} & 1.90 & 0.17 & 0.18 & 0.18 & 1.90 & 0.16 & 0.17 \\
            \footnotesize{coverage} & 0.98 & 0.93 & \bf{0.95} & 0.94 & 0.88 & 0.84 & \bf{0.95} \\
            \hline
			& \multicolumn{7}{c}{\small{$T=50$}}\\
			\hline
            \footnotesize{RMSE} & 0.43 & \bf{0.04} & \bf{0.04} & \bf{0.04} & 0.71 & 0.05 & \bf{0.04} \\
            \footnotesize{SD} & 0.38 & \bf{0.04} & \bf{0.04} & \bf{0.04} & 0.65 & \bf{0.04} & \bf{0.04} \\
            \footnotesize{bias} & -0.20 & -0.01 & -0.01 & -0.01 & -0.28 & 0.03 & \bf{0.00} \\
            \footnotesize{CI length} & 2.14 & 0.15 & 0.16 & 0.15 & 2.14 & 0.14 & 0.15 \\
            \footnotesize{coverage} & 0.99 & 0.94 & 0.96 & \bf{0.95} & 0.85 & 0.84 & 0.94 \\
            \hline
			& \multicolumn{7}{c}{\small{$T=60$}}\\
			\hline
            \footnotesize{RMSE} & 0.45 & \bf{0.04} & \bf{0.04} & \bf{0.04} & 0.75 & \bf{0.04} & \bf{0.04} \\
            \footnotesize{SD} & 0.42 & 0.04 & 0.04 & 0.04 & 0.71 & \bf{0.03} & 0.04 \\
            \footnotesize{bias} & -0.18 & -0.01 & -0.01 & -0.01 & -0.22 & 0.03 & \bf{0.00} \\
            \footnotesize{CI length} & 2.20 & 0.13 & 0.14 & 0.14 & 2.20 & 0.13 & 0.13 \\
            \footnotesize{coverage} & 0.98 & 0.93 & \bf{0.95} & 0.94 & 0.86 & 0.84 & 0.92 \\
			\hline\hline
    \multicolumn{8}{l}{\footnotesize{Notes: The numbers in the table are divided by 0.25 for RMSE, SD, bias, and CI length.}}\\
    \multicolumn{8}{l}{\footnotesize{Superior results are indicated in bold.}}
		\end{tabular}
	\end{center}
\end{table}

\end{appendices}

\end{document}

%% file: isesv.tex
\newdimen\defpicwidth
\newdimen\defepswidth
\def\SKparam#1#2{}
\newcommand{\R}{\mathbb{R}}
\long\def\@makecaption#1#2{%
  \vskip\abovecaptionskip
  \sbox\@tempboxa{#1. #2}%
  \ifdim \wd\@tempboxa >\hsize
    {\centerline{\renewcommand{\baselinestretch}{0.8}%
\small\normalsize\parbox{0.8\textwidth}{#1. #2}}}\par
  \else
    \global \@minipagefalse
    \hbox to\hsize{\hfil\box\@tempboxa\hfil}%
  \fi
  \vskip\belowcaptionskip}
{\begin{longtable}{|p{1mm}p{3mm}p{0.8\textwidth}p{1mm}|}
\hline\hline
&&&\\[-4mm]
&& \hskip-20pt \centerline{\bf \summaryname}&\\[1mm]
\hline\hline\endfirsthead
\hline\hline
&&&\\[-4mm]
&& \hskip-20pt \centerline{\bf \summarycontname }&\\[1mm]
\hline\hline\endhead
\hline\hline\endfoot
}
{\end{longtable}}
\providecommand{\pdf}[1]{}

\newcommand{\dirxplore}[1]{../#1}
\newcommand{\dircodes}[1]{\wwwebook/codes/#1}
\def\wwwebook{http://www.quantlet.de}
\def\definitionname{DEFINITION}
\def\remarkname{REMARK}
\def\summarycontname{Summary (continued)}
\def\summaryname{Summary}
\def\bookname{}
\ifx\sfb\undefined
  \newcommand{\sfb}{}
  \fi

  \def\bookname{sfs}
\SKparam{BEGIN}{xpl ps pdf}
\SKparam{INPUT}{titlesv}
\SKparam{END}{}
\SKparam{BEGIN}{html java}
\SKparam{INPUT}{matitcs}
\SKparam{END}{}
\SKparam{LIB}{sfs}
\def\wwwebook{http://www.quantlet.com/mdstat}

\newcommand{\E}{\mathop{\mbox{\sf E}}}     % ISE

\SKparam{BEGIN}{ps pdf}
\SKparam{INPUT}{dedic}
\SKparam{INPUT}{pref}
\SKparam{END}{}

\renewcommand{\P}{\mathrm{P}}            % ISE

\def\defeq{\stackrel{\mathrm{def}}{=}}  % for definitions

\def\remarkname{REMARK}

\newtheorem{theorem}{THEOREM}[section]

\newtheorem{lemma}{LEMMA}[section]

%% file: newsimulation.tex
\section{Simulation Study}\label{sim}
We illustrate the finite sample properties of the proposed AB-LASSO and AB-LASSO-SS estimators through numerical simulations, where we also compare them  with other alternative methods. Following  \citet{moral2013likelihood} and \citet{moral2019dynamic}, we use the following data generating process  of \citet{bun2006effects}: 
\begin{align*}
Y_{it} &= \alpha_i + \theta_1 Y_{i,t-1} + \theta_2 D_{it} + \varepsilon_{it},\\
D_{it} &=\rho D_{i,t-1} + \phi Y_{i,t-1} + \pi\alpha _i + v_{it}, \quad i=1,\ldots,N, \quad t=1,\ldots,T,
\end{align*}
where $\alpha_i\stackrel{\operatorname{i.i.d.}}{\sim}\operatorname{N}(0,2.96)$. We consider 2 cases, one with conditional homoskedasticity and the other with conditional heteroskedasticity. In the homoskedastic case, $\varepsilon_{it}$ and $v_{it}$ are i.i.d.\ $t(4)$ (Student's $t$-distribution with four degrees of freedom), whereas in the heteroskedastic case, $\varepsilon_{it}=\{1+0.5*\mathbf{1}(v_{it}>0)\}e_{it}$, with $e_{it}$ and $v_{it}$ being i.i.d.\ $t(4)$. The sequences ${e_{it}}$, ${v_{it}}$, and ${\alpha_i}$ are mutually independent. Here, $D_{it}$ is predetermined with respect to $\varepsilon_{it}$ but not strictly exogenous, with the parameter $\phi$ capturing the feedback from $Y_{i,t-1}$ to $D_{it}$. In addition, the fixed effects are correlated with both $D_{it}$ and the feedback. We set $\theta_1=0.75$, $\theta_2=0.25$, $\rho=0.5$, $\phi=-0.17$, and $\pi=0.67$. While our methods do not rely on stationarity, we initialize the process with the same starting values of $Y$ and $D$ as \citet{moral2013likelihood} and \citet{moral2019dynamic}, which ensure mean stationarity. In Appendix \ref{app.cab}, we conduct another numerical simulation where the data generating process is calibrated to the empirical study in Section \ref{app}.

We assume that $Y_{it}$ is first observed at $t=1$ and use all the available lags of $Y_{it}$ and $D_{it}$ to construct the moment conditions for AB-LASSO and AB-LASSO-SS with FOD, that is, 
$$
\E(Z_{it} \Delta\widetilde{\varepsilon}_{it}) = 0, \quad Z_{it} = (Y_{i,t-1}, \ldots, Y_{i1}, D_{it},\ldots,D_{i1})^{\top}, \quad t = 2,\ldots,T-1,
$$
and $\E(D_{i1} \Delta\widetilde{\varepsilon}_{i1}) = 0. $
See Remark  \ref{Remark:ic} for a description of how the AB-LASSO and AB-LASSO-SS are modified when we do not observe the initial condition $Y_{i0}$. 
%We utilize all available lags $Y_{i,t-2},\ldots$, and $D_{i,t-1},\ldots$, to construct moment conditions for both AB and AB-LASSO. 
%It is important to note that in our setting, we allow for correlation between $\varepsilon_{it}$ and $v_{is}$, $s > t$, making it necessary to project $\Delta\widetilde D_{it}$ onto the IV space as well. 
The LASSO fitting is carried out using post-LASSO with penalty tuning parameters $\lambda_t$ that are independent of the design matrix. 
Specifically, for each $t=2,\ldots,T-1$, we set $\lambda_t=c\sqrt{N}\Phi^{-1}(1-0.1/(2m_t))$, where $c=1.1$ and $m_t$ denotes the number of instruments, i.e., the number of regressors in the LASSO regression for each time period $t$.\footnote{For the practical choice of the constant $c$, we recommend calibrating it on the  data at hand. In our application-based calibration (Appendix \ref{app.cab}), we find that $c=1.5$ performs well and adopt this value for the empirical analysis in Section \ref{app}, while in this simulation study we find that $c=1.1$ performs well.} 
To account for heteroskedasticity, the penalty weights associated with each regressor $V_{it,k}$, for $k=1,\ldots,m_t$, are defined as $\sqrt{\sum_{i=1}^N V_{it,k}^2\hat\eta_{it}^2}\big/\sqrt{N}$, where $\hat\eta_{it}$ are preliminary estimates of the error terms in the first step regressions. 
This approach is adopted for its conservatism to prevent overfitting \citep{lasso2018}. 
% {\red We set the penalty weight associated with the $j$-th regressor, $X_{i,j}$, to be the standard deviation of $X_{i,j}^2\widehat\vps_{i}^2$, where $\widehat\vps_{i}$ denotes the residuals obtained by a preliminary regression.} 
%We set the penalty weights equal to one, i.e., $\omega_{ts}=1$ for all $t,s$. 
%SS stands for using sample splitting and cross fitting to estimate $\theta$ under AB-LASSO.  
For the AB-LASSO-SS estimator, we implement a $K$-fold sample-splitting and cross-fitting procedure with $K \in \{2,5\}.$  According to the formula for the asymptotic variance estimation shown in \eqref{var.est}, we calculate the standard error using the residuals based on the final estimate obtained after aggregation. 
% as follows: the entire sample is randomly split into $K$ parts alone the individual dimension. For each $k=1,\ldots,K$, the $k$-th fold serves as the main sample, while the remaining observations form the auxiliary sample. We perform the LASSO step using the auxiliary sample and estimate $\theta$ using the main sample. The $K$ estimates are then aggregated by taking the mean to produce the final estimate on $\theta$. We calculate the standard error using the residuals based on this final estimate. 
We also repeat the estimation 100 times using random sample splits and aggregate the results through the medians. 

In addition to the conventional two-step AB, we compare the results with several alternatives. First, the debiased AB of \citet{chen2019mastering} using one random split (i.e. 2 folds) along the cross-section dimension. We refer to this estimator as DAB-SS.  Second, the analytically debiased fixed effects estimator (DFE-A), which corrects the bias of the FE estimator arising from the incidental parameter problem \citep{nickell1980correcting,kiviet1995bias,hahn2002asymptotically}. Third, the likelihood-based estimator \citep{moral2013likelihood,moral2019dynamic}. For all the  comparison estimators, we use analytical standard error clustered at the individual level. We do not consider the half-panel jackknife bias corrected estimator of \citet{chudik2018half} because it relies on unconditional stationarity of all the variables, which is restrictive for many applications. For example, most of the covariates in our empirical application of Section \ref{app} are staggered policy indicators that are not stationary over time. 

For each estimator, we report the root mean square error (RMSE), standard deviation (SD) and bias in percent of the true parameter value, together with the length and empirical coverage of confidence intervals (CI) with a nominal confidence level of $95\%$. All of them are computed using 500 simulations. Tables \ref{table:n100t2.hetero} and \ref{table:n200t2.hetero}  display the results for the treatment coefficient $\theta_2$, which is typically of primary interest, for $T \in \{20; 30; 40; 50; 60\}$ in the heteroskedastic case for $N=100$ and $N=200$, respectively. Similar results for the homoskedastic case can be found in Appendix \ref{app.supp}.  These sample sizes lead to   $m \in \{360; 840; 1,520; 2,400; 3,480\}$ moment conditions for AB, which  are large relative to the corresponding sample sizes $n = NT \in \{2,000 ; 3,000; 4,000; 5,000; 6,000\}$ when $N=100$ and $n = NT \in \{4,000 ; 6,000; 8,000; 10,000; 12,000\}$ when $N=200$.  Thus, the resulting orders of the small bias condition are $m^2/n \in \{64.8; 235.2; 577.6; 1,152; 2,018.4\}$ when $N=100$ and $m^2/n \in \{32.4; 117.6; 288.8; 576; 1,009.2\}$ when $N=200$, which are not negligible.  

We find that AB-LASSO consistently outperforms AB and DAB across all evaluation criteria for all sample sizes. It is evident that using LASSO to select the most relevant moments  significantly reduces bias and results in more accurate coverage rates than AB. The CI length of our methods is generally shorter than that for AB, suggesting that the bias reduction does not come at the expense of more dispersion.  Comparing with DFE-A, our methods have similar SD, RMSE and CI length, but  smaller absolute bias. Moreover, AB-LASSO and AB-LASSO-SS display coverages closer to the nominal level than DFE-A. The undercoverage of DFE-A is due to the bias being comparable to the standard deviation.

When the panel is relatively short, %the performance of ML is comparable to AB-LASSO for $\theta_2$ and even superior for the autoregressive coefficient $\theta_1$.
the performance of AB-LASSO in terms of RMSE and CI coverage is comparable to that of ML.  Although ML attains a lower RMSE when $T=20$, this advantage vanishes for the other values of $T$ considered.  We acknowledge that selecting lags through LASSO might not be essential when the number of instruments is modest (e.g., when $T=20$). However, ML is only feasible when the number of observations is  large relative to the number of moment conditions,  since the sample covariance matrix is otherwise not positive definite. Specifically, with $N=100$, the maximum feasible value of $T$ is $50$. Moreover, for $N=100$ and $T=50$, the likelihood-based model fails to converge in roughly two-thirds of the simulations.\footnote{We implement the maximum likelihood estimator on dynamic panel models using the \texttt{dpm} function in the \textsf{R} package \texttt{dpm} \citep{dpm}.} In such cases, parameter estimates are available (as initial values are returned), but standard errors are not, as measures of fit are not reported. Consequently, we do not report ML results for $T\in\{50,60\}$ when $N=100$. Our methods remain robust as $T$ increases, demonstrating the practical suitability for longer panels when $N$ is moderate. In practice, long panels with $T>50$ and moderate $N$ arise in applications such as modeling GDP growth rates for OECD (Organization for Economic Co-operation and Development) countries. For example,  \citet{galvao2011quantile} used annual data from 1948 to 2008 on 18 countries, that is $N=18$ and $T=61$, where our methods remain feasible but ML becomes impractical. 
In general, relatively long panels are common in macroeconomic applications that use annual or quarterly data on a set of countries or regions (e.g., \citealp{garcia1987macroeconomic,hoogstrate2000pooling}).

%We also find that sample splitting helps reduce the overfitting bias of AB-LASSO for $\theta_1$ to some extent, though the effect is less pronounced for $\theta_2$, 
In this case sample splitting does not further improve the performance of AB-LASSO significantly, likely because FOD already mitigates much of the overfitting bias and there are only two right-hand-side observed variables in the equation for $Y_{it}$. Overall, the results for AB-LASSO-SS remain robust across different choices of $K$ in terms of estimation, while using $K=5$ split folds improves the inference precision in some cases by producing narrower CI  with less conservative coverage rates. %For DAB-SS, some improvement over AB is observed for the autoregressive coefficient $\theta_1$, but not for the treatment coefficient $\theta_2$. 
Comparing the results for $N=200$ and $T=20$ in Table \ref{table:n200t2.hetero} with those for $N=100$ and $T=40$ in Table \ref{table:n100t2.hetero}, we find that AB-LASSO and AB-LASSO-SS exhibit superior performance in the latter case, showing lower bias and more accurate CI coverage. This indicates their robustness when the time dimension $T$ is large relative to $N$, while the sample size $n=NT$ remains fixed.  %With a relatively large $N$, the bias in AB estimation is partially mitigated for short panels with $T=20$. However, as $T$ increases, AB continues to exhibit substantial bias and much wider confidence intervals, which affects the credibility of the coverage rate for inference purposes. 
Finally, we give an example comparing the complexity of AB and AB-LASSO in terms of computer memory usage. Implementing two-step AB with an efficient weighting matrix under heteroskedasticity for $N=200$ and $T=30$ using the \texttt{pgmm} function in the \textsf{R} package \texttt{plm} \citep{plm} on a single sample (with $840$ moment conditions) consumes approximately $1.7$ GB of peak memory, whereas AB-LASSO-SS with a single 5-fold partition requires only $45$ MB.\footnote{We track the RAM usage in \textsf{R} using the package \texttt{peakRAM} \citep{peakRAM}.}